\documentclass[a4paper,USenglish,cleveref, autoref, thm-restate,numberwithinsect]{lipics-v2021}

\nolinenumbers

\title{The theory of concatenation over finite models}

\author{Dominik D. Freydenberger}{Loughborough University, Loughborough, United Kingdom}{}{0000-0001-5088-0067}{Supported by EPSRC grant EP/T033762/1.}
\author{Liat Peterfreund}{DI ENS, ENS Paris, CNRS, PSL University, INRIA, Paris, France}{}{}{Supported by Fondation des Sciences Math\'{e}matiques de Paris (FSMP). A part of this work was done while affiliated with IRIF, CNRS, Universit\'{e} de Paris and with Edinburgh University.}

\authorrunning{D.\,D. Freydenberger and L. Peterfreund} 
\Copyright{Dominik D. Freydenberger and Liat Peterfreund} 



\relatedversiondetails[cite=thisPaper]{Conference Version}{} 

\acknowledgements{	The authors thank Joel D.\ Day, Manfred Kufleitner, Leonid Libkin, Sam M.\ Thompson, and the reviewers of the current and  previous versions  for their helpful comments.}

\hideLIPIcs  %

\EventEditors{Nikhil Bansal, Emanuela Merelli, and James Worrell}
\EventNoEds{3}
\EventLongTitle{48th International Colloquium on Automata, Languages, and Programming (ICALP 2021)}
\EventShortTitle{ICALP 2021}
\EventAcronym{ICALP}
\EventYear{2021}
\EventDate{July 12--16, 2021}
\EventLocation{Glasgow, Scotland (Virtual Conference)}
\EventLogo{}
\SeriesVolume{198}
\ArticleNo{129}
\usepackage{xspace}
\usepackage{tikz}
\usepackage{booktabs}
\usepackage{stmaryrd}
\usepackage{mathtools}

\newcommand{\bind}[2]{#1\{#2\}}
\newcommand{\wroot}[1]{\varrho(#1)}
\newcommand{\fovop}[1]{#1^{o}}
\newcommand{\fovcl}[1]{#1^{c}}
\newcommand{\conc}{\,}
\newcommand{\weqeq}{\mathbin{\dot{=}}}
\newcommand{\foeq}{\mathbin{\dot{=}}}
\newcommand{\lettPred}[1]{\mathsf{P}_{#1}}
\newcommand{\nextPred}{\mathsf{succ}}

\newcommand{\strEqPred}{\mathsf{Eq}}
\newcommand{\strEqPredSet}{\mathsf{EQ}}

\newcommand{\minConst}{\mathsf{min}}
\newcommand{\maxConst}{\mathsf{max}}
\newcommand{\constr}[2]{#1\mathbin{\dot{\in}}#2}

\newcommand{\formulaWidth}{\mathsf{wd}}
\newcommand{\treeWidth}[1]{\mathsf{tw}(#1)}
\newcommand{\struct}[1]{\mathcal{A}'_{#1}}
\newcommand{\structClassic}[1]{\mathcal{A}_{#1}}
\newcommand{\compl}[1]{\overline{#1}}

\newcommand{\ar}{\mathsf{ar}}

\newcommand{\VALC}{\mathsf{VALC}}
\newcommand{\INVALC}{\mathsf{INVALC}}

\newcommand{\fun}[1]{\llbracket #1 \rrbracket}

\newcommand{\df}{:=}
\newcommand{\logC}{\ensuremath{\mathsf{C}}}

\newcommand{\logCcon}[1]{\logC[#1]}
\newcommand{\logCreg}{\logCcon{\REG}}
\newcommand{\REG}{\mathsf{REG}}

\newcommand{\FC}{\ensuremath{\mathsf{FC}}}
 
\newcommand{\FCcon}[1]{\FC[#1]}

\newcommand{\FCreg}{\ensuremath{\FCcon{\REG}}}

\newcommand{\FCtxt}{\texorpdfstring{\textsf{FC}\xspace}{FC}}
\newcommand{\FCregtxt}{\texorpdfstring{\textsf{FC$[$REG$]$}\xspace}{FC[REG]}}
\newcommand{\FOstrtxt}{\texorpdfstring{\textsf{FO[Eq]}\xspace}{FO[Eq]}}

\newcommand{\FOord}{\ensuremath{\FO[<]}}
\newcommand{\FO}{\mathsf{FO}}
\newcommand{\MSO}{\mathsf{MSO}}

\newcommand{\FOstrcon}[1]{\FO[\strEqPredSet, #1]}
\newcommand{\FOstr}{\FO[\strEqPredSet]}

\newcommand{\FOstrReg}{\FOstrcon{\REG}}
\newcommand{\FOordtxt}{\texorpdfstring{\textsf{FO}$[\textsf{<}]$\xspace}{FO[<]\xspace}}

\newcommand{\Datalogtxt}{\textsf{Datalog}\xspace}

\newcommand{\DataSplog}{\mathsf{FC}\mhyphen\mathsf{Datalog}}
\newcommand{\DataSplogtxt}{\textsf{FC-Datalog}\xspace}

\newcommand{\ruleset}{\Phi}
\newcommand{\rulesymbs}{\mathcal{R}}
\newcommand{\splog}{\ifmmode{\mathsf{SpLog}}\else{\textsf{SpLog}}\xspace\fi}
\newcommand{\splogneg}{\mathsf{SpLog}^{\neg}}
\newcommand{\outDatalog}{\textsf{Ans}}
\mathchardef\mhyphen="2D
\newcommand{\fragEx}[1]{\mathsf{E}\mhyphen#1}
\newcommand{\fragExPos}[1]{\mathsf{EP}\mhyphen#1}

\newcommand{\ta}{\mathtt{a}}
\newcommand{\tb}{\mathtt{b}}
\newcommand{\pow}{\mathcal{P}}

\newcommand{\EF}{Ehrenfeucht–Fra\"{i}ss\'{e}\xspace}

\newcommand{\join}{\bowtie}
\newcommand{\sel}{\zeta^=}
\newcommand{\proj}{\pi}
\newcommand{\emptyword}{\varepsilon}
\newcommand{\ror}{\cup} 
  
\newcommand{\compclass}[1]{\mathsf{#1}}
\newcommand{\cc}{\mathbb{C}}
\newcommand{\setdiff}{-}
\newcommand{\patgraph}[1]{\mathcal{G}_{#1}}

\newcommand{\bigland}{\bigwedge}
\newcommand{\biglor}{\bigvee}

\DeclareMathOperator{\subword}{\sqsubseteq}
\DeclareMathOperator{\factor}{\sqsubseteq}
\DeclareMathOperator{\ppref}{\sqsubset_{\mathsf{p}}}
\DeclareMathOperator{\pref}{\sqsubseteq_{\mathsf{p}}}
\DeclareMathOperator{\supword}{\sqsupseteq}

\newcommand{\Lang}{\mathcal{L}} 
 
\newcommand{\Sub}{\mathsf{Fac}} 

\newcommand{\avs}[1]{(#1)}

\newcommand{\polyeq}{\mathbin{\equiv_{\mathsf{poly}}}}
 
\newcommand{\quot}[1]{``#1''}
\newcommand{\spn}[1]{[#1\rangle}
\newcommand{\SVars}[1]{\mathsf{Var}(#1)}

\newcommand{\fvar}{\mathsf{free}}

\newcommand{\ex}[2]{\exists#1\colon#2}
\newcommand{\fa}[2]{\forall#1\colon#2}
\newcommand{\substsubst}[3]{#1_{#2\mapsto#3}}
\newcommand{\var}{\mathsf{Var}}
\newcommand{\dom}{\mathsf{Dom}}
\newcommand{\mv}{\mathsf{w}}
\newcommand{\sv}{\mathfrak{u}}

\newcommand{\relsym}{\dot{R}}

\newcommand{\ie}{i.\,e.\xspace}
\newcommand{\eg}{e.\,g.\xspace}

\newcommand{\DTC}[1]{#1^{\mathsf{dtc}}}
\newcommand{\TC}[1]{#1^{\mathsf{tc}}}
\newcommand{\logIter}[5]{[#1\: #2, #3\colon #4](#5)}
\newcommand{\tcOp}[3]{[\tc\: #1,#2\colon #3]}
\newcommand{\dtcOp}[3]{[\dtc\: #1,#2\colon #3]}
\newcommand{\LFP}[1]{#1^{\mathsf{lfp}}}
\newcommand{\LFPcon}[2]{#1^{\mathsf{lfp}}[#2]}

\newcommand{\PFP}[1]{#1^{\mathsf{pfp}}}
\newcommand{\dtc}{\mathsf{dtc}}
\newcommand{\tc}{\mathsf{tc}}
\newcommand{\lfp}{\mathsf{lfp}}
\newcommand{\pfp}{\mathsf{pfp}}
\newcommand{\core}{\mathsf{core}}
\newcommand{\cored}{\mathsf{gcore}}
\newcommand{\RGX}{\mathsf{RGX}}
\newcommand{\RGXcore}{\RGX^{\core}}
\newcommand{\RGXcored}{\RGX^{\cored}}

\let\phi\varphi

\theoremstyle{plain}

\newtheorem*{fvtheorem}{Feferman-Vaught theorem}

\begin{document}
	\maketitle
	\begin{abstract}
We propose $\FC$, a new logic on words that combines finite model theory with the theory of concatenation -- a first-order logic that is based on word equations.
Like the theory of concatenation,  $\FC$ is built around word equations; in contrast to it, its semantics are defined to only allow finite models, by limiting the universe to a word and all its factors. As a consequence of this, $\FC$ has many of the desirable properties of $\FO$ on finite models,  while being far more expressive than $\FOord$. Most noteworthy among these desirable properties are sufficient criteria for efficient model checking, and capturing various complexity classes by adding operators for transitive closures or fixed points.	

Not only does $\FC$ allow us to obtain new insights and techniques for expressive power and efficient evaluation of document spanners, but it also provides a general framework for logic on words that also has potential applications  in other areas.   	
\end{abstract}		
\section{Introduction}\label{sec:intro}
This paper proposes a finite model version of the theory of concatenation: $\FC$, a new logic that is designed to describe properties of words and to query them.
While the idea of using logic on words is by no means new, the advantage of $\FC$ is its combination of expressive power and tractable model checking and evaluation.

\subparagraph*{Logic on words}
A  common way of using logic on words is \emph{monadic second-order logic} ($\MSO$)  over a linear order (\eg \cite{str:fin}). That is, a word $w$ is seen as a sequence of positions, and predicates $\lettPred{a}(x)$ express ``letter $a$ at position $x$ of $w$''.
This approach comes with two disadvantages for querying.
The first is that factors (continuous subwords)  cannot be expressed directly. Consider the query ``return all factors of $w$''.
As variables refer to positions, the query would not return a factor $u$ directly, but represent it as a set (or tuple) of positions that describe a specific occurrence of $u$ in $w$.
If $u$ occurs more than once, the query result would contain each occurrence -- unless the logic is powerful enough to prevent this.

This leads us to the second disadvantage, namely that $\MSO$ cannot compare factors of unbounded length.
That is, while $\MSO$ can express queries like ``return the positions of factors of length $k$ that occur twice in $w$'' for a fixed length $k$, it is impossible to express ``return the positions of factors that occur twice in $w$''; or non-regular languages, like that of all words $ww$ with $w\in\{\ta,\tb\}^*$. 
As a result,  many natural relations on factors of $w$ are inexpressible in $\MSO$, in particular the concatenation $x=y z$.

Another approach to logic on words is the \emph{theory of concatenation} (short: $\logC$). First defined by  Quine~\cite{qui:con}, this logic is built on \emph{word equations}, that is, equations of the form $xx\weqeq yyy$, where variables like $x$ and $y$ stand for words over a finite alphabet~$\Sigma$. 
While less prominent than $\FO$ or $\MSO$, the theory of concatenation has been studied extensively since the 1970s, with particular emphasis on word equations. A fairly recent survey on the satisfiability of word equations is~\cite{die:mor}. More current research on   word equations and the theory of concatenation   can be found in \eg~\cite{cio:sol, DBLP:conf/rp/DayGHMN18,  DBLP:conf/icalp/DayM20, DBLP:conf/mfcs/DayMN17, DBLP:conf/icalp/NowotkaS18,DBLP:conf/icalp/Saarela20}.

In contrast to $\MSO$, the logic $\logC$ allows us to treat words as words (instead of intervals of positions) and a position in $w$ can be expressed as the corresponding prefix of $w$.
More importantly, $\logC$ can express properties like ``$u$ is a factor of $v$'' or concatenation like $x=yz$.
This expressive power comes at a price -- even limited use of  negation leads to an undecidable theory (\ie, satisfiability is undecidable, see~\cite{dur:und,qui:con}). 
Contrast this to \emph{first-order logic} ($\FO$) over finite models: By Trakhtenbrot's theorem, satisfiability is undecidable; but the model checking problem is not just decidable, but can even be made tractable (see \eg~\cite{ebb:fin,lib:ele}).

Another situation where using queries for words together with an open universe causes problems occurs in \emph{string databases}, see~\cite{DBLP:journals/jacm/BenediktLSS03, DBLP:journals/jcss/BonnerM98, DBLP:journals/jcss/GinsburgW98, DBLP:journals/jcss/GrahneNU99}.
These query languages treat words as entries of the database instead of  operating on a single word. 
Furthermore, they offer transformation operations that assume an infinite universe. As a result, these  query language usually express Turing-complete functions from words to words.

\subparagraph{Introducing \FCtxt}
The new logic $\FC$ aims to bring the advantages of $\FO$ on finite models to the theory of concatenation. The universe for $\logC$ is usually assumed to be $\Sigma^*$, which means that there is no meaningful distinction between satisfiability and model checking.
The key idea of $\FC$ is  changing  universe from $\Sigma^*$ to the set of all factors of a word $w$; comparable to how the universe for $\MSO$ consists of all positions of a word $w$. 

As $\FC$-formulas are based on word equations, concatenation is straightforward to use. For example,
 ``return all factors that occur   twice in $w$'' can be expressed  as
\[\phi_1(x)\df \ex{p_1,p_2,s_1,s_2}{\bigl(\sv\weqeq p_1\conc x\conc s_1 \land \sv\weqeq p_2\conc x\conc s_2 \land\neg p_1\weqeq p_2\bigr)},\]
where $\sv$ represents $w$. In detail, $\phi_1$ expresses that there are two different ways of decomposing $w$ into $w=p\conc x\conc s$. If we also wanted to know the positions of these occurrences, we could return $p_1$ and $p_2$ (by removing their quantifiers), as these encode the start of each occurrence in $w$. This formula does not rely on the requirement that 
variables can only be mapped to factors of $w$, as the $\sv$ on the left side of the equations ensures this already. Instead, consider
\[\phi_2(x)\df \ex{y,z}{y\weqeq x\conc z\conc x},\]
which returns all factors $x$ that have two non-overlapping occurrences in $w$.
As we not need to know where in $w$ the factor $xzx$ occurs, $\sv$ is not needed in the formula.

The restriction to a finite universe allows us to translate various classical results from $\FO$ to $\FC$. Most importantly, model checking becomes not only decidable, but upper bounds can be lowered in the same way as for $\FO$ (Section~\ref{sec:mc}). 
In fact, $\FC$ can be extended with iteration operators to characterize complexity classes from $\compclass{L}$ to $\compclass{PSPACE}$, 
analogously to $\FO$  on ordered structures; which allows us to define a version of \Datalogtxt on words that includes concatenation (Section~\ref{sec:iterationRecursion}). 
Furthermore, Section~\ref{sec:logicForSpanners} also describes how $\FC$ can be extended with \emph{constraints} (aka \emph{predicates}), while still keeping model checking tractable. 
\subparagraph*{Spanners} An immediate application of $\FC$  is as a logic for  
\emph{document spanners} (or just \emph{spanners}). Spanners are  a rule-based framework for information extraction that was proposed by 
Fagin, Kimelfeld, Reiss, and Vansummeren~\cite{fag:spa} to study the formal properties of the query language AQL of IBM's SystemT for information extraction. 
They can be understood as a combination of regular expressions and relational algebra.

In the last years, spanners have received considerable attention in the database theory community. 
The two main areas of interest are 
expressive power~\cite{fag:spa, fre:splog, fre:doc,  mat:doc, nah:inc,pet:gra, pet:rec,  DBLP:conf/icdt/SchmidS21} and
efficient evaluation~\cite{ama:con, flo:con, fre:joi,  mat:doc, mor:eng, pet:gra, pet:com, DBLP:conf/icdt/SchmidS21};
further topics include
updates~\cite{ama:con,fre:dyn,los:fou}, 
cleaning~\cite{fag:dec}, 
distributed query planning~\cite{dol:split}, 
and a weighted variant~\cite{dol:wei}.

But most of these articles do not focus on the full class of spanners that was introduced by Fagin et al. (called \emph{core spanners}, as they describe the core of AQL), but a much smaller subclass, the \emph{regular spanners}. 
The difference between these is that regular spanners cannot express equality of factors. Hence, techniques for finite automata and $\MSO$  often work on regular spanners; but  they rarely work for core spanners.
Furthermore, although spanners are conceptually similar to relational algebra, many canonical approaches for relational databases and the underlying $\FO$ are not viable in the spanner setting. 
In particular, while \emph{acyclic conjuctive queries} are well-known to be tractable for $\FO$ (see \eg~\cite{abi:fou}), this does not hold for the corresponding class of spanners (see~\cite{fre:joi}).

Although ``pure'' $\FC$ is not powerful enough to express core spanners, extending it with constraints that decide regular languages results in a logic that captures core spanners (Section~\ref{sec:FCandSpanners}). In addition to providing us with a rich and natural class of tractable spanners, this connection also 
allows us to develop a new inexpressibility method (Section~\ref{sec:inexpressibility}).

\section{Preliminaries}\label{sec:prelim}
Let $\emptyword$ denote the \emph{empty word}.
We use $|x|$ for the length of a word, a formula, or a regular
expression $x$, or the number of elements of a finite set $x$.
A word $v$ is a \emph{factor} of a word $w$, written $v\factor w$, if there exists (possibly empty) words $p,s$ with $w=p v s$.
For words $x$ and $y$, let $x\pref y$  ($x$ is a \emph{prefix} of $y$) if $y=xs$ for some $s$, 
and $x\ppref y$ if $x\pref y$ and $x\neq y$.

For alphabets $A,B$, a \emph{morphism} is a function $h\colon A^*\to B^*$ with $h(u\cdot v)=h(u)\cdot h(v)$ for all $u,v\in A^*$. To define $h$, it suffices to define $h(a)$ for all $a\in A$.
Let $\Sigma$ be a finite \emph{terminal alphabet}, and let $\Xi$ be an infinite \emph{variable alphabet} with  $\Sigma\cap\Xi=\emptyset$. We assume $\Sigma$ is fixed and~$|\Sigma|\geq 2$, unless stated otherwise. As a convention, we use typewriter letters (like $\ta$ and~$\tb$) for terminals.
\subparagraph{Patterns and the theory of concatenation}
A \emph{pattern} is a word from $(\Sigma\cup\Xi)^*$. For every pattern $\eta\in(\Sigma\cup\Xi)^*$, let $\var(\eta)$ denote the set of variables that occur in $\eta$.
A \emph{pattern substitution} (or just \emph{substitution}) is a partial morphism $\sigma\colon (\Sigma\cup\Xi)^*\to\Sigma^*$ with $\sigma(\ta)=\ta$ for all $\ta\in\Sigma$. When applying a substitution $\sigma$  to a pattern $\eta$, 
we assume $\sigma$ is defined on $\var(\eta)$, that is, $\dom(\sigma)\supseteq \var(\eta)$.
A \emph{word equation} is a pair of patterns, that is, a pair  $(\eta_L,\eta_R)$ with $\eta_L,\eta_R\in(\Sigma\cup\Xi)^*$. We also write $\eta_L \weqeq \eta_R$, and call~$\eta_L$ and~$\eta_R$  the \emph{left side} and the \emph{right side} of the  equation. 
A \emph{solution} of $\eta_L \weqeq \eta_R$ is a substitution $\sigma$ with $\sigma(\eta_L)=\sigma(\eta_R)$. 

The theory of concatenation  combines word equations with first-order logic. First the syntax: 
	The set~$\logC$ of formulas of the \emph{theory of concatenation} uses 
	word equations $(\eta_L \weqeq \eta_R)$ with $\eta_L,\eta_R\in(\Sigma\cup\Xi)^*$ as atoms. 
	The connectives  are  conjunction, disjunction, negation, and quantifiers with variables from $\Xi$.
	For every $\phi\in\logC$, we define its set of \emph{free variables} $\fvar(\phi)$ by $\fvar(\eta_L \weqeq \eta_R)\df \var(\eta_L)\cup\var(\eta_R)$; extending this canonically.

The semantics build on  solutions of word equations:
	For all $\phi\in\logC$ and all pattern substitutions $\sigma$ with $\dom(\sigma)\supseteq\fvar(\phi)$, we define $\sigma\models\phi$ as follows:
Let	$\sigma\models(\eta_L \weqeq \eta_R)$ if $\sigma(\eta_L) = \sigma(\eta_R)$.  
For the quantifiers, we say $\sigma\models \ex{x}{\phi}$ (or $\sigma\models \fa{x}{\phi}$) if $\substsubst{\sigma}{x}{w}\models\phi$ holds for an (or all) $w\in\Sigma^*$, where  $\substsubst{\sigma}{x}{w}$ is defined by $\substsubst{\sigma}{x}{w}(x)\df w$ and $\substsubst{\sigma}{x}{w}(y)\df \sigma(y)$ for all $y\in (\Sigma\cup\Xi)\setdiff\{x\}$. 
The connectives' semantics are defined canonically. 
\begin{example}
 Let $\phi \df x\mathtt{abc} y  \weqeq y \mathtt{bca} x \land \neg(x\weqeq \emptyword \lor y\weqeq \emptyword)$. Then $\sigma\models \phi$ if and only if $\sigma(x\mathtt{abc} y)=\sigma(y \mathtt{bca} x)$, $\sigma(x)\neq \emptyword$, and $\sigma(y)\neq \emptyword$. For example, if $\sigma(x)=\mathtt{abca}$ and $\sigma(y)=\ta$.
\end{example}
We freely add and omit parentheses as long as the meaning stays clear.
$\fragEx{\logC}$, the \emph{existential fragment} of $\logC$, consists of those formulas that do not use universal quantifiers and that apply negation only to word equations. The \emph{existential-positive fragment} $\fragExPos{\logC}$ allows neither universal quantifiers, nor negation. We also use this notation for other logics that we define.
\section{Finite models in the theory of concatenation}\label{sec:theLogic}
The new logic \emph{finite model version of the theory of concatenation}, namely $\FC$,
is built around word equations; similarly to the theory of concatenation $\logC$.
The latter can be understood as first-order logic over the universe $\Sigma^*$ with  concatenation  -- see for example~\cite{DBLP:conf/lics/HalfonSZ17}, which refers to~$\logC$ as $\mathsf{FO}(A^*,\cdot)$. 
In other words, for $\logC$, we can consider the universe to be fixed (for a given terminal alphabet $\Sigma$).
The key idea of $\FC$ is to replace the universe $\Sigma^*$ with a single word and all its factors. In the formulas, this  word is represented by  a distinguished variable:
\begin{definition}
	We distinguish a variable ${\sv}\in\Xi$ and call it the \emph{universe variable}.
\end{definition}
As the universe variable represents the  universe (hence its name), it has a special role in both syntax and semantics of $\FC$. The \emph{syntax of~\FCtxt} restricts the syntax of $\logC$ in two ways: 
\begin{definition}\label{def:fc}\label{def:FCsyntax}
	The set $\FC$ of  \emph{\FC-formulas} is defined recursively: The atoms are word equations $(\eta_L \weqeq \eta_R)$ with $\eta_L\in \Xi$ and $\eta_R\in(\Sigma\cup\Xi)^*$. These can be combined using disjunction $(\phi\lor\psi)$, conjunction $(\phi\land\psi)$, negation $\neg\phi$, and quantifiers $\ex{x}{\phi}$ and $\fa{x}{\phi}$ with $x\in\Xi\setdiff\{\sv\}$.
\end{definition}
In other words, firstly, every word equation has a single variable on its left side.
 Secondly, the universe variable $\sv$ may not be bound by quantifiers.  
 The reason for the first restriction is a bit  subtle; we shall discuss it  after defining the semantics. 
But the other follows immediately from the intuition that $\sv$ shall represent the universe -- hence, binding it would make no sense.
For the same reason, we also exclude $\sv$ from the free variables of $\FC$-formulas:
\begin{definition}
The set $\fvar(\phi)$ of free variables of an $\FC$-formula $\phi$ is defined as for $\logC$-formulas, with the exception that  $\sv$ is not considered a free variable.  
\end{definition}
The \emph{semantics of~\FCtxt} combine those of $\logC$ with the additional condition that the universe consists only of factors of the content of the universe variable~$\sv$:
\begin{definition}\label{def:compat}
For $\phi\in\FC$ and a pattern substitution~$\sigma$ with $\dom(\sigma)\supseteq \fvar(\phi)\cup\{\sv\}$, we define 
$\sigma\models\phi$ as for~$\logC$, but
with  the additional condition that $\sigma(x)\factor\sigma(\sv)$ for all $x\in\dom(\sigma)$.
\end{definition}
To highlight the special role of $\sv$, we also write  $(w,\sigma)\models\phi$ if $\sigma\models\phi$ and $w=\sigma(\sv)$. 
We may shorten this to $w\models\phi$ if $\phi$ is a \emph{sentence} -- that is, if $\fvar(\phi)=\emptyset$.
We write $\phi(\vec{x})$ to denote that $\vec{x}$ is a tuple of free variables of $\phi$. 
\begin{example}
Define $\phi_1\avs{y} \df \ex{x}{x \weqeq\mathtt{papaya}\conc y\conc \mathtt{banana}}$
and $\phi_2 \df \ex{x}{(x \weqeq \mathtt{papaya}\lor x\weqeq\mathtt{banana})}$.
Then $(w,\sigma)\models \phi_1$ if and only if $\sigma(y)$ occurs in $w$ between $\mathtt{papaya}$ and $\mathtt{banana}$, and
$w\models \phi_2$ if and only if $w$ contains $\mathtt{papaya}$ or $\mathtt{banana}$ as factor. 
Finally, let $\phi_3(x)\df \ex{p,s}{\bigl(\sv\weqeq p\conc x\conc s \land \lnot\ex{\hat{p},\hat{s}}{(\sv\weqeq \hat{p}\conc x\conc \hat{s} \land\neg\hat{p}\weqeq p)}\bigr)} $. Then $(w,\sigma)\models\phi_3$ if and only if $\sigma(x)$ occurs exactly once in $w$.
\end{example}
When applying $\sigma$ to an $\FC$-formula, $\sigma(\sv)$ always needs to be defined -- otherwise, we would have no universe to work with. 
But $\FC$-formulas do not need to contain $\sv$. As a rule of thumb, $\sv$ is only required when referring to some ``global'' property of $w$.
If we describe properties that are more ``local'' (as in the next example), we usually do not need to use $\sv$.
\begin{example}\label{example:usefulRelations}
	Let $\phi^{\pref}(x,y)\df \ex{z}{y\weqeq xz}$. Then $\sigma\models\phi$ if and only if $\sigma(x)$ and $\sigma(y)$ are factors of $\sigma(\sv)$  with $\sigma(x)\pref \sigma(y)$. In other words, $\phi^{\pref}$ expresses $x\pref y$. 
	Consequently, we can express $x\ppref y$ through $\phi^{\ppref}(x,y)\df \phi^{\pref}(x,y)\land\neg x\weqeq y$.
	
	In fact, $\ppref$ and  inequality  can be expressed without  negation (or universal quantifiers). 
	First, define 
	$\phi(x)^{\neq\emptyword}\df \ex{y}{\biglor_{\ta\in\Sigma}x\weqeq \ta\conc y}$ to express $x\neq \emptyword$ -- that is, $\sigma\models\phi^{\neq\emptyword}$ if and only if $\sigma(x)\factor\sigma(\sv)$ and $\sigma(x)\neq\emptyword$.
	We use this in
	$\psi^{\ppref}(x,y)\df \ex{z}{(y\weqeq x\conc z\land \phi^{\neq\emptyword}(z))}$.
	Like $\phi^{\ppref}$, this expresses 
	$x\ppref y$;
	but without   negation.
	
	Finally, let $\phi^{\neq}(x,y)\df \psi^{\ppref}(x,y)\lor \psi^{\ppref}(y,x) \lor \bigvee_{\ta,\tb\in \Sigma, \ta\neq \tb} \ex{x_1,y_1,z}{(x\weqeq z\conc \ta\conc x_1\land y\weqeq z\conc \tb\conc y_1)}$. This states that $x\ppref y$, $y\ppref x$, or $x$ and $y$ differ after a common prefix $z$ -- that is, $x\neq y$.
\end{example}
We say  $\phi,\psi\in\FC$ are \emph{equivalent}, written $\phi\equiv\psi$, if for all $\sigma$ with $\dom(\sigma)\supseteq \fvar(\phi)\cup\fvar(\psi)\cup\{\sv\}$, we have that  $\sigma\models\phi$ holds if and only if $\sigma\models\psi$. 
Thus, in Example~\ref{example:usefulRelations}, we have $\phi^{\ppref}\equiv\psi^{\ppref}$.
If $\phi\in\FC$ is a sentence, we define its language  as $\Lang(\phi)\df\{w\mid w\models \phi\}$.
\begin{example}\label{ex:starfree}
	A language is called \emph{star-free} if it is defined by a regular expression $\alpha$ that is constructed from the empty set $\emptyset$,  terminals $\ta\in\Sigma$,  concatenation~$\cdot$,  union~$\ror$, and  complement~$\compl{\alpha}$.
	Given such an $\alpha$, we define $\phi^{\alpha}\df\ex{x}{(\sv\weqeq x \land \psi^{\alpha}\avs{x})}$, where  $\psi^{\alpha}\avs{x}$ is defined recursively by
	$\psi^{\emptyset}\avs{x}\df \lnot (x\weqeq x)$, 
	$\psi^{\ta}\avs{x}\df (x\weqeq \ta)$,
	$\psi^{(\alpha_1\cdot\alpha_2)}\avs{x}\df \ex{x_1,x_2}{\bigl(x\weqeq x_1\conc x_2 \land \psi^{\alpha_1}\avs{x_1}\land \psi^{\alpha_2}\avs{x_2}\bigr)},$ 
	$\psi^{(\alpha_1\ror\alpha_2)}\avs{x}\df \psi^{\alpha_1}\avs{x}\lor \psi^{\alpha_1}\avs{x}$, and
	$\psi^{\compl{\alpha}}\avs{x}\df \lnot\psi^{\alpha}\avs{x}.$
	Then $\sigma\models\psi^{\alpha}$ if and only if $\sigma(x)\in\Lang(\alpha)$ and $\sigma(x)\factor\sigma(\sv)$. Thus, $\Lang(\phi^{\alpha})=\Lang(\alpha)$. 
\end{example}
We are now ready to discuss why Definition~\ref{def:FCsyntax} restricts the left sides of word equations to single variables. 
Assume we allowed, for instance, the word equation $xy\weqeq yx$ in an $\FC$-formula, and consider the case of $\sigma(\sv)=\ta^3$ and $\sigma(x)=\sigma(y)=\ta^2$. 
Then  $\sigma(x)\factor\sigma(\sv)$,  $\sigma(y)\factor\sigma(\sv)$, and $\sigma(xy)=\sigma(yx)$ hold, but $\sigma(xy)=\ta^4$ is not a factor of $\sigma(\sv)$, which means that it is not in the  universe.

There are two straightforward ways of allowing arbitrary word equations $\eta_L\weqeq\eta_R$ in $\FC$ without changing the underlying universe.
The first is adding the additional requirements $\sigma(\eta_L)\factor\sigma(\sv)$ and $\sigma(\eta_R)\factor\sigma(\sv)$ to the definition of $\sigma\models (\eta_L\weqeq\eta_R)$.
This can also be understood as declaring the concatenation as undefined if its result is not a factor of $\sigma(\sv)$.
The second is interpreting $\eta_L\weqeq\eta_R$ as syntactic sugar for $\ex{x}{(x\weqeq\eta_L\land x\weqeq\eta_R)}$, where $x$ is a new variable. 

On the other hand, this re-interpretation of the solutions of word equations can be considered non-intuitive, which makes formulas that rely on these easy to misunderstand. 
To avoid these issues, this paper restricts every left side to a single variables, even though this is not strictly necessary.

\section{Properties of \FCtxt}\label{sec:properties}
In this section we analyze $\FC$ dynamically by discussing 
its \emph{evaluation problem}  -- given a formula $\phi\in \FC$ and a pattern substitution $\sigma$, decide whether $\sigma\models\phi$. We call the special case where $\phi$ is a sentence the  \emph{model checking problem}. 
We also consider  
the \emph{satisfiability problem} -- given $\phi\in\FC$, decide whether there is a pattern substitution $\sigma$ with $\sigma\models\phi$.
After that, we also consider aspects of optimization of formulas.
\subsection{Model checking vs satisfiability}\label{sec:mc}
For $\logC$, the model checking problem is undecidable. This is due to two reasons. Firstly,   the satisfiability problem for $\logC$ is undecidable by Quine~\cite{qui:con}. Secondly, for $\logC$, satisfiability reduces to model checking --  from a given $\phi\in\logC$, we can construct a $\logC$-sentence $\phi'$ by binding all free variables of $\phi$ existentially. Then $\phi$ is satisfiable if and only if $\sigma\models\phi'$, no matter which substitution $\sigma$ we choose.
In contrast to this, the finite universe of $\FC$ drastically reduces the complexity of model checking.
\begin{theorem}[label=thm:recog,restate=restateThmRecog]
Evaluation  is $\compclass{PSPACE}$-complete for  $\FC$ and $\compclass{NP}$-complete for  $\fragExPos{\FC}$. 
\end{theorem}
In fact, the proof shows that the lower bounds hold even in the special case of model checking $\ta\models\phi$.
Both the drop in complexity and the fact a very simple structure suffice
are comparable to $\FO$ on finite relations (see \eg~\cite{lib:ele}), and the proofs are equally straightforward. 
For both logics, the hardness of the problem comes from parameters of the formula and not of the word or the relational structure.
$\FO$ provides us with another parameter to lower the complexity of model checking. 
We define  the \emph{width} $\formulaWidth(\phi)$ of a formula $\phi$ as  the maximum number of free variables in any of its subformulas.
\begin{theorem}[label=thm:mc,restate=restateThmMC]
	Model checking for  $\FC$   can be solved in 
		time 
		$O(k|\phi|n^{2k})$, for $k\df\formulaWidth(\phi)$ and 
		$n\df |\sigma(\sv)|$.
\end{theorem} 
The proof also shows that this is only a rough upper bound; taking  properties of variables into account lowers the  exponent.
In principle, we can  apply various structure parameters for first-order formulas (see \eg Adler and Weyer~\cite{adl:tree})
to $\FC$.
This assumes that we treat word equations as atomic formulas, which is certainly possible -- but we can do better than that.
\subparagraph*{Decomposing patterns}
Using a word equation $x\weqeq \alpha$ as an atom results in a formula that has a width of at least $|\var(\alpha)|$.
Our goal is to lower that bound, by decomposing the pattern into a formula.
Technically, a pattern $\alpha=\alpha_1\cdots\alpha_n$ with $\alpha_i\in(\Xi\cup\Sigma)$  is a term $f(\alpha_1,\ldots,\alpha_n)$, where the function $f$ is the $n$-ary concatenation.
But there is a syntactic criterion that allows us to decompose $\alpha$ into a conjunction of binary concatenations.
This builds on  a result from combinatorics on words and formal languages, where
a pattern $\alpha$ is also treated as generators of the \emph{pattern languages} $\Lang(\alpha)$; the set of images of $\alpha$ under pattern substitutions. In this context, 
Reidenbach and Schmid~\cite{rei:pat} started a series of articles on classes of \emph{pattern languages} with a polynomial time membership problem (surveyed in~\cite{man:mat}), most of which rely on the following definition (see 
\cite{cyg:par} or \cref{app:thm:patToFC}
for the definition of treewidth).
\begin{definition}\label{def:patGraph}
The \emph{standard graph} of a pattern $\alpha=\alpha_1\cdots\alpha_n$ with $n\geq 1$ and $\alpha_i\in(\Sigma\cup\Xi)$ is  $\patgraph{\alpha}\df (V_{\alpha},E_{\alpha})$ with $V_\alpha\df\{1,\dots,n\}$ and $E_{\alpha}\df E^{<}_{\alpha}\cup E^{=}_{\alpha}$, where $E^{<}_{\alpha}$ is the set of all $\{i,i+1\}$ with $1\leq i<n$, and $E^{=}_{\alpha}$ is the set of all $\{i,j\}$ such that $\alpha_i$ is some $x\in\Xi$, and $\alpha_j$ is the next occurrence of $x$ in $\alpha$. Then $\treeWidth{\alpha}$, the \emph{treewidth} of $\alpha$, is the treewidth of $\patgraph{\alpha}$.
\end{definition}
As  artificial (but simple) example, consider the sequence of patterns $\alpha_n\df x_1x_1 x_2x_2 \cdots x_nx_n$.
Then $|\var(\alpha_n)|=n$,  affecting the width of formulas that use $\alpha_n$ accordingly, but $\treeWidth{\alpha_n}=1$.
Using tree decompositions, we can rewrite  patterns of bounded treewidth into formulas of bounded width
(similar to the proof of Kolaitis and Vardi~\cite{kol:con} for variable bounded $\FO$).
\begin{theorem}[label=thm:patToFC,restate=restateThmPatToFC]
Let $\phi\df \ex{x_1,\ldots,x_m}{y\weqeq \alpha}$. Then there exists $\psi\in\fragExPos{\FC}$  with  $\psi\equiv\phi$ and  $\formulaWidth(\psi)\leq 2\treeWidth{\alpha}+v$, where $v=2+|\fvar(y\weqeq\alpha)\setdiff\{x_1,\dots,x_m\}|$. 

For every fixed $k$, given  $\phi$ with $\treeWidth{\alpha}\leq k$, we can compute $\psi$ in polynomial time.
\end{theorem}
Combining Theorems~\ref{thm:mc} and~\ref{thm:patToFC} yields a (slightly) different proof of the
polynomial time decidability of the membership problem for classes of patterns with bounded treewidth from~\cite{rei:pat}.  
As pointed out in~\cite{day:mat}, 
bounded treewidth does not cover all pattern languages with a polynomial time membership problem, like \eg patterns $\alpha^k$ 
where $\treeWidth{\alpha}$ is bounded. But these languages can be expressed by $\ex{x}{(\sv\weqeq x^k \land \phi_{\alpha}(x))}$, where $\phi_{\alpha}(x)$ is a formula that expresses $x\in\Lang(\alpha)$, thus increasing the width by one.
We leave a systematic examination whether all criteria for patterns with a tractable membership problem map to $\FC$-formulas of bounded width for future work.

\subparagraph*{Satisfiability}
Another parallel to $\FO$ is that satisfiability is undecidable for $\FC$, even if we use only few variables. 
Let $\FC^k$ denote the set of formulas with width at most $k$.
\begin{proposition}[label=prop:FcSat,restate=restatePropFcSat]
	Satisfiability for $\FC^{3}$ is undecidable if $|\Sigma|\geq 2$.
\end{proposition}
The problem is trivial for $\FC^0$ (see the proof of Theorem~\ref{thm:nonrec}) and open for $\FC^1$ and $\FC^2$.

\subsection{Static optimization}\label{sec:nonrecto}
Apart from the width, Theorem~\ref{thm:mc} highlights the length of a formula as another parameter that influences the complexity of model checking.
While the length of the patterns in the word equations might not seem to be factor that is overly important, there are patterns where straightforward optimizations can lead to an exponential advantage.
\begin{example}\label{example:patCompress}
	For  $k\geq 1$, let $\phi_k(y)\df \ex{x}{y\weqeq x^{2^k}}$. Then $\phi_k\equiv \psi_k\df \ex{x_1,\ldots,x_k}{\bigl(y\weqeq x_1x_1 \land \bigland_{i=1}^{k-1} x_i \weqeq x_{i+1} x_{i+1}\bigr)}$, and $|\phi_k|$ is exponential in $k$, while $|\psi_k|$ is linear in $k$.   We can also rewrite $\psi_k$ into a formula of width 3 by pulling quantifiers inwards and reusing variables.
	More specifically, we first rewrite each $\psi_k$ into the equivalent formula 
	\[
	\ex{x_1}{(y\weqeq x_1x_1 \land 
		(
		\ex{x_2}{
			x_1 \weqeq x_2x_2	
			\land \cdots (\ex{x_{k-1}}{x_k\weqeq x_{k-1}x_{k-1}})
			\cdots }))}.\]
	Then we replace every variable $x_i$ with $x_1$ if $i$ is odd or $x_2$ if $i$ is even. The resulting formula has width 3 (due to $y$), is equivalent to $\phi_k$, and has the same length.
\end{example}
This raises the  questions whether we can computably minimize formulas and whether some fragments are more succinct than others.
We address these questions in order.
\begin{theorem}[label=thm:undecShort,restate=restateThmUndecShort]
There is no algorithm that, given $\phi\in\FC$, computes an equivalent $\psi$ such that $|\psi|$ is minimal.
This holds even if we restrict this to minimization within $\fragExPos{\FC^{4}}$.
\end{theorem}
This
 leaves open the decidability of, given $\phi\in\FC$ (or $\phi\in\fragExPos{\FC}$) and  $k>0$, is there an equivalent $\psi\in\FC^{k}$. 
But without suitable inexpressibility methods (see Section~\ref{sec:inexpressibility}), we cannot  even show that a  language is inexpressible in $\FC^k$ for some $k>0$, which complicates tackling this problem.
The proof of \cref{thm:undecShort} is actually more general and also demonstrates the undecidability of other common problems, like containment and equivalence.

Via Hartmanis'~\cite{har:god} meta theorem, certain undecidability results  provide insights into the relative succinctness of models (see~\cite{kut:phe} or \eg~\cite{fre:doc} for details). For two logics $\mathcal{F}_1$ and $\mathcal{F}_2$,  the \emph{tradeoff} from $\mathcal{F}_1$ to $\mathcal{F}_2$ is \emph{non-recursive} if, for every  computable  $f\colon\mathbb{N}\to\mathbb{N}$, there exists some $\phi\in\mathcal{F}_1$ that is expressible in $\mathcal{F}_2$, but $|\psi|\geq f(|\phi|)$ holds for every $\psi\in\mathcal{F}_2$ with $\psi\equiv\phi$.
\begin{theorem}[label=thm:nonrec,restate=restateThmNonrec]
	There are non-recursive tradeoffs 
from $\fragExPos{\FC^4}$ to regular expressions and~$\FC^0$; and
  from $\FC^{4}$ to $\fragExPos{\FC}$, 
 patterns, and  singleton sets $\{w\}$.
\end{theorem}
Note in particular that patterns can be parts of word equations. Hence, where Example~\ref{example:patCompress} showed an exponential advantage in the rewriting, Theorem~\ref{thm:nonrec} shows $\FC^{4}$ can obtain far larger advantages on certain classes of patterns.
\subsection{Iteration and recursion}\label{sec:iterationRecursion}
Iteration and recursion have been extensively studied in finite model theory and database theory. 
In particular,  $\FOord$ that is extended with operators for transitive closure or fixed points  captures various complexity classes (see \eg \cite{ebb:fin,lib:ele}).
This is also closely connected to the recursive query language  \Datalogtxt (see \eg \cite{abi:fou}).
In this section, we shall define \DataSplogtxt, an  $\FC$-analog of \Datalogtxt. 
On the way, we shall also see that using transitive closures or fixed points on $\FC$ instead of $\FOord$ characterizes the same complexity classes.

\subparagraph*{Iteration} The definitions of these operators and the resulting extensions of $\FC$ are straightforward adaptions of their $\FOord$-versions (see \eg \cite{lib:ele} or~\cite{ebb:fin}).   
But as they are also rather lengthy,  we only give the intuitions behind them (detailed definitions can be found in \cref{app:def:tcfp}).

For $k\geq 1$, an $\FC$-formula with $2k$ free variables can be viewed as generator of a relation over $R\subseteq S^k\times S^k$, where $S$ is the universe (the set of factors of $\sv$).
The operator $\tc$  computes the \emph{transitive closure} $\tc(R)$ of $R$.
If we view $R$ as edge set of a directed graph over $S^k$, then $\tc$ computes the reachability relation in this graph.
The \emph{deterministic transitive closure}  $\dtc$ is defined analogously, with the additional restriction that $\dtc$ stops at nodes with more than one outgoing edge.
$\TC{\FC}$ and $\DTC{\FC}$ extend $\FC$ with $\tc$ and $\dtc$,  respectively. 

For fixed points, we introduce special relation symbols as part of inductive definitions. 
Inside a fixed point operator, a formula $\phi$  may use a symbol $\relsym$  and at the same time also define the  relation $R$ inductively. 
We start with $R_0\df\emptyset$ and let $R_1$ be the relation that is defined by $\phi$ if $\relsym$ represents $R_0$. 
This is repeated,  each $R_i$ giving rise to $R_{i+1}$, until a fixed point is reached.
For \emph{least fixed points}, we ensure  $R_i\subseteq R_{i+1}$  for all $i$. 
For \emph{partial fixed points}, this  is not required.
We use $\LFP{\FC}$ and
$\PFP{\FC}$ for the respective extensions of $\FC$.
 
Complexity classes are commonly defined as classes of languages; and as we can treat $\FC$ and its extensions as language generators, connecting these two worlds is straightforward.
We say that a logic $\mathcal{F}$  \emph{captures} a complexity class $\cc$ if $\cc$  is the class of languages that are $\mathcal{F}$-definable -- that is, $\cc=\{\Lang(\phi)\mid\phi\in\mathcal{F}\}$.

The following result mirrors that for the respective extensions of $\FOord$:
\begin{theorem}[label=thm:capture,restate=restateThmCapture]
$\DTC{\FC}$, 
$\TC{\FC}$, 
$\LFP{\FC}$, 
$\PFP{\FC}$
capture
$\compclass{L}$,
$\compclass{NL}$,
$\compclass{P}$, 
$\compclass{PSPACE}$, respectively.
\end{theorem}
The result holds even if the formulas are required to be existential-positive.
Thus, $\FC$ and even $\fragExPos{\FC}$ behave  under fixed-points and transitive closures like $\FOord$. 

\subparagraph*{Recursion} This connection immediately suggests another: Recall that $\FO$ with least-fixed point operators can be used to define \Datalogtxt (see \eg Part~D of~\cite{abi:fou}). 
Analogously, we define \DataSplogtxt, a version of \Datalogtxt that is based on word equations. 

An $\DataSplog$-program is a tuple $P\df(\rulesymbs,\ruleset,\outDatalog)$, 
where $\rulesymbs$ is a set of relation symbols that contains a special output symbol \outDatalog, each $R\in\rulesymbs$ has an arity $\ar(R)$, and  $\ruleset$ is a finite set of rules 
$
R(\vec{x}) \leftarrow \phi_1(\vec{y}_1),\ldots,\phi_m(\vec{y}_m)
$
with $R\in\rulesymbs$, $m\geq 1$, each $\phi_i$ is an $\FC$-word equation, and each $x$ of $\vec{x}$ appears in some~$\vec{y}_i$.

We define $\fun{P}(w)$ incrementally, initializing the relations of all $R\in\rulesymbs$ to $\emptyset$. 
For each rule 
$R(\vec{x}) \leftarrow \phi_1(\vec{y}_1),\ldots,\phi_m(\vec{y}_m)$,  
we enumerate all $\sigma$ with   $\sigma(\sv)= w$ and check if $\sigma\models\ex{\vec{y}}{\bigland_{i=1}^{m}\phi_i}$, where $\vec{y}\df \left(\bigcup_{i=1}^m \vec{y}_i\right)\setdiff\vec{x}$. If this holds, we add $\sigma(\vec{x})$ to $R$. This is repeated until all relations have stabilized. Then $\fun{P}(w)$ is the content of the relation $\outDatalog$.
\begin{example}
Define an \DataSplogtxt-program $(\{\outDatalog,E\},\ruleset)$, where $\ar(\outDatalog)=0$, $\ar(E)=3$, and $\ruleset$ consists of the  rules
$\outDatalog()\leftarrow \sv\weqeq xyz, E(x,y,z)$, and
$E(x,y,z)\leftarrow x\weqeq\emptyword, y\weqeq\emptyword, z\weqeq\emptyword,$ and
$E(x,y,z)\leftarrow x\weqeq \hat{x}\ta, y\weqeq\hat{y}\tb,z\weqeq\hat{z}\mathtt{c},E(\hat{x},\hat{y},\hat{z})$.
This defines the language $\{\ta^n\tb^n\mathtt{c}^n \mid n\geq 0\}$.
\end{example}
\begin{theorem}[label=thm:datasplog,restate=restatethmDatasplog]
$\DataSplog$ captures $\compclass{P}$.
\end{theorem}
This is unsurprising, considering \Datalogtxt on ordered structures captures $\compclass{P}$, see \eg~\cite{lib:ele}, and  the analogous result for spanners with recursion~\cite{pet:rec}. But it allows us to use word equations as a  basis for \Datalogtxt on words.
This provides potential applications for future insights into acyclicity for patterns, which could be combined with existing techniques for \Datalogtxt.

\DataSplogtxt can also be seen as a generalization of \emph{range concatenation grammars (RCGs)}, see~\cite{bou:ran, kal:par}, to use outputs and relations. 
There has been some work on parsing of RCGs (see~\cite{kallmeyer2009earley} and its references).  
In the future, these might help identify tractable fragments of \DataSplogtxt. 
Vice versa, insights into the latter might lead to new approaches to RCG-parsing.

\section{\FCtxt\ as a logic for document spanners}\label{sec:logicForSpanners}
Fagin et al.~\cite{fag:spa} introduced \emph{document spanners} (or just \emph{spanners}) as a formal model of information extraction that is based on relational algebra (see \eg~\cite{abi:fou}).
This section connects spanners to \FCtxt.
After stating the necessary definitions (\cref{sec:spanners}), we extend $\FC$  into a logic for spanners (Section~\ref{sec:FCandSpanners}) and then use this for an inexpressibility proof (Section~\ref{sec:inexpressibility}). 
\subsection{Spans and document spanners}\label{sec:spanners}
 A \emph{span} of $w\df \ta_1 \cdots \ta_n$ with $n\geq 1$ is an interval $\spn{i,j}$ with $1\leq i\leq j\leq n+1$. 
It describes the factor $w_{\spn{i,j}}=\ta_i\cdots \ta_{j-1}$. 
For  finite $V\subset\Xi$ and  $w\in\Sigma^*$, a 
\emph{$(V,w)$-tuple} is a function $\mu$ that maps  each variable in $V$ to a span of~$w$.
A \emph{spanner} with variables $V$ is a function $P$ that maps every $w\in\Sigma^*$ to a set $P(w)$ of $(V,w)$-tuples. We use $\SVars{P}$
for the variables of a spanner $P$.
Accordingly,  a spanner $P$ is a function that takes an input word $w$ and computes a relation $P(w)$ of $(\SVars{P},w)$-tuples.

Like~\cite{fag:spa}, we base spanners  on
\emph{regex formulas};
 regular expressions with variable bindings~$\bind{x}{\alpha}$. 
This matches the same words as the expression $\alpha$ and assigns the corresponding span of~$w$ to the variable $x$.
For the purpose of this article, this informal definition shall suffice.
Detailed definitions of the syntax and semantics of regex formulas can be found in~\cite{fag:spa} (the original definition that uses parse trees) and~\cite{fre:splog} (a more lightweight definition that uses the ref-words from Schmid~\cite{sch:cha}).

A~regex formula is \emph{functional} if  on every word, every match has exactly one assignment for  each variable. 
The set of functional regex formulas is $\RGX$.
For $\alpha\in\RGX$, we define the spanner $\fun{\alpha}$ as follows.
Every match on $w\in\Sigma^*$ defines a $(\SVars{\alpha},w)$-tuple~$\mu$, where each $\mu(x)$ is the span assigned to~$x$; and  $\fun{\alpha}(w)$ is the set of all these $\mu$.

\begin{example}\label{ex:regexformula}
We consider the regex formula $
\alpha \df \Sigma^*  \bigl(\bind{x}{\mathtt{banana}} \ror \bind{x}{\mathtt{papaya}}\bigr) \Sigma^*,
$
which matches every word $w$ that contains an occurrence of $\mathtt{banana}$ or $\mathtt{papaya}$.  
The corresponding spanner $\fun{\alpha}(w)$ contains all spans $\spn{i,j}$ with $w_{\spn{i,j}}\in\{\mathtt{banana}, \mathtt{papaya}\}$.
Next, we define
$
\beta \df \Sigma^* \bind{x}{\Sigma^*}\Sigma^* \bind{y}{\Sigma^*}\Sigma^*.
$
For every $w\in\Sigma^*$, we have that $\fun{\beta}(w)$ contains those $\mu$ where $\mu(x)$ refers to a span  to the left of $\mu(y)$.
\end{example}
We use the spanner operations union~$\cup$, natural join~$\join$, projection~$\pi$, set difference~$\setdiff$, and  equality selection $\sel$.
For spanners $P_1$ and $P_2$ with $\SVars{P_1}=\SVars{P_2}$, \emph{union} and \emph{set difference} are defined  by $(P_1\cup P_2)(w)\df P_1(w)\cup P_2(w)$ and $(P_1\setdiff P_2)(w)\df P_1(w)\setdiff P_2(w)$ on all $w\in\Sigma^*$. 
Furthermore, the \emph{projection} $\proj_V P$ for a spanner $P$ and a set of variables $V\subset\SVars{P}$ is obtained for every $w\in\Sigma^*$ by restricting the domain of every $\mu\in P(w)$ to $V$. 

The \emph{natural join} $P_1\join P_2$ combines spanner results by merging tuples that agree on the common variables. That is, $(P_1\join P_2)(w)$
contains those $(\SVars{P_1}\cup\SVars{P_2},w)$-tuples~$\mu$ for which there exist $\mu_1\in P_1(w)$ and $\mu_2\in P_2(w)$ such that $\mu_1(x)=\mu_2(x)$ for all $x\in \SVars{P_1}\cap \SVars{P_2}$.
An important consequence of this  definition is that join is defined using spans (and thereby positions in the input word), not using the factors that occur in the spans. 
To  compare factors, we use the \emph{equality selection} $\sel_{x,y}P$ with  $x,y\in\SVars{P}$. This is defined by $\sel_{x,y}P(w)\df\{ \mu\in P(w)\mid w_{\mu(x)}=w_{\mu(y)}\}$ for $w\in\Sigma^*$.

By combining regex formulas with symbols for spanner operations, we obtain \emph{spanner representations}; and their semantics are defined by applying the operations.
The class of \emph{generalized core spanner representations} $\RGXcored$ consists of combinations of $\RGX$ and any of the five operators; the \emph{core spanner representations} $\RGXcore$ exclude set~difference.
According to Fagin et al.~\cite{fag:spa}, ``core spanners'' capture the core functionality of IBM's SystemT. 

\begin{example}\label{ex:spanner}
	Let $\alpha$ and $\beta$  be the regex formulas from Example~\ref{ex:regexformula}.
	We define the spanner representation
$
	\varrho_1 \df \alpha(x) \join \alpha(y) \join \beta(x,y).
$
	Then $\SVars{\varrho_1}=\{x,y\}$, and  $\fun{\varrho_1}(w)$ contains those $\mu$ where $\mu(x)$ occurs before $\mu(y)$ in $w$ and each of $w_{\mu(x)}$ and $w_{\mu(y)}$ is $\mathtt{banana}$ or $\mathtt{papaya}$.
	Now let $\varrho_2 \df \sel_{x,y} \varrho_1$. Then $\fun{\varrho_2}(w)$ is the subset of $\fun{\varrho_1}(w)$ that also  has 
	$w_{\mu(x)}=w_{\mu(y)}$.
\end{example}
We identify spanners and their representations; \eg by referring to a representation $\varrho$ as a spanner (technically, $\fun{\varrho}$ is the spanner) or by calling the elements of $\RGXcore$ \emph{core spanners}.

\subsection{Adding expressive power to \FCtxt}\label{sec:constraints}
As core spanners are based on regular expressions, they can define all regular languages. 
This makes them more powerful than $\fragExPos{\FC}$.
To prove this, we first connect $\FC$ to~$\logC$. 
\begin{lemma}[label=lem:FCtoC,restate=restateLemFCtoC]
	Given $\phi\in\FC$, we can construct in polynomial time $\psi\in\logC$ such that $\sigma\models\phi$ if and only if $\sigma\models\psi$. This also preserves the properties existential and existential-positive.
\end{lemma}
Hence, $\fragExPos{\FC}$ is not more expressive than $\fragExPos{\logC}$, which cannot express all regular languages -- not even comparatively ``harmless'' languages like \eg $\{\ta,\tb\}^*\mathtt{c}$ (see Karhumäki, Mignosi, Plandowski~\cite{kar:exp}).
While we could define this specific language using negation, we shall address the issue in a way  that generalizes far beyond regular languages and that does not require us to leave the existential-positive fragment (and its friendlier upper bounds).
Complexity is also a reason why we do not use $\MSO$ or define a second-order version of $\FC$.

Instead, take inspiration from $\logC$ (see Diekert~\cite{die:mak}). 
The \emph{theory of concatenation with regular constraints}, $\logCreg$, extends $\logC$ by allowing \emph{regular constraints} $\constr{x}{\alpha}$ as atoms, where $x\in\Xi$, $\alpha$ is a regular expression, and $\sigma\models \constr{x}{\alpha}$ if $\sigma(x)\in\Lang(\alpha)$. 
We define $\FCreg$ analogously, where $\sigma\models \constr{x}{\alpha}$ has the additional condition that $\sigma(x)\factor\sigma(\sv)$ must hold.
\begin{theorem}[label=thm:alsoWithConstraints,restate=restateThmAlsoWithConstraints]	
Theorem~\ref{thm:recog} and Theorem~\ref{thm:mc} also hold if we replace $\FC$ with $\FCreg$ and $\fragExPos{\FC}$ with   $\fragExPos{\FCreg}$.
\end{theorem}
In other words, evaluation is $\compclass{PSPACE}$-complete for $\FCreg$ and $\compclass{NP}$-complete for $\fragExPos{\FCreg}$, and formula width can be used as parameter to bound model checking for $\FCreg$.
This generalizes to all constraints that can be decided in polynomial time, which allows us to adapt $\FC$ to other settings as well. 

For example, string solvers often use \emph{length constraints}.
There are predicates that compare words by applying arithmetic to their lengths, like $|x|+|y|=|z|$.
While the applications of $\fragExPos{\logC}$ in a string solver context usually rely on deciding satisfiability, cases where model checking suffices could benefit from using $\FC$ with appropriate constraints.

Regarding prior work, the $\logCreg$-fragments $\splog$ and $\splogneg$ were introduced in~\cite{fre:splog} as alternatives to $\RGXcore$ and $\RGXcored$, respectively. As these ensure the finite universe purely through syntax, they  are  more cumbersome than $\FC$ and do not generalize as nicely.

Before we connect $\FCreg$ to spanners, we take a brief look at restricted regular expressions that can be expressed in $\FC$. We call a regular expression  \emph{simple} if the  operator~$^*$ is only applied to  terminal words or to $\Sigma$ (a shorthand for $\bigcup_{\ta\in\Sigma} \ta$). 
That is, if $\Sigma=\{\ta,\tb,\mathtt{c}\}$, then $(\mathtt{abc})^*\Sigma^*$ is simple, but $(\ta\cup\tb)^*$ and $(\ta(\tb)^*)^*$ are not.
\begin{lemma}[restate=restateLemSimple,label=lem:simple]
	For every simple regular expression $\alpha$, there is $\phi^{\alpha}(x)\in\fragExPos{\FC}$ such that $(w,\sigma)\models\phi^{\alpha}$ if and only if $\sigma(x)\in\Lang(\alpha)$ and $\sigma(x)\factor w$.
\end{lemma}
The proof uses the characterization of commuting words (see \eg  Lothaire~\cite{lot:com}).
We shall use this Lemma in the proof of Theorem~\ref{thm:equalLength}, to replace regular constraints.
\subparagraph*{\FCregtxt and Spanners}\label{sec:FCandSpanners}
As we want to use $\FCreg$ for spanners, we still need to close a formal gap, namely that  spanners reason over positions in a word, while $\FCreg$  reasons over words. 
We bridge this gap through the notion of one \emph{realizing} the other, which~\cite{fre:splog} introduced for the logic $\splog$.
We begin with formulas that realize spanners.
\begin{definition}\label{def:splogreal}	
	A substitution $\sigma$ \emph{expresses} a $(V,w)$-tuple $\mu$ if $\dom(\sigma)\supseteq\{x^P,x^C\mid x\in V\}$ and, for all $x\in V$, we have $\sigma(x^P)=w_{\spn{1,i}}$ and $\sigma(x^C)=w_{\spn{i,j}}$ for $\spn{i,j}=\mu(x)$.
	
	A formula  $\phi\in\FCreg$ \emph{realizes} a spanner  $P$ 
	if  $\fvar(\varphi)=\{x^P,x^C\mid x\in\SVars{P}\}$
	and, for all $w\in\Sigma^*$, we have 
	$(w,\sigma)\models\phi$ if and only if $\sigma$ expresses some $\mu\in P(w)$.
\end{definition}
In other words, $x^C$ is $w_{\mu(x)}$ (the content of $x$), and $x^P$ is the prefix of $w$ before $w_{\mu(x)}$.
\begin{example}
	In Example~\ref{ex:regexformula},  we defined
	$\alpha \df \Sigma^*  \bigl(\bind{x}{\mathtt{banana}} \ror \bind{x}{\mathtt{papaya}}\bigr) \Sigma^*.$
	Its spanner $\fun{\alpha}$ is realized by
$\phi\avs{x^P,x^C}\df \ex{y}{\sv\weqeq x^P x^C y} 
	\land \bigl(x^C\weqeq \mathtt{banana} \lor x^C \weqeq \mathtt{papaya} \bigr).$
	
	Then $(w,\sigma)\models \phi$ if $\sigma$ expresses some $\mu\in\fun{\alpha}(w)$. That is, $\sigma(x^C)$ contains $w_{\mu(x)}$ (\ie, $\mathtt{banana}$ or $\mathtt{papaya})$, and $\sigma(x^P)$ contains the prefix in $w$ before it.
\end{example}
To show that $\FCreg$ cannot express more than the classes of spanners that we consider, we also define the notion of spanners that realize formulas.
\begin{definition}
 A spanner $P$ \emph{realizes} $\phi\in\FCreg$ if  $\SVars{P}=\fvar(\varphi)$ and, for all $w\in\Sigma^*$, we have 
$\mu\in P(w)$ if and only if $(w,\sigma)\models\phi$  for the $\sigma$  with $\sigma(x)\df w_{\mu(x)}$ for all  $x\in\SVars{P}$.
\end{definition}
There are \emph{polynomial-time conversions} from a class of formulas (or spanners) $A$ to a class of spanners (or formulas) $B$ if, given $x\in A$, we can compute in polynomial time $y\in B$ that realizes $x$.
We write $A\polyeq B$ if there are polynomial-time conversions from $A$ to $B$ and from $B$ to $A$.
\begin{theorem}[label=thm:FCvsSpanners, restate=restateThmFCvsSpanners]
 $\FCreg\polyeq\RGXcored$ and  $\fragExPos{\FCreg}\polyeq\RGXcore$. 
\end{theorem}
\subsection{Inexpressibility for \FCtxt, \FCregtxt, and spanners}\label{sec:inexpressibility}
There are currently only few inexpressibility methods for $\FC$ and  $\FCreg$, as there are only few such methods for related models like spanners or the theory of concatenation.
A detailed discussion from the point of view of $\RGXcore$ and $\splog$ can be found in Section~6 of~\cite{fre:splog}.
These techniques do not account for negation, which makes them inapplicable for $\FC$ or $\FCreg$.
A standard tool for  $\FO$-inexpressibility  are \EF games (\eg \cite{lib:ele}). But as concatenation acts as a generalized  addition,  using these for $\FC$ or $\FOstr$ is far from straightforward.
Another standard tool is the Feferman-Vaught theorem (see~\cite{DBLP:journals/apal/Makowsky04}). While this can be used for $\FC$, the factor universe of $\FC$ makes decomposing the structure into disjoint sets inconvenient. 
Instead of following down this road, we introduce $\FOstr$, an extension of $\FOord$ that has the same expressive power as $\FC$.
\subparagraph*{Connecting \FCtxt to \FOordtxt}\label{sec:fo}\label{sec:fostreq}
In this section, we establish connections between $\FC$ and ``classical'' relational first-order logic.
It is probably safe to say that in finite model theory, the most common way of applying first-order logic to words is the logic $\FOord$ (and the more general $\MSO$). This uses the equality~$\foeq$  and a vocabulary that consists of a binary relation symbol~$<$ and unary relation symbols $\lettPred{\ta}$ for each $\ta\in\Sigma$.
Every word $w=a_1\cdots a_n\in\Sigma^+$ with $n\geq 1$ is represented by a structure $\structClassic{w}$ with universe~$\{1,\dots,n\}$. For every $\ta\in \Sigma$, the relation $\lettPred{\ta}$ consists of those $i$ that have $a_i = \ta$. 
To simplify dealing with~$\emptyword$, we slightly deviate from this standard structure. For every $w\in\Sigma^*$, we extend $\structClassic{w}$ to  $\struct{w}$ by adding an additional ``letter-less'' node $|w|+1$ that occurs in no $\lettPred{\ta}$. 
Then we have a one-to-one correspondence between pairs $(i,j)$ with $i\leq j$ from the universe of $\struct{w}$ and the spans $\spn{i,j}$ of~$w$ (see Section~\ref{sec:spanners}), and  $w=\emptyword$ does not require a special case.
\begin{definition}
$\FOstr$ extends $\FOord$ with constants $\minConst$ and $\maxConst$, the binary relation symbol $\nextPred$, and the 4-ary relation symbol $\strEqPred$.
For every $w\in \Sigma^*$ and the corresponding structure $\struct{w}$, these symbols express  $\minConst=1$, $\maxConst=|w|+1$, $\nextPred=\{(i,i+1)\mid 1\leq i \leq |w|\}$, and $\strEqPred$ contains those $(i_1,j_1,i_2,j_2)$ with $i_1\leq j_1$ and $i_2\leq j_2$ such that $w_{\spn{i_1,j_1}}=w_{\spn{i_2,j_2}}$.
We write
$(w,\alpha)\models\phi$ to denote that $\alpha$ is a satisfying assignment for $\phi$ on~$\struct{w}$. 
\end{definition}
\begin{example}
	The $\FOstr$-formula
	$ 
	\ex{x}{\strEqPred(\minConst,x,x,\maxConst)}
	$
	defines $\{ww\mid w\in\Sigma^*\}$.
\end{example} 
 Technically, we do not need the symbols $\minConst$, $\maxConst$, or $\nextPred$, as these can be directly expressed in $\FOord$. 
But  these constants allows us to better preserve the structural similarities when converting between various fragments of $\FC$ and $\FOstr$.  

When comparing $\FC$ to $\FOstr$, we need to address that one operates on words and the other on positions.
We can handle this in a way that is similar to the situation between $\FC$ and spanners; and this can be used to show that there are polynomial time conversions between $\FC$ and $\FOstr$ that preserve the properties existential and existential-positive, and only marginally increase the width of the formulas. 
This is discussed in more detail in \cref{app:lem:FClanguagesFOh}.

In fact, these transformations show that one could choose $\FOstr$ over $\FC$ as a logic for words (or for spanners, if one extends $\FOstr$ with regular constraints or generalizes it to $\MSO$ with $\strEqPred$). This is a valid choice, if one prefers writing $\ex{x_1,\dots,x_6}{\bigl(\lettPred{\mathtt{p}}(x_1)\land\lettPred{\mathtt{a}}(x_2)\land\lettPred{\mathtt{p}}(x_3)\land\lettPred{\mathtt{a}}(x_4)\land\lettPred{\mathtt{y}}(x_5)\land\lettPred{\mathtt{a}}(x_6)\land\bigland_{i=1}^5\nextPred(x_i,x_{i+1})\bigr)}$ over $\ex{x}{x\weqeq\mathtt{papaya}}$ or  if one wants to express  $x\weqeq yz$ as $\ex{x^m}{\bigl(\strEqPred(x^o,x^m,y^o,y^c)\land\strEqPred(x^m,x^c,z^o,z^c)\bigr)}$ instead.

Details on these conversions (and the required definitions) can be found in the full version of this paper. 
For the sake of finding an inexpressibility result, we only require the following.

\begin{lemma}[label=lem:FClanguagesFO,restate=restateLemFClanguagesFO]
A language is definable in $\FC$ if and only if it is definable in $\FOstr$.
\end{lemma}

\subparagraph*{Proving inexpressibility}
Lemma~\ref{lem:FClanguagesFO} allows us to  use Feferman-Vaught theorem, at least when considering languages that are restricted enough.
\begin{lemma}[label=lem:anbn,restate=restateLemAnBn]
	There is no $\FC$-formula that defines   $\{\ta^n\tb^n\mid n\geq 1\}$.
\end{lemma}
Moreover, we can show that regular constraints offer no help for defining this language.
\begin{theorem}[label=thm:equalLength,restate=restateThmEqualLength]
	 $\FCreg$ cannot express the equal length relation $|x|=|y|$.
\end{theorem}
As $\FCreg$ has the same expressive power as $\RGXcored$, this is the first inexpressible result for $\RGXcored$ on non-unary alphabets.
The proof has two parts, which both rely on the limited structure of the language  $\ta^n\tb^n$.
One part is using Lemma~\ref{lem:anbn}, wich applies the Feferman-Vaught theorem. 
The other is using Lemma~\ref{lem:simple} to eliminate the regular constraints, which 
 is based on combinatorics on words.
The authors expect that a more general inexpressibility method for $\FC$ (or even $\FCreg$) would  need to 
combine more advanced techniques from combinatorics on words (like those in~\cite{kar:exp})  with methods from logic.

\section{Conclusions and future work}\label{sec:conclusions}
On words, concatenation is one of the most natural operations.
But as seen for~$\logC$, using concatenation with first-order logic quickly becomes undecidable.
Restricting the universe to a word and all its factors changes the situation drastically.
In contrast to $\logC$, the resulting logic $\FC$ has a meaningful distinction between satisfiability and model checking; 
and the latter is not only decidable, but we can use the structure of the formula to derive upper bounds  in the same way as for $\FO$ over finite structures.
In addition to this, $\FC$ can also replace $\FOord$ as ``base'' logic for characterizing complexity classes.
Hence, while one might certainly make a case against the claim that $\FC$ is \emph{the} finite model version of the theory of concatenation, the results leave little doubt that it is at least \emph{a} valid approach.

$\FC$ also provides an extendable framework for querying and model checking words, in particular for scenarios that rely on expressing that factors  appear multiple times.
If more expressive power is needed, $\FC$ is easily extended with constraints, without affecting the lower bounds on evaluation and model checking.
In particular, we can translate  core and generalized core spanners to $\FCreg$ and then analyze or optimize these formulas with respect to parameters like width.
To a degree, this was also possible the spanner logic $\splog$, but $\FC$ is more elegant, easier to use, and behaves much more like $\FO$ on relational databases.

\subsection*{Future work}
Many fundamental questions remain open, in particular for  model checking and related problems, like evaluation and enumeration.

\subparagraph*{Compilation into tractable fragments}
One promising direction is the compiling of formulas into equivalent formulas of a fragment where these problems can be solved more efficiently.
For example, \cref{thm:mc} shows that bounding the width of the formulas leads to tractable model checking. 
\cref{thm:patToFC} then provides us with a sufficient criterion for formulas that can be rewritten into formulas with a lower width, by decomposing the pattern of word equations.
It is likely that this approach can be further refined by not just rewriting single patterns, but taking the larger formula into account. 
This approach can also be used with other structure parameter for formulas (like acyclicity and bounded tree width), by developing a corresponding variant of \cref{thm:patToFC}.
One example of this is~\cite{fre:split}, which adapts the concept of acyclic conjunctive queries to $\FC$.

A more fundamental question is whether all tractable fragments of $\FC$ can be explained through the criterion of bounded width (or are subset of a larger tractable fragment that is explained through it). 
A good starting point for this line of investigation is the question whether all classes of pattern languages with a polynomial time membership problem can be explained through bounded width.

\subparagraph*{Model checking as parsing}
Another approach -- that is not investigated in the present paper -- is the connection to parsing algorithms.
This is a natural question, as model-checking $\fragExPos{\FC}$-formulas can be understood as a parsing problem (where variables are mapped to factors of the input word). 
Promising starting points for this are parsing algorithms for RCGs (recall \cref{sec:iterationRecursion}) and related grammars, and the extraction grammars from~\cite{pet:gra}.

\subparagraph*{Data structures} 
Model checking algorithms will likely benefit from specialized data structures. 
For example, a naive representation of all factors of a word of length $n$ would contain about $O(n^2)$ elements, and if these are just represented directly as words, this would take $O(n^3)$ memory.
But using data structures like suffix trees and suffix arrays, one can create in time $O(n)$ a data structure that allows us to enumerate all factors with constant delay (see~\cite{fre:split}, which also examines small word equations). 

While these optimizations do not matter if one considers polynomial time efficient enough, it would be very useful to know which fragments can be model-checked in time $O(n^k)$ for small~$k$, or even in sub-quadratic time.

\subparagraph*{Inexpressibility and satisfiability} Our results on inexpressibility also leave many questions open.
Lemma~\ref{lem:anbn} heavily relies on the limited structure of the language.
This is the same situation as in Section~6.1 of~\cite{fre:splog}, which describes an inexpressibility technique for $\fragExPos{\FCreg}$.
Although these two approaches provide us with some means of 
proving inexpressibility, they only cover special cases, and much remains to be done.
It seems likely that a more general method will need to combine approaches from finite model theory (like the Feferman-Vaught theorem that we used for \cref{lem:anbn}) with techniques from combinatorics on words (like those in~\cite{kar:exp} that \cite{fre:splog} uses).
A related problem that is still open is whether $\fragExPos{\FC}$ has the same expressive power as $\fragExPos{\logC}$.

Of particular interest is finding a method to prove inexpressibility in $\FC^k$ for some $k>0$.
This problem relates to the open questions whether there are algorithms that minimize the width of a formula, and for which $k$ the fragment $\FC^{k+1}$ is more expressive than $\FC^{k}$.
The authors conjecture that this holds for all $k\geq 0$, which would contrast with $\FOord$, where the fragment of formulas with width three has the same expressive power as the full logic.
Finally, it remains open whether satisfiability is decidable for $\FC^1$ or $\FC^2$.

\subparagraph*{Beyond FC}
Using $\FC$ as a logic for spanners (and other models, potentially) raises further questions. 
For example, while every tractable fragment of $\FCreg$ maps to a tractable fragment of core spanners (namely, those that are obtained by converting the formulas), there is no guarantee that the obtained fragment is natural. 
Hence, a more detailed investigation into conversions between $\FCreg$ and $\RGXcore$ is justified. 

There are  many other possible directions. For example, one could easily define a second-order version of $\FC$ and adapt various results from \textsf{SO}. Moreover, $\FC$ could be examined from an algebra point of view, or related to rational and regular relations.
\bibliography{theoConFM}
\clearpage
\appendix
\newpage
\section{Appendix for Section~\ref{sec:properties}}
\subsection{Proof of Theorem~\ref{thm:recog}}\label{app:thm:recog}
\restateThmRecog*

\begin{proof}
 This upper bounds are straightforward: For $\fragExPos{\FC}$, we only need to deal with existential quantifiers, and as every quantified variable has to be a factor of the universe variable, these  can be guessed. For every substitution and every word equation, $\sigma\models x\weqeq\alpha$ can be verified in linear time. This results in an  an $\compclass{NP}$-algorithm. For $\FC$, we can represent all quantified variables in polynomial space, and enumerating all possible choices for these still results in a $\compclass{PSPACE}$-algorithm. 

	\proofsubparagraph{PSPACE lower bound} The idea is very similar to  Theorem 6.16 in Libkin~\cite{lib:ele}. 
	We  use  reduction from the QBF-3SAT problem, which is stated as follows: Given a well-formed quantified Boolean formula $\psi=Q_1 v_1\colon \cdots Q_k v_k\colon \psi_C$, where $k\geq 1$, $Q_i\in\{\exists,\forall\}$, and $\psi_C$ is a propositional formula in 3-CNF, decide whether $\psi$ is true. This problem is $\compclass{PSPACE}$-complete, see \eg Garey and Johnson~\cite{gar:com}.
	
	Let $\psi=Q_1 \hat{x}_1\colon \cdots Q_k \hat{x}_k\colon \bigland_{i=1}^m \bigl(\ell_{i,1}\land \ell_{i,2}\land \ell_{i,3}\bigr)$ with $\ell_{i,j}\in\{\hat{x}_l,\neg \hat{x}_l\mid 1\leq l \leq k\}$.
	Choose $\ta\in\Sigma$. The $\FC$-formula that we construct shall represent each propositional variable $\hat{x}_l$ with a variable $x_l$, and shall use $x_l=\ta$ and $x_l=\emptyword$ to represent $\hat{x}_l=1$ and $\hat{x}_l=0$, respectively.
	We define a formula $\phi\in\FC$ by $\phi\df\phi^Q_1$, where
	\begin{align*}
		\phi^Q_i &\df \begin{cases}
			\ex{x_i}{ \phi^Q_{i+1}} & \text{if $Q_i=\exists$,}\\
			\fa{x_i}{\left((x_i\weqeq \ta \lor x_i\weqeq \emptyword)\rightarrow \phi^Q_{i+1}\right)} & \text{if $Q_i=\forall$}
		\end{cases}\\
		\shortintertext{for $1\leq i \leq k$, and}
		\phi^Q_{k+1}&\df \bigland_{j=1}^m (\phi^{\ell}_{j,1}\lor\phi^{\ell}_{j,2}\lor\phi^{\ell}_{j,3}),\\
		\phi^{\ell}_{j,l}&\df \begin{cases}
			x_v\weqeq  \ta & \text{ if $\ell_{i,j}=\hat{x}_v$}\\
			x_v\weqeq  \emptyword & \text{ if $\ell_{i,j}=\neg\hat{x}_v$}
		\end{cases}
	\end{align*}
	for $1\leq j \leq m$ and $l\in\{1,2,3\}$. Clearly, $\tau\models\phi^Q_{k+1}$ if and only if $\tau$ encodes a satisfying assignment of the propositional formula. Moreover, the universal quantifiers $\fa{x_i}{}$ are used such that the only interesting substitutions for $x_i$ are those that map $x_i$ to $\ta$ or to $\emptyword$.
	
	Let $\sigma(\sv)\df\ta$. Then $\sigma\models\phi$ if and only if $\psi$ is true. As $\phi$ can be constructe in polynomial time and QBF-3SAT is $\compclass{PSPACE}$-complete, this means that the model checking problem for  $\FC$ is $\compclass{PSPACE}$-hard.
	
	\proofsubparagraph{NP lower bound} The proof for the $\compclass{PSPACE}$ lower bound can be adapted to a reduction from 3SAT by  choosing only existential quantifiers.
\end{proof}
\subsection{Proof of Theorem~\ref{thm:mc}}\label{app:thm:mc}
\restateThmMC*

\begin{proof}
The proof is an extension of the bottom-up evaluation for the $\FO$-case (see \eg Theorem~4.24 in Flum and Grohe~\cite{flu:par}). 
Let $\phi\in\FC$ and $w\in\Sigma^*$. 
For convenience, let $k\df\formulaWidth(\phi)$ and $n\df |w|$. As every variable in $\phi$ must be mapped to a factor of $w$, this means that we have $O(n^{2})$ possible assignments for each variable. 
  
For every word equation $x\weqeq \alpha$, we know that $\formulaWidth(x\weqeq \alpha)\leq k$. This means that there are $O(n^{2k})$ different $\tau$ that could satisfy $(w,\tau)\models (x\weqeq \alpha)$. We can create a  list of all these $\tau$ in time $O(n^{2k+1})$ by enumerating the $O(n^{2k})$ many possible choices and checking each choice in time $O(n)$.

We can lower the complexity to $O(n^{2k})$ by representing each assignment $\tau(x)$ as a pair of pointers $(i,j)$, where $(i,j)$ determine the beginning and end of $\tau(x)$ in $w$. We then pre-compute a table of all pairs of such pairs that determine the same word. This table contains at most $O(n^2)$ entries and can be computed in time $O(n^3)$.
By starting with $\tau(x)$, this table can then be used to check $(w,\tau)\models(x\weqeq\alpha)$ in time $O(k)$, bringing the complexity of generating the list down to $O(kn^{2k})$.

Of course, this is still only a rough upper bound. For example, if a variable $y$ is the first or last variable of $\alpha$, there are only $O(n)$ possible assignments for $y$; and if $\alpha$ starts or ends with terminals, this restricts the possible choices for $x$.

The lists of results can then be combined as in the relational case, requiring time $O(kn^{2k})$ in each inner node of the parse tree of $\phi$. After computing all these sets, we check whether the list for the root not is non-empty, and return the corresponding result.
As the number of these is bounded by $|\phi|$, we arrive at a total running time of $O(k|\phi|n^{2k})$.
\end{proof}
\subsection{Proof of Theorem~\ref{thm:patToFC}}\label{app:thm:patToFC}
\restateThmPatToFC*
\begin{proof}
We first give a short summary of the definition of tree decompositions, treewidth, and nice tree decompositions (based on Chapter~7 of~\cite{cyg:par}). Readers who are familiar with these are invited to skip over to the actual construction.

\proofsubparagraph{Tree decompositions} A \emph{tree decomposition} of a graph $G= (V,E)$ is a tree $T$ with a function $B$ that maps every node $t$ of $T$ to a subset of $V$ such that:
\begin{enumerate}
	\item for every $i\in V$, there is at least one node $t$ of $T$ such that $i\in B(t)$,
	\item for every edge $(i,j)\in E$,  there is at least one node $t$ of $T$ such that $i,j\in B(t)$,
	\item for every $i\in V$, the set of nodes $t$ of $T$ such that $i\in B(t)$ induces a connected subtree of~$T$.
\end{enumerate}
The \emph{width} of a tree decomposition $(T,B)$ is the size of the largest $B(t)$ minus one. The \emph{treewidth} $\treeWidth{G}$ of $G$ is the minimal possible treewidth over all tree decompositions of $G$. 
A~tree decomposition $(T,B)$ of $G$ is called \emph{nice} if, firstly, $T$ has a root $r$ such that  $B(r)=\emptyset$ and $B(l)=\emptyset$ for every leaf $l$ of $T$, and secondly, every non-leaf node is of one of the following types:
\begin{itemize}
	\item \emph{introduce node (for $i$):} a node $t$ with exactly one child $t'$ such that $B(t)=B(t')\cup\{i\}$ with  $i\notin B(t')$,
	\item \emph{forget node (for $i$):} a node $t$ with exactly one child $t'$ such that $B(t)=B(t')\setdiff\{i\}$ with $i\in B(t')$,
	\item \emph{join node:} a node $t$ with exactly two children $t_1$ and $t_2$ such that $B(t)=B(t_1)=B(t_2)$. 
\end{itemize}
Recall that if a graph has a tree decomposition of width at most $k$, it also has a nice tree decomposition of width a most $k$. Moreover, for every $i\in V$, there is exactly one forget node.

\proofsubparagraph{Construction}
Let $\phi\df \ex{x_1,\ldots,x_m}{y\weqeq \alpha}$ with $\alpha=\alpha_1\cdots\alpha_n$ and $\alpha_1,\ldots,\alpha_n\in (\Sigma\cup\Xi)$. We assume $n\geq 2$, as the other cases are trivial, and $\sv\notin\var(\alpha)$, as this makes the explanation easier (by allowing us to skip comments like ``add a quantifier for the variable, unless it is $\sv$). As $\sv$ does not affect the width, this does not affect the bounds anyway.

Furthermore, it is enough to just consider the case of $\fvar(\phi)=\emptyset$, that is, $\{x_1,\ldots,x_m\}=\{y\}\cup\var(\alpha)$.
This means that we want to construct a formula $\psi$ with $\formulaWidth(\psi)\leq 2\treeWidth{\alpha}+2$.
We can then generalize this to free variables by removing the corresponding existential quantifiers, increasing the width accordingly.

The key idea is to decompose $\alpha$ into a conjunction of smaller equations -- that is, we treat $\alpha$ as a term that uses the binary concatenation function. 
This results in the intermediate formula
\[
\phi'\df \ex{\vec{x},\vec{z}}{\bigl(z_2\weqeq \alpha_1\alpha_2\land y\weqeq z_{n-1}\alpha_n \land \bigland_{i=2}^{n-2} z_{i+1}\weqeq z_i \alpha_{i+1}  \bigr)},
\]
where $\vec{x}\df x_1,\ldots,x_m$, $\vec{z}=z_2,\ldots,z_{n-1}$, and the $z_i$ are new variables.
As intuition, each $z_i$ represents the subpattern $\alpha_1\cdots\alpha_i$, meaning that  a pattern substitution $\sigma$ that satisfies the quantifier-free part of $\phi'$ has $\sigma(z_i)=\sigma(\alpha_1\cdots\alpha_i)$.

Our goal is now to reduce the width of $\phi'$ similar to \cref{example:patCompress} -- that is, by moving the quantifiers inward (and shuffle the word equations around accordingly).
As in the proof of Kolaitis and Vardi~\cite{kol:con} for variable bounded $\FO$, the necessary information follows directly from from the tree decomposition.

Consider a nice tree composition $(T,B)$ of $\patgraph{\alpha}= (V_{\alpha},E^{<}_{\alpha}\cup E^{\mathsf{eq}}_{\alpha})$ (recall \cref{def:patGraph}) that has width $k$. 
We shall rewrite $T$ into the parse tree of our formula $\psi$ by assigning each word equations and quantifiers of $\phi'$ to a node of $t$.  

The first step is to annotate the node of $t$ with variables. 
For each node $t$ of $T$, we consider each $i\in B(t)$.  
First, we annotate $i$ with $z_i$, unless $i=n$, where we use $y$ instead, or $i=1$, where we do nothing.
Then we examine the corresponding pattern position $\alpha_i$.
If
$\alpha_i=x$ for some variable $x\in \Xi$, we annotate the node with $x$.
We do not need to do anything if  $\alpha_i\in \Sigma$, as terminals do not affect the width of the formula.

Next, we place each word equation $\eta$ of $\phi'$ at the highest node in $T$ that is marked with its variables (the number of variables in $\eta$ is two or three, depending on whether the pattern position was a terminal or a variable).
This is always possible, as $E^{<}$ ensures that we can connect ``neighbouring'' $z_i$ and $z_{i+1}$, and $E^{=}$ ensures that the occurrences of each variable in $\alpha$ are connected.
For each $z_i$, we place $\exists{z_i}$ at the forget node for $i$; and $\exists{y}$ is placed at the forget node for $n$.
Finally, for each $x_i$, we place $\exists x_i$ at the highest forget node of any  $j$ that has $\alpha_j=x_i$.

We then obtain $\psi$ by using conjunctions to combine all word equations in a node and to combine children in the tree, keeping the quantifiers where they are.
As $|B(t)|\leq k+1$, each node in $T$ was annotated with at most $2k+2$ variables, which means that $\formulaWidth(\psi)\leq 2k+2$.

If $\alpha$ has bounded treewidth, the nice tree decomposition can be computed in polynomial time, which means that $\psi$ can be constructed in polynomial time.
\end{proof}
\subsection{Proof of Proposition~\ref{prop:FcSat}}\label{app:prop:FcSat}
\restatePropFcSat*
First, note that the proof of \cref{thm:nonrec} (see \cref{app:thm:nonrec}) shows that satisfiability is undeciable for $\FC^4$. 
But that proof is not easily modified to $\FC^3$.
Furhermore, Durnev~\cite{dur:und} shows undecidability of satisfiability for the $\forall\exists^3$-fragment of $\fragExPos{\logC}$; that is, for formulas of the form $\fa{s}{\ex{x,y,z}{\phi}}$, where $\phi$ is positive and quantifier-free. 
But  the formula that is constructed in~\cite{dur:und} is not an $\FC$-formula, as it contains equations of the form $x0\weqeq 0x$.
In principle, one could prove Proposition~\ref{prop:FcSat} by rewriting the proof from~\cite{dur:und} or sketching which changes need to be made. But as the following proof is short enough (and as it is an opportunity to use \textsc{Fractran}), we give an original proof instead.
\begin{proof}
We show the undecidability by providing a reduction from the halting problem for \textsc{Fractran}-programs (introduced by Conway~\cite{con:frac}). A  \textsc{Fractran}-program is a finite sequence $P\df (\frac{n_1}{d_1},\frac{d_2}{n_2},\dots,\frac{d_k}{n_k})$ with $k\geq 1$ and $n_i,d_i\geq 1$. The input (and only memory) is a natural number $m\geq 1$. 

This is interpreted as follows: In each step, we search the list of fractions in the program $P$ from left to right until we find the first fraction $\frac{n_i}{d_i}$ such that the product $m\frac{n_i}{d_i}$ is a natural number. If no such fraction can be found, $P$ terminates. Otherwise, we update $m$ to $m\frac{n_i}{d_i}$ and proceed to the next step.

By reducing the fractions, we can ensure that all $n_i$ and $d_i$ are co-prime. Furthermore, as we are interested in termination, we can exclude cases where $d_i = 1$. 
The halting problem for  \textsc{Fractran} (deciding whether a program $P$ terminates on an input number $n\geq 1$) is undecidable (see Kurtz and Simon~\cite{kur:undec}).

Given $P$ and $m$, our goal is to construct a sentence $\phi\in\FC^3$ that is satisfiable if and only if $P$ terminates on input $m$. Assume that $\Sigma\supseteq\{0,1\}$.
The construction shall ensure that $w\models \phi$  if and only if $w\in 0(1^+0)^+$ is an encoding of an accepting run of $P$ on~$m$. More formally, we will have 
$$w= 0\conc 1^{c_1}\conc 0\conc 1^{c_2}\conc 0\conc 1^{c_3}\conc 0\conc \cdots \conc0\conc 1^{c_{t-1}}\conc 0\conc 1^{c_t}\conc0,$$ 
with $c_j\geq 1$ for all $j$, 
where $c_1=m$, each $c_{i+1}$ is the number that succeeds $c_i$ after applying one step of $P$, and $c_t$ is a number on which $P$ terminates (\ie, $c_t$ is divided by no $d_i$). We first define $\phi_{cod}$ to be the following  $\FC^3$-sentence:
\begin{align}
\phi_{cod}\df      & \ex{x}{\sv\weqeq 0 1^m 0\conc x} \label{eq:f1}\\
\land\: & \ex{x}{\sv\weqeq x\conc 0} \label{eq:f2}\\
\land\: & \lnot\ex{x,y}{\sv\weqeq x\conc 00\conc y} \label{eq:f3}\\
\land\: & \biglor_{a\in \Sigma\setdiff\{0,1\}}\lnot\ex{x,y}{\sv\weqeq x\conc a\conc y} \label{eq:f4}
\end{align}
The parts of the conjunction have the following roles:
\eqref{eq:f1} expresses that $w$ starts with $01^m0$, 
\eqref{eq:f2} states that it ends on $0$, 
\eqref{eq:f3} requires that it $w$ does not contain $00$, and \eqref{eq:f4} forbids all letters other than $0$ and $1$. Hence, these four parts together ensure that $w\in 01^m(01^+0)^*$ holds. Hence, if $w\in\Lang(\phi_{cod})$, we know that $w$ encodes a sequence $c_1,\ldots,c_t\geq 1$ for some $t\geq 1$ with $c_1=m$. 
The next step is defining the following sentence:
\begin{equation*}
\phi_{term}\df \fa{x,y}{\Big(
	\big(\sv\weqeq x\conc 0\conc y \conc 0 \land \lnot\ex{x,z}{y\weqeq x\conc 0\conc  z}\big)
	\rightarrow \lnot \biglor_{i=1}^k \ex{x}{y\weqeq x^{d_i}}
	\Big)}
\end{equation*}
The left side of the implication states that $y$ contains the last block of $1$s in $w$, the right side that the length of $y$ is not divided by any $d_i$. In other words, $\phi_{term}$ expresses that $P$ terminates on $c_t$.
All that remains is defining a formula that expresses that $c_{i+1}$ is the successor of $c_i$ when one step of $P$ is applied. This is the job of the following sentence:
\begin{multline*}
\phi_{step}\df \fa{x,y,z}{
\Big(
\big(x\weqeq 0\conc y\conc 0\conc z \conc 0 \land \lnot\ex{x,z}{y\weqeq x\conc 0 \conc z} \land \lnot\ex{x,y}{z\weqeq x\conc 0 \conc y}\big)\\
\rightarrow
\biglor_{j=1}^k 
\Big(
\ex{x}{(y\weqeq x^{d_j} \land z=x^{n_j})}
\land 
\bigland_{l<j}\lnot\ex{x}{y\weqeq x^{d_l}}
\Big)
\Big)
}
\end{multline*}
This formula expresses that, if $w$ contains $01^{c_i}01^{c_{i+1}}0$, then $c_{i+1}=\frac{n_j}{d_j}c_i$ holds for some $1\leq j\leq k$, and $c_i$ is not divided by any $d_l$ with $l<j$.

We now put the parts together and define $\phi\df\phi_{cod}\land\phi_{step}\land\phi_{term}$. Then $\Lang(\phi)\neq\emptyset$ if and only if $w$ encodes a terminating run of the  \textsc{Fractran}-program $P$ on the input $m$. In other words, $P$ terminates on $m$ if and only if the constructed $\phi\in\FC^3$ is satisfiable. 
As the halting problem for $\textsc{Fractran}$ is undecidable, we conclude that satisfiability for $\FC^3$ is undecidable.
\end{proof}
\subsection{Proof of Theorem~\ref{thm:undecShort}}\label{app:thm:undecShort}
\restateThmUndecShort*
We actually prove the more general following result.
\begin{theorem}\label{thm:undec}
	  For $\phi,\psi \in \fragExPos{\FC^{4}}$, we can decide neither $\Lang(\phi)\subseteq \Lang(\psi)$ (containment), nor $\Lang(\phi)\subseteq \Lang(\psi)$ (equivalence), nor whether $\Lang(\phi)$ is $\Sigma^*$, regular, a pattern language, or expressible in $\FC^0$.
	Furthermore, given $\phi\in\fragExPos{\FC^{4}}$, we cannot compute an equivalent $\psi$ such that $|\psi|$ is minimal.
\end{theorem}
The key idea of the proof (to be presented in \cref{app:thm:undec:proof}) is to adapt the  proof of Theorem~14 of~\cite{fre:ext}. That paper examines \emph{extended regular expressions} with one variable, also called \emph{xregex} with one variable in~\cite{fag:spa} and (due to one anonymous reviewer's strong encouragement) in~\cite{fre:doc}.

These xregex extend classical regular expressions with a variable binding operator $(\alpha)\%x$ and a variable recall operator $x$ for a single variable $x$. 
For example, the xregex $((\ta\ror\tb)^*)\%x\conc x$ creates the language of all $ww$ with $w\in\{\ta,\tb\}^*$, and $(\ta^*)\%x\conc\tb\conc x\conc\tb\conc x$ the language of all $\ta^n \tb \ta^n \tb \ta^n$ with $n\geq 0$. This is as much understanding of syntax and semantics as we need for purpose of this paper\footnote{The interested reader can find much more on xregex can be found in~\cite{fre:ext} and, more recently, \cite{fre:det}. In particular, the latter uses a much nicer form of semantics that is due to~\cite{sch:cha} which was also used in~\cite{fre:splog} to simplify the semantics of regex formulas.}. 

The main effort of the proof of \cref{thm:undec} is translating the xregex into $\fragExPos{\FC^{4}}$ formulas. 
To simplify this translation, we shall introduce  \emph{regex patterns}, a combination of regular expressions and patterns, as notational shorthand for $\FC$ (\cref{sec:regexPatterns}). 

To simplify the presentation, we briefly use regular constraints from Section~\ref{sec:constraints} -- that is, we first define formulas from the more powerful fragment $\FCreg$, and then convert these to $\fragExPos{\FC^{4}}$-formulas. Readers who come directly from the main part of the paper to this part of the appendix should have a quick look at \cref{sec:constraints} before they proceed here (but rest assured that this forward reference does not lead to any cyclical reasoning). 

In fact, to further simplify the actual proof , 
But before that, we observe the following result on the expressive power of $\FC^{0}$:
\begin{lemma}\label{lem:FC10}
A language $L\subseteq\Sigma^*$ is definable in $\FC^{0}$ if and only if $L$ is finite or co-finite.
\end{lemma}
\begin{proof}
The \emph{if}-direction is straightforward. If $L$ is finite, we have $L=\{w_1,\ldots,w_n\}$ for some $n\geq 0$ and can simply define $\phi\df \biglor_{i\in[n]} (\sv\weqeq w_i)$. Likewise, we can define every co-finite language using negation.

For the \emph{only-if}-direction, the proof of \cref{lem:FCtoFO} (see \cref{app:lem:FCtoFO}, ``special cases'') allows us to exclude all cases where the universe variable $\sv$ appears on the right side of a word equation. Hence, we can assume that $\phi\in\FC^{0}$ is defined using only equations of the form $\sv\weqeq w$ with $w\in\Sigma^*$, conjunctions, disjunctions, and negations (without other variables, quantifiers play no role). Hence, $\Lang(\phi)$ is obtained by combining singleton languages $\{w\}$ with intersection, union, and complement. Singleton languages are finite (of course), negation turns a finite into a co-finite language (and vice versa), and unions and intersections preserve the property ``finite or co-finite''. Thus, $\Lang(\phi)$ is finite or co-finite.
\end{proof}
\subsubsection{Regex patterns and regex equations}\label{sec:regexPatterns}
A \emph{regex pattern} is a tuple $\alpha=(\alpha_1,\ldots,\alpha_n)$ with $n\geq 0$, where each $\alpha_i$ is either a variable $x\in\Xi$, or a regular expression. If all regular expressions in $\alpha$ are simple (see \cref{lem:simple}), we say that $\alpha$ is a \emph{simple regex pattern}.

We use regex patterns instead of patterns to extend word equations to \emph{regex equations}: every $x\in \Xi$ and every regex pattern $\alpha$ can be combined into a regex equation $x\weqeq \alpha$, which is simple if $\alpha$ is simple.

To define the semantics of $x\weqeq \alpha$ with $\alpha=(\alpha_1,\ldots,\alpha_n)$ and $n\geq 0$, let $R\subseteq [n]$ be the set of all $i$ such that $\alpha_i$ is regular expression, and let $V\df [n]\setdiff R$ be the set of all $i$ such that $\alpha_i$ is a variable.
We now define the $\fragExPos{\FCcon{\REG}}$-formula (see Section~\ref{sec:constraints})
\[\phi^{x\weqeq\alpha} \df \ex{y_1,\ldots,y_n}{\bigl( x\weqeq y_1 \cdots y_n \land \bigland_{i\in R} \constr{y_i}{\alpha_i} \land\bigland_{i\in V} y_i\weqeq\alpha_i \bigr)},\]
where $\fvar(\phi^{x\weqeq\alpha})=(\{x\}\cup\var(\alpha))\setdiff\{\sv\}$. 
We write regex patterns like patterns that contain regular expressions; see the following example.
\begin{example}
Let $\phi{}\df \ex{x,y}{(\sv\weqeq x\conc \ta\tb^*\ta\conc y)}$. Then $\Lang(\phi)$ is the set of all $w\in\Sigma^*$ that contain a factor $\ta\tb^n\ta$ with $n\geq 0$. Note that the word equation in $\phi$ is simple, as $\ta\tb^*\ta$ -- the only regular expression in $\alpha$ -- is simple.
\end{example}
In fact, every regex formula can be interpreted as a formula of bounded width:
\begin{lemma}\label{lem:regexequation}
For every regex equation $(x\weqeq \alpha)$, we can construct  $\psi\in\fragExPos{\FCcon{\REG}}$ with $\psi\equiv\phi^{x\weqeq\alpha}$ and $\formulaWidth(\psi)=|\fvar(\phi^{x\weqeq\alpha})|+3$. 
If $\alpha$ is simple, then we can have $\psi\in\fragExPos{\FC}$.
\end{lemma}
\begin{proof}
The first part can be achieved by reordering the quantifiers of $\phi^{x\weqeq\alpha}$ in the same way as  in \cref{example:patCompress}. For example, if all $\alpha_i$ are regular expressions and $n$ is even, we can define
\begin{gather*}
\psi\df \ex{y,z_1}{\bigl( x\weqeq yz_1 \land \constr{y}{\alpha_1} \land \hphantom{x}} \\
\hphantom{\ex{z,y_1}{}}	\ex{y,z_2}{\bigl(z_1\weqeq yz_2 \land \constr{y}{\alpha_2} \land\hphantom{x} }\\
\hphantom{\ex{z_1,y,z_2}{}}\ex{y,z_1}{\bigl(z_2\weqeq yz_1 \land \constr{y}{\alpha_3} \land\hphantom{x}}\\
\hphantom{\ex{z_1,y,z_2,y}{}}			\vdots\\			
\hphantom{\ex{z_1,y,z_2,y}{}}		\ex{y,z_2}{\bigl( z_1\weqeq y z_2 \land \constr{y}{\alpha_n}  \bigr)}				
			\cdots \bigr)
		\bigr)         
	\bigr).
\end{gather*}
If $\alpha_i=x_i$ for some variable $x_i\in \Xi$, we can avoid using $y$ in this case and write $z_1 \weqeq  x_i z_2$  instead of $(z_1 \weqeq x_i z_2)\land(y\weqeq\alpha_i)$ if $i$ is even, or the respective other case if $i$ is odd.
In addition to the free variables from $\fvar(\phi_{x,\alpha})$, the width is only increased by the three additional variables $y,z_1,$ and $z_2$.

If the regex pattern $\alpha$ is simple, then all its regular expressions are simple. As shown in the proof of \cref{lem:simple} (see~\cref{app:lem:simple}), we can then replace every constraint $\constr{y}{\alpha_i}$ with an equivalent $\fragExPos{\FC}$-formula $\psi_{\alpha_i}(y)$. Moreover, we can see in that proof that $\formulaWidth(\psi_{\alpha_i})=3$; and as we can reuse the variables $z_1$ and $z_2$ in $\psi_{\alpha_i}$, replacing the constraints in $\psi$ does not increase its width. 
\end{proof}
\subsubsection{Main part of the proof of Theorem~\ref{thm:undec}}\label{app:thm:undec:proof}
\begin{proof}
As stated above, we shall translate  the proof of Theorem~14 in~\cite{fre:ext} to $\fragExPos{\FC^4}$.
The central construction of that proof is as follows: Given a so-called  \emph{extended Turing machine}\footnote{Note that the details of these extended Turing machines do not matter to our proof, as our translations function on a purely syntactical level.} $M$,  define a language $\VALC(M)\subseteq\{0,\#\}^*$ that contains exactly one word for every \emph{valid} computation of $M$ (\ie, an accepting computation on some input). In other words, there is a one-to-one correspondence between $\VALC(M)$ and each word that is accepted by $M$. 
Then construct an xregex $\alpha$ with $\Lang(\alpha)=\INVALC(M)\df \{0,\#\}^*\setdiff\VALC(M)$. 

Our main goal is now to show that this xregex $\alpha$ can be converted into a formula $\phi\in\fragExPos{\FC^{4}}$. That is, given $M$, we construct $\phi$ with $\Lang(\phi)=\Lang(\alpha)=\INVALC(M)$. We first assume that $\Sigma=\{0,\#\}$ and discuss larger alphabets later (as we can only use simple regular expressions, this step is not completely trivial).
 After that, we discuss the undecidability results that  follow from this construction.

\proofsubparagraph{Creating the formula (binary alphabet)}
As shown in~\cite{fre:ext}, given $M$, one can construct a one-variable xregex $\alpha$ with $\Lang(\alpha)=\INVALC(M)$. The construction is rather lengthy; but it is described in a way that allows us to only consider the necessary modifications.

As one might expect, $\alpha$ is obtained by enumerating all possible types of errors that can cause a word to be an element of $\INVALC(M)$. The proof in~\cite{fre:ext} distinguishes two different types of error: \emph{structural errors}, where a word cannot be interpreted as the result of encoding a sequences of configurations of $M$,  the first configuration is not initial, or the last configuration is not accepting; and \emph{behavioral errors}, where we assume that it is an encoding of a sequence of configuration, but at least one configuration in the sequence does not have the right successor.

While structural errors can be handled with a classical regular expression, behavioral errors require the use of variables to handle the tape contents correctly. This makes makes expressing the structural errors straightforward for xregex, but requires considerable effort for~$\FC$. We first deal with structural errors.

\proofsubparagraph{Structural errors} In the encoding that is defined in~\cite{fre:ext}, every configuration of $M$ is encoded as a word from the language
\[L_C\df \{ 00^{t_1}\#00^{t_2}\#00^{a}\#0^q \mid t_1,t_2\geq 0, a\in\{0,1\}, q\in[n] \}\]
where $n$ is the number of states of $M$ (hence, $q$ encodes the current state). Here, $0^{t_1}$ and $0^{t_2}$ are unary encodings of the tape contents to the left and right of the head, and $a$ is the head symbol under the head.  The sequence of configurations of $M$ is then encoded as a word from the language 
\[L_{seq}\df \{ \#\# c_1 \#\# c_2 \#\# \cdots \#\# c_n \#\# \mid n\geq 1, c_i\in L_c \text{ for all }i\in[n]\}.\]
Now define $L_S$ as the subset of $L_{seq}$ where $c_1$ has state $q=1$, $t_1=0$, and $t_2>0$ (meaning initial state and head starting on the left of a  non-empty input), and $c_n$ has symbol $a$ under the head and is in a state $q$ such that $M$ halts. We now say that $w\in\Sigma^*$ has a \emph{structural error} if $w\notin L_V$. As $\VALC(M)\subseteq L_S$ must hold, having $w\notin L_S$ is sufficient for $w\in\INVALC(M)$.

We first define a formula $\phi_{seq}$ for the complement of $L_{seq}$. We define $\phi_{seq}$ using we use simple regex equations.  As these have no free variables, we can use \cref{lem:regexequation} to  interpret $\phi_S$ as formula from $\fragExPos{\FC^{3}}$. We begin with
\[\phi_{seq,1} \df (\sv\weqeq\emptyword) \lor   (\sv\weqeq 0\Sigma^*) \lor (\sv\weqeq \Sigma^*0) \lor (\sv\weqeq \# 0\Sigma^*) \lor (\sv\weqeq \Sigma^* 0\#) \lor (\sv\weqeq \#) \lor (\sv\weqeq \#\#).\]
Then we have $w\notin \Lang(\phi_{seq,1})$ if and only if $w$ is of the form $\#\# \Sigma^* \#\#$. Building on this, let
\begin{multline*}
\phi_{seq,2} \df \phi_{S,1} \lor (\sv \weqeq \Sigma^* \#\#\# \Sigma^*) \lor (\sv\weqeq \Sigma^* \#\# 0^+ \#\# \Sigma^*) \lor (\sv\weqeq \Sigma^* \#\# 0^+ \# 0^+ \#\# \Sigma^*) \\
\lor (\sv\weqeq \Sigma^* \#\# 0^+ \# 0^+ \# 0^+ \#\# \Sigma^*)
\lor (\sv\weqeq \Sigma^* 0 \# 0^+ \# 0^+ \# 0^+ \# 0\Sigma^*).
\end{multline*}
Observe that $L_{seq}$ uses double hashes $\#\#$ to separate encodings of configurations, and single hashes $\#$ to separate the components within an encoded configuration. Now we have $w\notin\Lang(\phi_{seq,2})$ if and only if  $w$ is of the form
$ \#\# (0^+ \# 0^+ \# 0^+ \# 0^+ \#\#)^+.$ Next, let
\[\phi_{seq}\df \phi_{seq,2} \lor (\sv\weqeq \Sigma^* 00^n\#\#\Sigma^*)\lor (\sv\weqeq \Sigma^* 00\#0^+\#\#\Sigma^*). \]
In the encoding, each block of $0$s to the left of a double hash encodes a state. Hence, the first part of $\phi_{seq}$ (after $\phi_{seq,2}$) expresses that there is an encoding of a state $q$ that is not in the state set $[n]$ of $M$. Likewise, the second part expresses that there is a tape symbol $a$ that is not $0$ or $1$. Consequently, we have $w\notin\Lang(\phi_{seq})$ if and only if $w\in L_{seq}$.
In other words, $\phi_S$ defines the complement of $L_{seq}$. To extend this into a $\phi_V$ for the complement of $L_V$, we need to define two types of errors; namely, that the first encoded configuration is not initial, and that the last configuration is not halting.
The first is handled by
\[\phi_{S,1} \df \phi_{seq} \lor (\sv\weqeq \#\#00\Sigma^*) \lor (\sv\weqeq \#\#0^+\#0\#\Sigma^*) \lor (\sv\weqeq \#\#0^+\#0^+\#0^+\#00\Sigma^*)\]
which has cases where the first configuration has $t_1\neq 0$, $t_2=0$, or a $q\neq 1$ (in this order). Finally, let $\overline{H}\subseteq \{0,1\}\times[n]$ be the set of all $(a,q)$ such that $M$ does not halt when reading symbol $a$ in state $q$, and define
\[\phi_{S} \df \phi_{S,1} \lor \biglor_{(a,q)\in \overline{H}} (\sv\weqeq \Sigma^* \#00^{a}\#0^q\#\#),\]
which expresses that $M$ would not halt on the last configuration in the sequence. Now we have $w\in\Lang(\phi_S)$ if and only if $w\notin L_S$; which means that $\phi_S$ describes exactly the words that have a structural error. Recall that we can interpret $\phi_S$ as a formula from $\fragExPos{\FC^{3}}$.
\proofsubparagraph{Behavioral errors and combining the parts}
For these behavioral errors, first note that Section~3.3 of~\cite{fre:doc} explains that the xregex for $\INVALC(M)$ from the proof in~\cite{fre:ext} have no stars over the variable operators. Moreover, they  can be rewritten into a union of xregexes that have no disjunctions over the variable operators (these are called \emph{regex paths} in~\cite{fre:doc}), simply by factoring out the disjunctions. But in our terminology, these regex paths can be viewed as sentences of the form 
\[\psi=\ex{x}{\bigl(\sv\weqeq \alpha \land \constr{x}{\beta}   \bigr)},\]
where $\alpha$ is a regex pattern that has $x$ as only variable (all other positions are regular expressions) and $\beta$ is a regular expression. Moreover, one can verify by going through all the cases in the definitions of the behavioral errors in~\cite{fre:ext} that in every case, both the regex pattern $\alpha$  and the regular expression $\beta$ are simple. Hence, we can apply \cref{lem:regexequation} and interpret each $\psi$ as a formula from  $\fragExPos{\FC^{4}}$. Then we define the sentence $\phi_B$ as the disjunction of all these $\psi$, thus describing all behavioral errors.
We then define $\phi\df \phi_S\lor\phi_B$ and have $\Lang(\phi)=\INVALC(M)$ with $\phi\in\fragExPos{\FC^{4}}$.

\proofsubparagraph*{Larger alphabets} For larger alphabets, we need to address the problem that simple regular expressions can only express $\Sigma^*$, but not $A^*$ for $A\subset \Sigma$ with $|A|\geq 2$ (this is not expressible in $\fragExPos{\FC}$, as it is not even expressible in $\fragExPos{\logC}$, see Example~23 in~\cite{kar:exp} (together with our  Lemma~\ref{lem:FCtoC}). 
Hence, while $0^*$ is not problematic, $\{0,\#\}^*$ is not expressible. Luckily, any word that contains some letter from $\Sigma\setdiff\{0,\#\}$ is invalid anyway. The errors that were described by formulas with regex patterns that contain $\Sigma^*$ still describe the errors they described before; and they also describe new ones. We extend $\phi_{S}$ with an additional disjunction $\biglor_{a\in\Sigma\setdiff\{0,\#\}} (\sv\weqeq\Sigma^* a\Sigma^*)$ to catch all words that consist only of the new letters. But no other changes are required.

\proofsubparagraph{Undecidable problems} As shown in~Lemma~10 of~\cite{fre:ext}, the pecularities of extended Turing machines that are used in the construction do not affect the ``usual'' undecidability properties that one expects from Turing machines. In particular, we have that, given an extended Turing machine $M$, the question whether
\begin{enumerate}
	\item\label{undec1} $M$ accepts \emph{at least one} input is semi-decidable but not co-semi-decidable, and
	\item\label{undec2}  $M$ accepts \emph{finitely many} inputs is neither semi-decidable, nor co-semi-decidable.
\end{enumerate}
Given $M$, we can construct $\phi\in \fragExPos{\FC^{4}}$ with $\Lang(\phi)=\INVALC(M)$. Hence, the following questions are undecidable:
\begin{itemize}
	\item $\Lang(\phi)\stackrel{?}{=}\Sigma^*$ is not semi-decidable, as we have $\INVALC(M)=\Sigma^*$ if and only if $\VALC(M)=\emptyset$.
	This also gives us undecidability of containment and equivalence.
	\item ``Is $\Lang(\phi)$ regular?'' is neither semi-decidable, nor co-semidecidable. 
	As shown in Lemma~13 of~\cite{fre:ext}, we have that $\INVALC(M)$ is regular if and only if it is co-finite, which holds if and only if $\VALC(M)$ is finite.
	\item ``Is there a pattern $\alpha$ with $\Lang(\alpha)=\Lang(\phi)$?'' is not semi-decidable. We shall prove this by showing that such an $\alpha$ exists if and only if $\Lang(\phi)=\Sigma^*$. Assume there is an $\alpha$ with $\Lang(\alpha)=\INVALC(M)$. As $\INVALC(M)$ contains the words $0$ and $\#$, we know that $\alpha$ cannot contain any terminals (as these would occur in all words in the pattern language). This means that there must be a variable $x$ that occurs exactly once in $\alpha$ (otherwise, we could  generate neither $0$ nor $\#$). Hence, $\Lang(\alpha)=\Sigma^*$, as we can generate every $w\in\Sigma^*$ by defining $\sigma(x)\df w$ and $\sigma(y)\df\emptyword$ for all other variables. 
	\item ``Is $\Lang(\phi)$ expressible in $\FC^{0}$?'' is neither semi-decidable, nor co-semi-decidable. By \cref{lem:FC10}, the languages that are expressible in $\FC^{0}$ are finite or co-finite. $\INVALC(M)$ cannot be finite, and it is co-finite if and only if $\VALC(M)$ is finite. 
\end{itemize}
The non-existence of a computable minimization function also follows from the undecidability of the question whether $\Lang(\phi)=\Sigma^*$, using the same argument as for Theorem~4.9 in~\cite{fre:doc}:
Every reasonable definition of the length of the formula will ensure that there are only finitely many $\phi$ such that $|\phi|$ is minimal and $\Lang(\phi)=\Sigma^*$. Thus, the set of these minimal representations is finite and thereby decidable. We could then decide $\Lang(\phi)\stackrel{?}{=}\Sigma^*$ by applying the minimization algorithm to $\phi$ and checking whether the result is in the finite set.
\end{proof}
\subsection{Proof of Theorem~\ref{thm:nonrec}}\label{app:thm:nonrec}
\restateThmNonrec*

\begin{proof}	
Most of our reasoning relies on the undecidabilities that we established in \cref{thm:undec}. Like~\cite{fre:doc}, we use a meta-theorem by Hartmanis~\cite{har:god} 	that basically states that for two systems of representations $A$ and $B$ such that given a representation $r\in B$, it is not co-semi-decidable whether $r$ has an equivalent representation in $A$, there is a  non-recursive tradeoff from $B$ to~$A$. See Kutrib~\cite{kut:phe} for details and background, and the proof of Theorem~4.10 for a detailed execution of the reasoning behind that meta-theorem.

Hence, \cref{thm:undec} gives us non-recursive tradeoffs from $\fragExPos{\FC^{4}}$ to $\FC^{0}$ and all representations of regular languages (regular expressions, DFAs, NFAs, etc). Note that the lower bound for the trade-off to patterns remains open, as we have only established that the corresponding problem is not semi-decidable.

Regarding the tradeoffs from $\FC^{4}$, we first observe that the non-recursive tradeoff to $\fragExPos{\FC}$ follows directly is analogous to the proof of Theorem~4.11 in~\cite{fre:doc}, which demonstrates a non-recursive tradeoff from $\RGXcored$ to $\RGXcore$. That proof relies on the same construction for $\INVALC(M)$ as Theorem~\ref{thm:undec}; and we have established that regular constraint are not required for that.

For the remaining tradeoffs, we make use of the fact that we can now use negations. This allows us to adapt the proof of \cref{thm:undec} to obtain more undecidability results.
Given~$M$, the proof of \ref{thm:undec} allows us to construct $\phi\in\fragExPos{\FC^{4}}$ with $\Lang(\phi)=\INVALC(M)$. Hence, we have $\neg\phi\in\FC^{4}$ and $\Lang(\neg\phi)=\VALC(M)$.  

Next, observe that although is not directly shown in~\cite{fre:ext}, it follows directly by using the same methods that given $M$, it is  neither semi-decidable, nor co-semi-decidable whether  $M$ accepts \emph{exactly one} input. 

Hence, given $\psi\in\FC^{4}$, the question whether there is a word $w\in\Sigma^*$ with $\Lang(\psi)=\{w\}$ is neither semi-decidable, nor co-semi-decidable. By invoking Hartmanis' meta-theorem, we obtain the non-recursive tradeoffs from $\FC^{4}$ to pattern languages.

Furthermore, observe that every pattern language $\Lang(\alpha)$ is either an infinite language (if~$\alpha$ contains at least one variable) or a singleton language $\{w\}$ (if~$\alpha$ contains no variables; \ie, $\alpha=w$ for some $w\in\Sigma^*$). Hence, this gives us non-recursive tradeoffs to pattern languages as well. 

This raises the question whether the non-recursive tradeoff also exists if we only consider pattern languages with variables (after all, focusing on the special case of singleton languages might be considered a form of cheating).

Although we leave the case for $\FC^{4}$ open, we can show non-recursive tradeoff from $\FC^{5}$ to patterns with variables. Given $M$, we can define $\psi\in\FC^{5}$ with 
\[
\Lang(\psi)\df \VALC(M)\conc 0\#^3 0 \conc \Sigma^*,
\]
by defining 
\[
\psi\df \ex{x,y}{\bigl( (\sv\weqeq x\conc 0\#^30\conc y ) \land\neg\hat{\phi}(x) \bigr)},
\]
where $\hat{\phi}$ is obtained from the $\phi$ that is constructed from $M$ as in the proof of \cref{thm:undec} by replacing all occurrences of $\sv$ with a new variable $x$.

Now we claim that a pattern $\alpha$ with $\Lang(\alpha)=\Lang(\phi)$ exists if and only if $\VALC(M)$ contains exactly one element. The \emph{if}-direction is clear. Hence, assume such an $\alpha$ exists. 
As pattern languages are always either singleton languages or infinite, we know that $\VALC(M)\neq \emptyset$. By definition of $\Lang(\psi)$, this means that $\Lang(\alpha)$ is infinite, which means that $\alpha$ contains at least one variable.
 
Moreover, as no word in $\VALC(M)$ contains $\#^3$ as a factor, we know that every $w\in\Lang(\alpha)$ has a unique factorization $w = u\cdot 0\#^30 v$ with $u\in\VALC(M)$ and $v\in\Sigma^*$.
We now consider the uniquely defined factorization
\[
\alpha = u_0 x_1 u_1 \cdots x_n u_n
\]
for some $n\geq 1$, with $u_0,\ldots,u_n\in \Sigma^*$ and $x_1,\ldots,x_n\in \Xi$. Now assume that $u_0$ does not have a prefix from the language $\VALC(M)\cdot 0\#^30$, and define a pattern substitution $\sigma$ with $\sigma(x_1)\df \#^4$. Then $\sigma(\alpha)$ has a prefix of the form $\sigma(u_0\cdot x_1)=u_0\cdot \#^4$. 
But as $\#^4$ is not factor of any word in $\VALC(M)$, this means that $\sigma(\alpha)$ does not have a factorization $w = u\cdot 0\#^30 v$ with $u\in\VALC(M)$ and $v\in\Sigma^*$, as the $\#^4$ would need to occur in the $v$, which would lead us to the conclusion that $u_0$ has a prefix from $\VALC(M)\cdot 0\#^30$ and contradict our assumption that this is not the case.

Hence, we now consider the case that $u_0$ has a prefix from the language $\VALC(M)\cdot 0\#^30$. As $M$ cannot continue its computation after stopping, we have $|\VALC(M)|=1$.

Hence, $\Lang(\psi)$ can be expressed with a pattern with variables if and only if $M$ accepts exactly one input. This means that this expressibility is neither semi-decidable nor co-semi-decidable; the latter allows us to use Hartmanis' meta-theorem to conclude non-recursive tradeoffs from $\FC^{5}$ to patterns with at least one variable.
\end{proof}
\subsection{Definitions and results for Section~\ref{sec:iterationRecursion}}\label{app:def:tcfp}
In this section, we also consider the \emph{data complexity} of model checking and evaluation problems.
In contrast to the \emph{combined complexity}, where the formula and the pattern substitution are both part of the input, the data complexity fixes the formula and considers only the substitution as part of the input.
We are going to rely on the following result and its proof:
\begin{proposition}\label{prop:dataComp}
The data complexity of the evaluation problem for $\FC$ is in $\compclass{L}$.
\end{proposition}
\begin{proof}
Fix a formula $\phi\in\FC$. Given a pattern substituation $\sigma$ for $\phi$ with $w\df\sigma(\sv)$, we can decide $(w,\sigma)\models\phi$ in logarithmic space by constructing a deterministic two way automaton $M_{\phi}$ that has a finite number of read-only input heads that do not move outside the input $w$ (see \eg 
Kozen~\cite{koz:the}, Lecture~5). These heads can act as pointers to positions in $w$.

We construct $M_{\phi}$ recursively along the definition of $\phi$. In every step, each $x\in\fvar(\phi)$ is represented by two pointers, the start and end of an occurrence of $\sigma(x)$. 

If $\phi$ is a word equation $x\weqeq \alpha$, we need to check whether $\sigma(x)=\sigma(\alpha)$. This can be done inside the part of $w$ that represents $\sigma(x)$, by processing all positions $\alpha_i\in(\Xi\cup\Sigma)$ of the pattern $\alpha$. 
This needs a few other pointers -- \eg, for the location in $\sigma(x)$, the position $i$ in $\alpha$, the location inside of $\sigma(\alpha_i)$, but it is finite.

If $\phi$ is a conjunction, disjunction, or negation, we just use the machines of the corresponding formulas as subroutines, processing their result(s) accordingly.

If $\phi= \ex{x}{\psi}$ or $\phi= \fa{x}{\psi}$, we can enumerate all possible choices for $\sigma(x)$ successively by moving the two pointers for $x$ around (\eg, both pointers start at the very left, then the end pointer moves to the right stepwise; when it reaches the end, it the start pointer moves one step rightward and the end pointer returns to it).
For each choice, $M_{\psi}$ is called as a subroutine. If a factor occurs multiple times in $w$, this assignment will be chosen multiple times, but this is not an issue for logspace complexity.
\end{proof}
We are also going to rely on the following definition:
\begin{definition}
	For $w\in\Sigma^*$, let  $\Sub(w)\df\{u\subword w\}$.
\end{definition}
We also adopt the convention that we can treat tuples as sets, in particular by writing $x\in \vec{y}$ or $\vec{x}\cup\vec{y}$.

\cref{app:lfppfp} considers least and partial fixed points, \cref{app:tcdtc} considers transitive closure operators.
\subsubsection{Fixed points}\label{app:lfppfp}
Our first step towards defining $\FC$ with fixed points is interpreting $\FC$-formulas as functions that map relations on words to relations on words. 
To this end, we extend $\FC$ with a  relation symbol $\relsym$ that represents the input relation. 
In contrast to the  constraints that we define in \cref{sec:constraints}, the relation $R$ for $\relsym$ is not assumed to be fixed. Instead, we define the notion of a \emph{generalized pattern substitution} $\sigma$, that also maps $\relsym$ to a relation $\sigma(\relsym)\subseteq (\Sigma^*)^{\ar(\relsym)}$.  
For an $\ar(\relsym)$-tuple of variables $\vec{x}$, we then have $(w,\sigma)\models\relsym(\vec{x})$ if  $\sigma(\relsym) \subseteq \Sub(w)^{\ar(\relsym)}$ and $\sigma(\vec{x})\in\sigma(\relsym)$. We call the formulas that are extended in this way $\FCcon{\relsym}$-formulas.
\begin{definition}\label{def:fpfunction}
Let $\phi$ be an $\FCcon{\relsym}$-formula and $k\df \ar(\relsym)$.
For every $w\in\Sigma^*$ and every $k$-tuple $\vec{x}$ over $\fvar(\phi)$, we define the function from $k$-ary relations over $\Sub(w)$ to  $k$-ary relations over $\Sub(w)$ by
\[F^{\phi}_{\vec{x},w}(R)\df \{\sigma(\vec{x})\mid (w,\sigma)\models \phi, \sigma(\relsym)=R\}\]
for every $R\subseteq \Sub(\vec{w})^k$. 
We use this to define a sequence of relations by $R_0\df\emptyset$ and $R_{i+1}\df F_{\vec{x},\vec{w}}^{\phi}(R^i)$ for all $i\geq 0$.
\end{definition}
\begin{example}\label{ex:lfpEqualLength}
	Let $\ar(\relsym)=2$, and define the $\FCcon{\relsym}$-formula
	\begin{equation*}
	\phi\avs{x,y}\df (x\weqeq\emptyword \land y\weqeq\emptyword) \lor
	\ex{\hat{x},\hat{y}}{\left(\biglor_{a\in \Sigma}\bigvee_{b\in \Sigma} \bigl(x\weqeq a\cdot\hat{x} \land y\weqeq b\cdot\hat{y}\land \relsym(\hat{x},\hat{y})    \bigr) \right)}. 
	\end{equation*}
	Using a straightforward induction, one can prove that for every $w\in\Sigma^*$, we have that $F^{\phi}_{(x,y),w}$ defines a sequence of relations, where each $R_i$ contains the pairs $(u,v)$ where $u,v\in\Sub(w)$ and $|u|=|v|< i$. In other words, for $i> |w|$, we have that $R_i$ expresses the equal length relation on $\Sub(w)$.
\end{example}

For every set $A$ and every function $f\colon\pow(A)\to\pow(A)$, we say that $S\subseteq A$ is a \emph{fixed point of $f$} if $f(S)=S$. A fixed point $S$  of $f$ is the \emph{least fixed point} if $S\subseteq T$ holds for every fixed point $T$ of $f$. We denote the least fixed point of $f$ by $\lfp(f)$. Using basic fixed point theory, see \eg Ebbinghaus and Flum~\cite{ebb:fin}, we can prove the $\FC$-version of a basic result for $\FO$:
\begin{lemma}\label{lem:lfp}
	Let  $\phi\in \fragExPos{\FCcon{\relsym}}$, let $w\in \Sigma^*$, and let $\vec{x}$ be a $k$-tuple over $\fvar(\phi)$. 		
	Then there exists  $c\leq |w|^{2k}$ such that $R_c= \lfp(F_{\vec{x},w}^{\phi})$.
\end{lemma}
\begin{proof}
	First, observe that $F_{\vec{x},w}^{\phi}$  is a function $F_{\vec{x},w}^{\phi}\colon \pow(S)\to\pow(S)$ for $S\df \Sub(w)^k$. 
	Furthermore, note that $S$ is a finite set with $|S|\leq |w|^{2k}$.
	
	To prove the claim, we use two further notions from fixed point theory: $F_{\vec{x},w}^{\phi}$ is called \emph{monotone} if $A\subseteq B$ implies $F_{\vec{x},w}^{\phi}(A)\subseteq F_{\vec{x},w}^{\phi}(B)$ for all $A,B\subseteq S$. It is \emph{inductive} if $R_i\subseteq R_{i+1}$ for all $i\geq 0$. 
	
	As we are dealing with the existential-positive fragment of $\FCcon{\relsym}$, the function $F_{\vec{x},w}^{\phi}$ is monotone (this can be proven with a straightforward induction). But every monotone function from $\pow(S)$ to $\pow(S)$ is also inductive (see \eg Lemma~8.1.2 in~\cite{ebb:fin}). Hence,
	for $c\df |S|$, the relation $R_c$ is the least fixed point of $F_{\vec{x},w}^{\phi}$ (this holds for every inductive function $\pow(S)\to\pow(S)$, see \eg Lemma~8.1.1 in~\cite{ebb:fin}). Hence, $c$ is polynomial in $|w|$.
\end{proof}
In other words, least fixed points for sequences of relations that are defined by $\FC$-formulas behave in the same way as for $\FO$-formulas. Accordingly, we can extend $\FC$ with least fixed points in the same way that $\FO$ can be extended with least fixed points:
\begin{definition}\label{def:lfp}
	Let $\relsym$ be a relation symbol, $k\df\ar(\relsym)$, and $\phi\in\fragExPos{\FCcon{\relsym}}$. 
	For all $k$-tuples $\vec{x}$ and $\vec{y}$ over $\Xi\setdiff\{\sv\}$,
we define 
$\logIter{\lfp}{\vec{x}}{\relsym}{\phi}{\vec{y}}$
as an \emph{LFP-formula} that has   free variables $(\fvar(\phi)\setdiff \vec{x})\cup\vec{y}$.
	
	For every pattern substitution $\sigma$, we define $\sigma\models \logIter{\lfp}{\vec{x}}{\relsym}{\phi}{\vec{y}}$ if there exists an extended pattern substitution $\tau$ with 
	\begin{enumerate}
		\item $\tau\models \phi$,
		\item $\tau(\sv)=\sigma(\sv)$,
		\item $\tau(\vec{x})=\sigma(\vec{y})$,
		\item $\tau(z)=\sigma(z)$ for all $z\in(\fvar(\phi)\setdiff\vec{x})$, and
		\item $\tau(\relsym)=\lfp\bigl( F_{\vec{x},\sigma(\sv)}^{\phi}\bigr)$.
	\end{enumerate}
	We generalize this multiple relation symbols and to nested fixed point operators, and we use $\LFP{\FC}$ to denote the logic that is obtained by adding these LFP-formulas as base cases to the definition of $\FC$. We extend this to  $\LFPcon{\FC}{\REG}$-formula by also allowing regular constraints.
\end{definition}
\begin{example}
Recall the formula $\phi\avs{x,y}$ from Example~\ref{ex:lfpEqualLength} such that $\lfp(F^{\phi}_{(x,y),w})$ is the equal length relation on every $w\in\Sigma^*$. We use this to define the $\LFPcon{\FC}{\REG}$-sentence
\[\psi\df \ex{x,y}{\bigl( \sv\weqeq xy \land \constr{x}{\ta^*} \land \constr{y}{\tb^*} \land \logIter{\lfp}{(x,y)}{\relsym}{\phi}{(x,y)}\bigr)},\]
which defines the language of all words $\ta^n\tb^n$ with $n\geq 0$.
\end{example}
\begin{lemma}\label{lem:recognizeLFP}
The data complexity of the evaluation problem for $\LFP{\FC}$ is in $\compclass{P}$. 
\end{lemma}
\begin{proof}
Let $\phi\in\LFP{\FC}$. We want to show that for every pattern substitution $\sigma$, we can decide in polynomial time whether $\sigma\models\phi$. 
To do that, we extend the  the proof of Theorem~\ref{thm:recog} (see~\cref{app:thm:recog}) to include LFP-formulas.

To check whether 
$(w,\sigma)\models\logIter{\lfp}{\vec{x}}{\relsym}{\phi}{\vec{y}}$
for some $\phi(\vec{x})$, we  compute $\lfp\left( F_{\vec{x},w}^{\phi}\right)$. As shown in Lemma~\ref{lem:lfp}, this is equivalent to computing $R_{|S|}$, for $S\df \Sub(w)^k$, where $k\df |\vec{x}|$.

This can be done inductively by computing each $R_{i+1}$ from $R_{i}$ with $R_0=\emptyset$. 
In each of these induction steps, we determine $R_{i+1}$ by enumerating all extended substitutions $\tau$ that have $\tau(\relsym)=R_i$ and satisfy the following conditions:  
\begin{itemize}
	\item  $\tau(\sv)=\sigma(\sv)$,
	\item $\tau(\vec{x})=\sigma(\vec{y})$, and
	\item $\tau(z)=\sigma(z)$ for all $z\in(\fvar(\phi)\setdiff\vec{x})$.
\end{itemize}
For each such $\tau$, we check whether $(w,\tau)\models\phi$. 

This check can be done in polynomial time, according to our induction assumption (relation predicates can be evaluated with a lookup if the relation has been computed, and constraints are assumed to be decidable in polynomial time). As $|\vec{x}|=|\fvar(\phi)|$, and as there are at most $|w|^{2}$ different choices for  $\tau(x)$, there are at most $|w|^{2|\fvar(\phi)|}$ different $\tau$. Hence, each level $R_{i+1}$ can be computed using polynomially many checks that each take polynomial time. 

We only need to compute polynomially many levels until reaching the least fixed point $R_{|S|}$. Hence, $\lfp\left( F_{\vec{x},w}^{\phi}\right)$ can be computed in time that is polynomial in $|w|$; and by the induction assumption, $\sigma\models\logIter{\lfp}{\vec{x}}{\relsym}{\phi}{\vec{y}}$ can then be decided in polynomial time.

Apart from that, the proof proceeds as in the proof of \cref{prop:dataComp}; substituting $\compclass{P}$ for~$\compclass{L}$.
\end{proof}
The function $F^{\phi}_{\vec{x},w}$ from Definition~\ref{def:fpfunction} can also be used to define partial fixed points. We define the \emph{partial fixed point} $\pfp(F^{\phi}_{\vec{x},w})$ by $\pfp(F^{\phi}_{\vec{x},w})\df R_i$ if $R_i=R_{i+1}$ holds for some $i\geq 0$, and $\pfp(F^{\phi}_{\vec{x},w})\df \emptyset$ if $R_i\neq R_{i+1}$ holds for all $i\geq 0$. 

We then define PFP-formulas $\logIter{\pfp}{\vec{x}}{\relsym}{\phi}{\vec{y}}$ analogously to LFP-formulas, the only difference being that $\phi$ can be any $\FC$-formula and is not restricted to the existential-positive fragment:
\begin{definition}\label{def:pfp}
	Let $\relsym$ be a relation symbol, $k\df\ar(\relsym)$, and $\phi\in\FCcon{\relsym}$. 
	For all $k$-tuples $\vec{x}$ and $\vec{y}$ over $\Xi\setdiff\{\sv\}$,
	we define 
	$\logIter{\pfp}{\vec{x}}{\relsym}{\phi}{\vec{y}}$
	as a \emph{PFP-formula} with free variables $(\fvar(\phi)\setdiff \vec{x})\cup\vec{y}$.
	
	For every pattern substitution $\sigma$, we define $\sigma\models\logIter{\pfp}{\vec{x}}{\relsym}{\phi}{\vec{y}}$ if there exists an extended pattern substitution $\tau$ with 
		\begin{enumerate}
		\item $\tau\models \phi$,
		\item $\tau(\sv)=\sigma(\sv)$,
		\item $\tau(\vec{x})=\sigma(\vec{y})$,
		\item $\tau(z)=\sigma(z)$ for all $z\in(\fvar(\phi)\setdiff\vec{x})$, and
		\item $\tau(\relsym)=\pfp\bigl( F_{\vec{x},\sigma(\sv)}^{\phi}\bigr)$.
	\end{enumerate}
	We generalize this multiple relation symbols and to nested fixed point operators, and we define  $\PFP{\FC}$ analogously to Definition~\ref{def:lfp}.
\end{definition}
\begin{lemma}\label{lem:recognizePFP}
	The data complexity of the evaluation problem for $\PFP{\FC}$ is in $\compclass{PSPACE}$. 
\end{lemma}
\begin{proof}
This proof proceeds similar to the one of Lemma~\ref{lem:recognizeLFP}, the only difference is the bound on the number of $R_i$ that need to be checked.
To test if $\sigma\models\logIter{\pfp}{\vec{x}}{\relsym}{\phi}{\vec{y}}$, we  need to compute $\pfp\left( F_{\vec{x},\sigma(\sv)}^{\phi}\right)$.
As the underlying universe $\Sub(\sigma(\sv))$ is finite, we only need to enumerate up to $2^{|\sigma(\sv)|^{2k}}$ different $R_i$, where $k\df \ar(\relsym)$. 
Moreover, we only need to keep each $R_i$ in memory until $R_{i+1}$ has been constructed; after that, $R_i$ can be overwritten with $R_{i+2}$.
Each current $R_i$ can be represented as $\tau(\vec{x}) \in \Sub(\sigma(\sv))^k$, which means that the whole procedure can run in $\compclass{PSPACE}$.
Apart from this, the proof proceeds as in Lemma~\ref{lem:recognizeLFP} (and then as in \cref{prop:dataComp}), but using $\compclass{PSPACE}$ instead of $\compclass{P}$ (or $\compclass{L}$).
\end{proof}
We revisit $\LFP{\FC}$ and $\PFP{\FC}$ in \cref{app:thm:capture} for the proof of Theorem~\ref{thm:capture}.
\subsubsection{Transitive closures}\label{app:tcdtc}
For every relation  $R\subseteq (\Sigma^*)^k \times(\Sigma^*)^k$ with $k\geq 1$, we define its \emph{transitive closure}~$\tc(R)$ as the set of all $(r,\hat{r})\subseteq (\Sigma^*)^k \times(\Sigma^*)^k$ for which there exists a sequence $r_1,\ldots,r_n\in  (\Sigma^*)^k$ with $n\geq 1$, $r_1=r$, $r_n=\hat{r}$, and $(r_i,r_{i+1})\in R$ for $1\leq i < n$.

The \emph{deterministic transitive closure} of $R$, written $\dtc(R)$, is defined by adding the additional restriction that for every $1\leq i < n$, there is no $(r_i,s)\in R$ with $s\neq r_{i+1}$.
\begin{definition}\label{def:splDTC}
	Let $\phi\in\FC$ and for $k\geq 1$, choose two $k$-tuples $\vec{x}$ and $\vec{y}$  over $\fvar(\phi)$, and two $k$-tuples $\vec{s}$ and $\vec{t}$ over $\Xi\setdiff\{\sv\}$. 
	Then 
	$\logIter{\tc}{\vec{x}}{\vec{y}}{\phi}{\vec{s},\vec{t}}$
	 is a TC-formula and 
	$\logIter{\dtc}{\vec{x}}{\vec{y}}{\phi}{\vec{s},\vec{t}}$
	 is a DTC-formula. 
	 Both have $\bigl(\fvar(\phi)\setdiff (\vec{x}\cup\vec{y})\bigr)\cup(\vec{s}\cup\vec{t})$ as set of free variables.
	
	For every pattern substitution $\sigma$, we define
	$\sigma\models\logIter{\tc}{\vec{x}}{\vec{y}}{\phi}{\vec{s},\vec{t}}$
	if $(\sigma(\vec{s}),\sigma(\vec{t}))\in\tc(R_{\sigma})$, where $R_{\sigma}$ is the set of all $(\tau(\vec{x}),\tau(\vec{y}))$ such that 
	\begin{enumerate}
		\item $\tau\models\phi$, and
		\item $\tau(z)=\sigma(z)$ for all $z\in\fvar(\phi)\setdiff (\vec{x}\cup\vec{y})$.
	\end{enumerate}
	The analogous definition applies to DTC-formulas, substituting $\dtc(R_\sigma)$ for $\tc(R)_{\sigma}$.
	
	We generalize this to  multiple and nested applications of the closure operators; and we use $\TC{\FC}$ or $\DTC{\FC}$ to denote the logics that are obtained by adding these TC- or DTC-formulas as base cases to the definition of $\FC$. 
\end{definition}
As we do not require that $\vec{x}$ and $\vec{s}$ (or $\vec{y}$ and $\vec{t}$) are distinct, we use $\tcOp{\vec{x}}{\vec{y}}{\phi}$ as shorthand for  $\logIter{\tc}{\vec{x}}{\vec{y}}{\phi}{\vec{x},\vec{y}}$.
We now consider some examples.
\begin{example}
	We define the $\fragExPos{\DTC{\FC}}$-formula 
	\begin{align*}
	\phi&\df \ex{x,y}{\Bigl( (\sv\weqeq y) \land \dtcOp{x}{y}{\psi} \land \bigl(x\weqeq \emptyword \lor \biglor_{a\in\Sigma} y\weqeq a\bigr)   \Bigr)},\\
	\psi\avs{x,y}&\df \biglor_{a\in\Sigma} (x\weqeq  a\cdot y\cdot a).
	\end{align*}
	Then $w\models \phi$ if and only if $w$ is a palindrome over $\Sigma$. The formula $\phi$ expresses that $x$ can be obtained from $y$ by concatenating  one occurrence of some letter $a$ to the left and one to the right of $y$. By applying the transitive closure, we obtain the relation of all $(x,y)$ such that $x = u\cdot y \cdot u^{\mathsf{R}}$, where $u\in\Sigma^*$ and $u^{\mathsf{R}}$ is the reversal of $u$. 
	
	Note that $\psi$ expresses the relation of all $(x,y)$ with $x = aya$ for some $a\in \Sigma$. Hence, each word has exactly one successor in this relation, which means that we can indeed use $\dtcOp{x}{y}{\psi}$. But if we wrote $\dtcOp{y}{x}{\psi}$ instead, then  there could be multiple successors for some $x$ (depending on the content of $\sv$), which means that $\dtc$ would fail.
\end{example}
\begin{example}
	Consider a directed graph $G=(V,E)$ with $V=\{v_1,\ldots,v_n\}\subseteq\{0,1\}^+$ and $n\geq 1$. Define $\mathsf{enc}(E)$ as an encoding of $E$ over $\{0,1,\#,\$\}$ such that $\mathsf{enc}(E)$ contains the factor $\$ v_i \# v_j\$$ if and only if $(v_i,v_j)\in E$. We define the $\fragExPos{\TC{\FC}}$-formula
\[	\phi\avs{x,y}\df 
\tcOp{x}{y}{\ex{z}{\bigl( z \weqeq \$  x \# y  \$ \land \constr{x}{\{0,1\}^+}  \land \constr{y}{\{0,1\}^+}  \bigr)}}.\]
	Then $(\mathsf{enc}(E),\sigma)\models{\phi}$ if and only if $\sigma(x)=v_i$, $\sigma(y)=v_j$, and  $v_j$ can be reached from $v_i$ in one or more steps.	
\end{example}
Next, we examine the data complexity of model checking $\TC{\FC}$ and $\DTC{\FC}$.
\begin{lemma}\label{lem:recognizeTC}
	The data complexity of the evaluation problem  is in $\compclass{NL}$ for $\TC{\FC}$ and in $\compclass{L}$ for $\DTC{\FC}$.
\end{lemma}
\begin{proof}
Again, we extend the proof of \cref{prop:dataComp} (see~\cref{app:thm:recog}) by describing how we evaluate DTC- and TC-formulas. Although some modifications are required in our setting, the basic idea is the same as for $\FO$-formulas with $\dtc$- or $\tc$-operators (see, \eg, Theorem~7.4.1. in~\cite{ebb:fin}). 
Given $\sigma$ and a TC-formula $\logIter{\tc}{\vec{x}}{\vec{y}}{\phi}{\vec{s},\vec{t}}$ (or a DTC-formula like this),  first note that on any given structure $w$, the underlying universe can have up to $n^2$ elements for $n\df |w|$. This means that the closures can create paths up to length $n^{2k}$.

We then construct as logspace-Turing machine $M_0$ for $\phi$ that will be used as a sub-routine. Using one counter per main variable, we can implement a counter from $1$ to $n$. Combining $2k$ of these, we can create a counter that counts up to $n^{2k}$. 
Now we can progress as in the relational case: we invoke $M$ as a subroutine at most $n^{2k}$ times to checking whether there is a path from $\sigma(\vec{x})$ to $\sigma(\vec{y})$. For $\tc$, this involves guessing the next step; for $\dtc$, it involves checking that the successor is unique. Hence, all this can be done in $\compclass{NL}$ and $\compclass{L}$, respectively.

We then  integrate this into larger formulas via \cref{prop:dataComp}; using the fact that $\compclass{L}$ and $\compclass{NL}$ are both closed under complement. Of course, we could reprove this for $\compclass{NL}$ by imitating the proof of the Immerman-Szelepcs\'{e}nyi theorem by means of $\DTC{\FC}$; but this would not provide us with any new insights.
\end{proof}

\subsection{Proof of Theorem~\ref{thm:capture}}\label{app:thm:capture}
$\LFP{\FC}$ and
$\PFP{\FC}$ are defined in \cref{app:lfppfp}. 
$\DTC{\FC}$ and 
$\TC{\FC}$ are defined in \cref{app:tcdtc}.
Both build on \cref{app:def:tcfp}.
\restateThmCapture*
\begin{proof}
We want to show that a language is definable in a logic  $\mathcal{F}$ if and only if it belongs to the complexity class $\cc$, where $\mathcal{F}$ ranges over 
$\DTC{\FC}$, 
$\TC{\FC}$, 
$\LFP{\FC}$, or
$\PFP{\FC}$, 
and $\cc$ over 
 $\compclass{L}$,
 $\compclass{NL}$,
 $\compclass{P}$, and
$\compclass{PSPACE}$, respectively.

We have already established the direction from the logics to the complexity classes, namely in Lemma~\ref{lem:recognizeTC} for $\DTC{\FC}$ and
$\TC{\FC}$, in Lemma~\ref{lem:recognizeLFP} for $\LFP{\FC}$, and in Lemma~\ref{lem:recognizePFP} for $\PFP{\FC}$. These results rely on \cref{prop:dataComp} and, thus, on Lemma~\ref{lem:FCtoFO} (which allows us to convert $\FC$-formulas into $\FO$-formulas).

For the other direction, one might ask whether it is possible to use Lemma~\ref{lem:FOtoFC} (the reverse direction of Lemma~\ref{lem:FCtoFO}). In particular, we have that for each of the extensions of $\FC$, the correspondingly extended version of $\FOord$ captures the complexity class.

But we have the additional goal of showing that an $\fragExPos{\FC}$-formula is enough; and just applying Lemma~\ref{lem:FOtoFC} to the proofs that the authors found in literature would have required considerable hand-waving.

\proofsubparagraph{Capturing \textsf{L} and \textsf{NL} with \textsf{dtc} and \textsf{tc}} As explained by \eg Kozen~\cite{koz:the} (Lecture~5), a language $L$ is in $\compclass{L}$ (or in $\compclass{NL}$) if and only if  there is some $k\geq1$ such that $L$ is accepted by a deterministic (or on-deterministic) finite automaton $A$ that has $k$-many two-way input heads that are read-only and cannot move beyond the input. We assume without loss of generality that $A$ does not read the left end-marker (this can be realized in the finite control).

Let $n$ denote the number of states of $A$. We assume that the state set is $[n]$, that the starting state is 1 and that the accepting state is $n$.
Given such an automaton $A$, our goal is to construct a sentence $\phi$ such that $\Lang(\phi)=\Lang(A)\cap \Sigma^{\geq n}$. The finitely many missing words can then be added with a straightforward disjunction.

In the construction, the universe variable $\sv$ represents the input $w$ of $A$. The head with number $i\in[k]$ is modeled by a variable $x_i$, where its current position $j\in[|w|]$ is represented as $w_{\spn{1,j}}$ (that is, the prefix of $w$ that has length $j-1$). Likewise, the current state $q\in [n]$ is represented by $w_{\spn{1,q}}$. 

Our goal is to define a formula $\psi$ that encodes the successor relation $R$ for $A$. Using $\dtc$ or $\tc$, we can then  simulate the behavior of $A$ on $w$. 
To this end, we define two types of helper formulas. Firstly, for $q\in [n]$, we define a formula $\psi^Q_q\avs{x}$ that expresses ``$x$~represents state $q$'' by having $\sigma\models \psi^Q_q(x)$ if and only if $\sigma(x)=\sigma(\sv)_{\spn{1,q}}$. Let
$\psi^Q_1(x)\df (x\weqeq \emptyword)$ and 
\[\psi^Q_{q+1}(x)\df \ex{\hat{x},z}{\biglor_{a\in \Sigma}\bigl(x\weqeq \hat{x}a \land \sv\weqeq xz\land \psi^Q_{q}(z)\bigr)}\] 
for all $1\leq q < n$.
Next, for each $a\in\Sigma$, we define 
\[\psi^{read}_{a}\avs{x}\df \ex{z}{\sv\weqeq x a z},\] 
which expresses that  ``the letter after the prefix $x$ is $a$''. We also define 
\[\psi^{read}_{\dashv}\avs{x}\df \ex{z}{\sv\weqeq x}\]
 We shall use  these two types of formulas to check the content of  the input heads $i$ (namely, whether head $i$ reads $a\in\Sigma$ or the right end marker $\dashv$). Finally, we define
\[\psi_{succ}\avs{x,y}\df\biglor_{a\in\Sigma} y\weqeq xa\] 
to express that ``$y$ is one letter longer than $x$'', which we shall use for the head movements.
Now we are ready to put the pieces together. For $\vec{a}=(a_1,\ldots,a_k)\in(\Sigma\cup\{\dashv\})^k$ and $q\in [n]$, define 
\[\psi^{\vec{a}}_q \df\psi^Q_q(x_0)\land \bigland_{i=1}^k  \psi^{read}_{a_i}(x_i)\land \psi^{mov}_{q,\vec{a},i}(x_i,y_i),\]
where the head movements are simulated by
\[\psi^{mov}_{q,\vec{a},i}(x_i,y_i) \df\begin{cases}
	\psi_{succ}(x_i,y_i) & \text{ if $A$, when reading $\vec{a}$ in state $q$, moves head $i$ to the right,}\\
	\psi_{succ}(y_i,x_i) & \text{ if $A$, when reading $\vec{a}$ in state $q$, moves head $i$ to the left,}\\
	x_i\weqeq y_i & \text{ if $A$, when reading $\vec{a}$ in state $q$, does not move head $i$.}
\end{cases}\]
This gives us the immediate successor for each combination $\vec{a}$ of input letters (including the end-marker) and each state $q$. To get all possible successors, we combine these into
\[\psi\avs{\vec{x},\vec{y}}\df \biglor_{q\in [n]} \biglor_{\vec{a}\in(\Sigma\cup\{\dashv\})^k} \psi^{\vec{a}}_q\avs{\vec{x},\vec{y}},\]
where $\vec{x}=(x_0,\ldots,x_k)$ and $\vec{y}=(y_0,\ldots,y_k)$. We now define
\begin{align*}
\phi &\df \ex{\vec{x},\vec{y}}{\bigl( \bigland_{i=0}^{k} (x_i \weqeq \emptyword) \land \psi^Q_n(y_0) \land \dtcOp{\vec{x}}{\vec{y}}{\psi} \bigr)} &\text{if $A$ is deterministic, and}\\
\phi &\df \ex{\vec{x},\vec{y}}{\bigl( \bigland_{i=0}^{k} (x_i \weqeq \emptyword) \land \psi^Q_n(y_0) \land \tcOp{\vec{x}}{\vec{y}}{\psi} \bigr)} &\text{if $A$ is non-deterministic.}
\end{align*}
Outside the closure operators, the formula expresses that all $\vec{x}$ encodes the initial position ($x_0$ encodes the starting state 1 and all tapes are at the very left), and that $\vec{y}$ encodes a halting position. If $A$ is non-deterministic, then we can use the $\tc$-operator to obtain the transitive closure of the successor relation on the configurations of $A$ on $w$. 
If $A$ is deterministic, than every configuration has at most one successor, meaning that the $\dtc$-operator also compute the transitive closure, having the same intended effect.

Thus, for all $w\in\Sigma^*$, we have $w\models\phi$ if and only if $w\in\Lang(A)$ and $|w|\geq n$. As mentioned above, the ``missing words'' from the set $W\df \Lang(A)\setdiff \Lang(\phi)$ can be added now by defining a formula $\phi\lor\biglor_{w\in W} \sv\weqeq w$.

We conclude that $\TC{\FC}$ captures $\compclass{NL}$ and that $\DTC{\FC}$ captures $\compclass{L}$. Moreover, note that we used only a single closure operator, and that the formulas are existential-positive (inside and outside of the closure operator).

\proofsubparagraph{Capturing \textsf{P} with \textsf{lfp}}
Hence, we give an outline of the full proof, which takes key-ideas from the proof of Theorem~7.3.4 in~\cite{ebb:fin}. Again, the main challenge is ensuring that the formula is existential-positive.

For every language $L\in\compclass{P}$, there is a Turing machine $M$ that decides $L$ in polynomial time.
We assume that $M$ has one read-only input tape over $\Sigma$ and a read-write  work tape that extends to the right and has a tape alphabet $\Gamma=\{0,\ldots,m\}$ for some $m\geq 1$. For the state set $Q$, we assume $Q=\{0,\ldots,n\}$ for some $n\geq 1$, where $0$ is the initial and $n$ the single accepting state. 
When starting, each head is on the left of its tape (position~0), the machine is in state 1, and each cell of the work tape contains $0$.

As $M$ decides $L$ in polynomial time, there is a natural number $d$ such that on each input $w\in\Sigma$, we have that $M$ terminates after at most $|w|^d$ steps. During this run, $M$ will not visit more than $|w|^d$ tape positions. 

For each $i$  with $0\leq i\leq |w|$, let $w_i$  be the prefix of $w$ that has length $i$. 
For $k\geq 1$, we identify each $k$-tuple $\vec{v}=(v_1,\ldots,v_k)$ with the number $N(\vec{v})\df \sum_{i=1}^k (|v_i||w|^{i-1})$.  
Hence, we can use two $d$-tuples of variables, a tuple $\vec{t}=(t_1,\ldots,t_d)$ that  to encode time stamps and a tuple $\vec{p}=(p_1,\ldots,p_d)$ that encodes positions on work tape (where $0$ is the leftmost position).
The construction will ensure that both tuples will only take on prefixes of $w$ as values.

Our simulation of $M$ will be able to run for $(|w|+1)^{d}-1$ steps; but this does not affect the outcome.

Our goal is to define a relation $R$ that starts with the initial configuration of $M$ on $w$, and then uses the $\lfp$-operator to compute each successor configuration. As the time (and, hence, the space) of $M$ are bounded, this is enough. 

The relation $R$ with arity $2d+2$ shall represent the configuration of $M$ on $w$ in step $N(\vec{t})$ as follows:
\begin{itemize}
\item $(\vec{t}, w_0, w_{q}, \vec{\emptyword})$ to denote that $M$ is in state $q\in Q$,
\item $(\vec{t}, w_1, w_{i-1}, \vec{\emptyword})$ to denote that the input head is at position $i$ with $0\leq i <|w|$,
\item $(\vec{t}, w_2, \emptyword, \vec{p})$ to denote that the working head is at position $N(\vec{p})$,
\item $(\vec{t}, w_3, w_\gamma,\vec{p})$ to denote that the working tape contains $\gamma\in\Gamma$  at position $i$,
\end{itemize}
where $\vec{\emptyword}$ is shorthand for the $d$-tuple that has $\emptyword$ on all positions.

Like in the case for transitive closures, the constructed formula will only be correct for $w$ of sufficient length; but again, the finitely many exceptions can be added later. In particular, we assume that $|w|\geq c$ for $c \df \max\{3,|\Gamma|,|Q|\}$. The only $w_i$ that we refer to explicitly through their number $i$ are $w_0$ to $w_3$ for the first component of $R$, $w_q$ with $q\in Q$, and $w_\gamma$ with $\gamma\in\Gamma$. For each one of these, $|w|\geq c$ guarantees that the input $w$ is large enough to encode them. 

As we encode various things in these prefixes $w_i$, it is helpful to define a successor relation 
\[\psi_{succ}\avs{x,y}\df \biglor_{a\in\Sigma} (y\weqeq x a) \land \ex{z}{(\sv\weqeq y\conc z)}\]
which expresses that $x$ and $y$ are prefixes of $\sv$, and $y$ is one letter longer than $x$.
We first use  this  in the shorthand formulas 
$\psi^{pre}_i(x)$ for $0\leq i \leq c$ to express the prefixes $w_i$. Let
$\psi^{pre}_0\avs{x}\df (x\weqeq \emptyword)$ and 
$
\psi^{pre}_{i+1}\avs{x}\df \ex{y}{\psi^{succ}(y,x)}
$
for $0\leq i < c$. 

We are now ready to define the formula $\psi_{init}(x_1,\ldots,x_{2d+2})$, which expresses the initial configuration:
\begin{multline*}
\psi_{init}(x_1,\ldots,x_{2d+2}) \df \\
\bigl(
\bigland_{i=1}^{d} x_i\weqeq \emptyword \land (\psi^{pre}_0(x_{d+1})\lor \psi^{pre}_1(x_{d+1})\lor \psi^{pre}_2(x_{d+1})) \land \bigland_{i=d+2}^{2d+2} x_i\weqeq \emptyword 
\bigr)\\
\lor \bigl(\bigland_{i=1}^{d} x_i\weqeq\emptyword \land \psi^{pre}_3(x_{d+1})\land (x_{d+2}\weqeq \emptyword) \land \bigland_{j=d+3}^{2d+2} \ex{y}{\sv\weqeq x_jy}\bigr).
\end{multline*}
The first part of the formula ensures that the machine starts in state 0, that each head is at position 0. The second part ensures that all tape cells are set to the blank symbol $0$. Note that the tape position $\vec{p}$ is stored in the last $d$ components of the tuple (\ie, from $d+3$ to $2d+2$). To get all possible $\vec{p}$, these variables can take on any prefix $w_i$ of $w$.

To describe the successor of a time stamp or the movement of the working head, we also define a relation $\psi_{d}^{succ}(\vec{x},\vec{y})$ for $d$-tuples $\vec{x}$ and $\vec{y}$ such that $\sigma\models\psi_{d}^{succ}$ if and only if every component of $\sigma(\vec{x})$ and $\sigma(\vec{y})$ is a prefix of $w$, and $N(\sigma(\vec{y})) = N(\sigma(\vec{x}))+1$. The basic idea is as for $\psi_{succ}$, but extending it to multiple digits by taking into account all cases where a carry might happen (this is straightforward, but rather tedious). 
Using this idea and  the proper prefix relation~$\ppref$ from Example~\ref{example:usefulRelations}, we also construct an existential-positive formula  $\psi_{d}^{\neq}(\vec{x},\vec{y})$ that expresses $N(\sigma(\vec{x})) \neq N(\sigma(\vec{y}))$, similar to how we expressed $\neq$ in that example.

We also define formulas $\psi_a(x)\df \ex{y}{\sv\weqeq x\conc a \conc y}$ for every $a\in\Sigma$. If $x$ represents the position of the input head, $\psi_a(x)$ expresses that this head is reading the letter $a$.

This is now all that we need to describe the successor relation $R$ on configurations of $M$. 
We define an LFP-formula 
\[\psi\df \logIter{\lfp}{\vec{x}}{\relsym}{(\psi_{init}\lor\psi_{next})}{\vec{x}},\]
where $\vec{x}\df(x_1,\ldots,x_{2d+2})$ and the  $\fragExPos{\FCcon{\relsym}}$-formula $\psi_{next}$ is constructed as follows:
\begin{itemize}
	\item Using existential quantifiers, we retrieve a time stamp $\vec{t}$ from $R$, and the uniquely defined state $q$, input head position $i$, working head position $\vec{p}$, and working head content $\gamma$ for $\vec{p}$  for this time stamp $\vec{t}$. 
	\item If $\vec{t}=w^d$, nothing needs to be done. Hence, we can assume that this is not the case.
	\item As $M$ is a deterministic Turing machine, the combination of state, current input symbol, and current tape symbol uniquely determine a combination of head movements and working tape action. Which of these applies can be determined by a big disjunction over all combinations of applying $\psi^{pre}$ to the state and the working tape symbol, and $\psi_a$ to the input symbol. For each of these cases, we create a sub-formula that describes head movements and the tape action in the time stamp $\vec{t}'$ with $N(\vec{t}')=N(\vec{t})+1$. We shall store $\vec{t}'$ in the free variables $x_1$ to $x_d$ of $\psi_{next}$.
	\item The sub-formula then has a disjunction over the four possible choices for $x_{d+1}$ (namely, for $\psi^{pre}_1(x_{d+1})$  to $\psi^{pre}_4(x_{d+1})$.
	\item For $\psi^{pre}_1$, the next state, we simply ensure that the correct successor state is stored in $x_{d+2}$, and set all remaining variables to $\emptyword$.
	\item For $\psi^{pre}_2$, the input head position, we use use $\psi_{succ}$ to pick position $i+1$ or $i-1$ if the head moves, or just use the same position.
	\item For $\psi^{pre}_3$, the working head position, we use $\psi^{succ}_d$ analogously.
	\item For $\psi^{pre}_4$, the working tape contents, we distinguish whether the cell is affected by the tape operation or not; that is, whether the cell is at position $\vec{p}$ or not. If it is at that position, we return the new cell content. If not (which can be tested with $\psi^{\neq}_d$), we retrieve the cell content for time stamp $\vec{t}$ from $\relsym$ using existential quantifiers and return it unchanged.
\end{itemize}
Now, $\psi$ computes the relation of all encodings of configurations that $M$ reaches on input $w$. All that remains is checking for the existence of an accepting configuration. We define 
$$
\phi \df \ex{\vec{x}}{\bigl( \psi^{pre}_0(x_{d+1}) \land \psi^{pre}_n(x_{d+2}) \land    \psi(\vec{x})\bigr)}
$$
for $\vec{x}=(x_1,\ldots,x_{2d+2})$. Then we have $w\models\phi$ if and only if $(w,\sigma)\models \psi$ for some $\sigma$ such that $\sigma(\vec{x})$ contains the encoding of a configuration that reaches the accepting state $n$.
Hence, $\Lang(\phi)=L$.

\proofsubparagraph{Capturing \textsf{PSPACE} with \textsf{pfp}} We can show this by modifying the $\lfp$-construction: Instead of using time stamps, each stage of the relation in $\relsym$ only keeps the most recent configuration and uses it to construct the next. As $L$ is decidable in $\compclass{PSPACE}$, this can be done using the tuple $\vec{p}$. 
As we have already established that the $\lfp$-construction is possible with an existential-positive formula, this modification is straightforward.
\end{proof}
\subsection{Proof of Theorem~\ref{thm:datasplog}}\label{app:thm:datasplog}
\restatethmDatasplog*
\begin{proof}
We can directly rewrite every $\DataSplog$-program into an equivalent $\LFP{\fragExPos{\FC}}$-formula. By Theorem~\ref{thm:capture}, these are in $\compclass{P}$. 
	
For the other direction, we know from the proof  of the $\lfp$-case of Theorem~\ref{thm:capture} that every language in $\compclass{P}$ is recognized by a formula from $\LFP{\fragExPos{\FC}}$ that consist of existential quantifiers over  a single $\lfp$-operator. After transforming the underlying formula into a union of conjunctive queries (using the same rules as for relational logic), we immediately obtain an equivalent $\DataSplog$-program.
\end{proof}
\section{Appendix for Section~\ref{sec:logicForSpanners}}
\subsection{Proof of Lemma~\ref{lem:FCtoC}}\label{app:lem:FCtoC}
\restateLemFCtoC*

\begin{proof}
	Given $\phi\in \FC$, we construct $\psi\in\logC$ such that the latter's syntax simulates the $\FC$-semantics. This can be done by adding proper guards for all variables of $\phi$.	
	We show this using an induction along the definition of $\FC$. 

	In each of the steps, it is shall be easy to see that $|\psi|\in O(|\phi|\conc\formulaWidth(\phi))$, and that $\psi$ can be constructed in time that is proportional to its length.
	The construction also introduces neither new  universal quantifiers, nor new negations. Hence, the resulting formula is existential or existential-positive if and only the original formula had this property.
	Another invariant of the induction is that the constructed $\psi$ can be interpreted as a $\logC$- and as an $\FC$-formula without changing its meaning. That is, $\sigma\models\psi$ shall hold under $\logC$-semantics if and only if it holds under $\FC$-semantics.
	
	\proofsubparagraph{Word equations} If $\phi$ is of the form $x\weqeq\eta$, we distinguish two cases. Firstly, consider the case where $x=\sv$. Then $\sigma(\sv)=\sigma(\eta)$ implies $\sigma(y)\subword \sigma(\sv)$ for all $y\in\var(\eta)$, which means we can define $\psi\df\phi$ (only that we now treat $\psi$ as a $\logC$-formula instead of an $\FC$-formula).	
	
	We need a little more effort in the second case, namely if $x\neq \sv$. Here, we define
\[
	\psi\df  
\ex{p,s}{	\bigl(\sv\weqeq p\conc x\conc s \land \sv \weqeq  p\conc \eta \conc s	\bigr)}
.\]
This embeds the equation $x\weqeq \eta$ in $\sv$. Under both semantics, $\sigma\models\psi$ if and only if $\sigma(x)\subword\sigma(\sv)$ and $\sigma(x)=\sigma(\eta)$. This also implies $\sigma(y)\subword\sigma(x)$ for all $y\in\var(\eta)$.

	\proofsubparagraph{Conjunctions} For $\phi = (\phi_1 \land \phi_2)$ with $\phi_1,\phi_2\in\FC$, we first construct  $\psi_1,\psi_2\in\logC$ through recursion and combine these to $\psi\df (\psi_1\land\psi_2)$. 
	
	\proofsubparagraph{Disjunctions} For $\phi = (\phi_1 \lor \phi_2)$ with  $\phi_1,\phi_2\in\FC$, we also first construct  $\psi_1,\psi_2\in\logC$ through recursion.
	But now we cannot just apply a disjunction, as $\fvar(\phi_1)\neq \fvar(\phi_2)$ might hold, which means that a substitution that satisfies one of the subformulas can leave free variables that only occur in the other ``unguarded''. We address this through the following construction:
	\begin{equation*}
		\psi\df 
		\Bigl(\psi_1\land \bigland_{x\in \fvar(\phi_1)\setdiff\fvar(\phi_2)} \ex{p,s}{\sv\weqeq p\conc x\conc s} \Bigr)
		\lor 
		\Bigl(\psi_2\land \bigland_{x\in \fvar(\phi_2)\setdiff\fvar(\phi_1)} \ex{p,s}{\sv\weqeq p\conc x\conc s} \Bigr)
	\end{equation*}
In other words, we guard every variable $x$ that appears only in one of the subformulas by stating that $x\subword \sv$ must hold.
 This is one of the cases where the width affects the length of $\psi$, as we have $|\psi|\in O(|\phi|\conc\formulaWidth(\phi))$.
	
	\proofsubparagraph{Negations} For $\phi=\neg\phi_1$, we first recurse on $\phi_1$ and construct  $\phi_1$. We then define   
\[	\psi\df \neg\psi\land\bigland_{x\in\fvar(\phi_1)} \ex{p,s}{\sv\weqeq p\conc x\conc s}.\]
	The right part of the formula acts as guard that expresses $x\subword\sv$ for every  $x\in\fvar(\psi_1)$. Hence, under both semantics, $\sigma\models\psi$ if and only if we have that $\sigma\models\phi_1$ does not hold although  $\sigma(x)\subword\sigma(\sv)$ holds for all $x\in\fvar(\phi_1)$.
	 This is another case where the width affects $|\psi|$.
	
	\proofsubparagraph{Quantifiers} If $\phi=\ex{x}{\hat{\phi}}$ or $\phi=\fa{x}{\phi_1}$, then  $x\neq\sv$ holds by by definition. Hence, we can obtain $\psi_1$ trough recursing on $\phi_1$. 
	If $\phi=\ex{x}{\hat{\phi}}$, we can simply define $\psi \df \ex{x}{\psi_1(x)}$.
	But if $\phi=\fa{x}{\phi_1}$, we need to take into account that the universe under $\logC$-semantics is $\Sigma^*$, which means that just applying the universal quantifier would check ``too many'' possible values of~$x$. Instead, we define 
	\[
	\psi\df \fa{x}{\bigl(\psi_1 \lor \neg \ex{p,s}{\sv\weqeq p\conc x\conc s} \bigr)}.
	\]
	That is, the quantifier considers only those $x$ with $x\subword\sv$. This is the only case where we introduce a negation (but at this case cannot occur for formulas that belong to the existential fragment, this does not affect the claim of this lemma).
\end{proof}
\subsection{Proof of Theorem~\ref{thm:alsoWithConstraints}}\label{app:thm:alsoWithConstraints}
\restateThmAlsoWithConstraints*
\begin{proof}
Expanded, the statement of the claim is (from \cref{thm:recog} and \cref{thm:mc}, respectively): 
\begin{itemize}
	\item Evaluation  is $\compclass{PSPACE}$-complete for  $\FCreg$ and $\compclass{NP}$-complete for  $\fragExPos{\FCreg}$. 
	\item  Model checking for  $\FCreg$   can be solved in 
	time 
	$O(k|\phi|n^{2k})$, for $k\df\formulaWidth(\phi)$ and 
	$n\df |\sigma(\sv)|$. 
\end{itemize}	
Regarding the first claim, the lower bounds follow immediately. For the upper bond, it is enough to observe that checking $(w,\sigma)\models\constr{x}{\alpha}$ is the same as checking whether $\sigma(x)\in\Lang(\alpha)$ holds, which can be decided in time that  polynomial in $|\sigma(x)|<|w|$ and $|\alpha|$ (for example, $O(|w||\alpha|))$ when using the Thompson algorithm. Hence, we can add this as a new base case in the proof of \cref{thm:recog} (see \cref{app:thm:recog}).

Regarding the second claim, note that we can extend the proof of \cref{thm:mc} (see \cref{app:thm:mc}) in two ways when dealing with a subformula  $\constr{x}{\alpha}$ on a word $w\in\Sigma^*$. In every case, we can keep with the ``bottom-up'' approach, and generate a list of all $u\subword w$ that satisfy $u\in\Lang(\alpha)$. For $n\df|w|$, we have $O(n^2)$ candidates, and each can be checked in time $O(n |\alpha|)$, which results in a total time of $O(n^3 |\alpha|)$ for this step. 
But there is also a special case that is likely to occur: If the constraint is used as a part of a conjunction $\phi\land\constr{x}{\alpha}$ with $x\in\fvar(\phi)$, we can build the results for $\phi$ and then use $\constr{x}{\alpha}$ to restrict them (that is, apply a semijoin). 
\end{proof}
\subsection{Proof of Lemma~\ref{lem:simple}}\label{app:lem:simple}
\restateLemSimple*
\begin{proof}
Let $\alpha$ be a simple regular expression -- that is, the star operator~$*$ can only be applied to words or to $\Sigma$, which we treat as a shorthand for $\bigcup_{a\in\Sigma} a$.	
We define $\phi^{\alpha}\in\fragExPos{\FC}$ recursively; and the only case that it is not straightforward and was not shown in Example~\ref{ex:starfree} is the star operator.
We first get the straightforward cases out of the way and define  
\begin{itemize}
	\item $\phi^{\emptyset}(x)\df (x\weqeq a x)$ for some $a\in\Sigma$,
	\item $\phi^{\emptyword}(x)\df (x\weqeq \emptyword)$,
	\item $\phi^{a}(x)\df (x\weqeq a)$ for each $a\in\Sigma$,
	\item $\phi^{(\alpha_1\cdot\alpha_2)}(x) \df \ex{x_1,x_2}{(x\weqeq x_1\conc x_2 \land \phi^{\alpha_1}(x_1)\land \phi^{\alpha_2}(x_2))}$, and
	\item $\phi^{(\alpha_1\ror\alpha_2)}(x) \df  \phi^{\alpha_1}(x)\lor \phi^{\alpha_2}(x)$.
\end{itemize}
For the star operator, simple regex-formulas allow only two choices, namely $\alpha=\Sigma^*$ or $\alpha=s^*$ for some $s\in\Sigma^*$. The first case is also straightforward; we define $\phi^{\Sigma^*}(x)\df (x\weqeq x)$. 

For the second case, we exclude the case $s=\emptyword$, for which the formula $x\weqeq\emptyword$ suffices. 
Our construction adapts the construction for the respective result for $\fragEx{\logC}$ (Theorem~5 in~\cite{kar:exp}) to $\fragExPos{\FC}$ and uses the following well-known fact from combinatorics on words: 
For every $w\in\Sigma^*$, let $\wroot{w}$  denote the \emph{root} of $w$; that is, the shortest word $r$ such that $w$ can be written as $w=r^k$ for some $k\geq 0$. For all $u,v\in\Sigma^+$, we have $uv=vu$ if and only if $\wroot{u}=\wroot{v}$ (see \eg Proposition~1.3.2 in Lothaire~\cite{lot:com}). 

This allows us to express $s^*$ in the following way: Let $p\geq 1$  be the unique value for which $s=\wroot{s}^p$ holds. We now define
\[\phi^{s^*}(x)\df (x\weqeq \emptyword) \lor (x\weqeq w) \lor \psi(x),\]
where
\[\psi(x)\df 
\begin{cases}
	\ex{y}{(x\weqeq ys \land x\weqeq sy)}
	&\text{ if $p=1$,}\\
	\ex{y,z}{\bigl(x\weqeq y^p \land z\weqeq y\cdot \wroot{s} \land z\weqeq \wroot{s}\cdot y\bigr)} &\text{ if $p\geq 2$.}
\end{cases}\]
Next, we show that $(w,\sigma)\models\phi^{s^*}$ if and only if $\sigma(x)\in s^*$ and $\sigma(x)\subword w$.

We begin with the \emph{only-if-direction}. Let $\sigma(x)=s^i$ for $i\geq 0$ and $\sigma(\sv)=w\supword\sigma(x)$. If $i=0$ or $i=1$, we have $(w,\sigma)\models  (x\weqeq \emptyword)$ or $(w,\sigma)\models  (x\weqeq s)$. 
Hence, we can assume $i\geq 2$. 
We first consider the case $p=1$.
Let $\tau\df\substsubst{\sigma}{y}{s^{i-1}}$. Then $\tau(y)\subword\tau(x)\subword \tau(\sv)$ holds by definition. Furthermore, we have $\tau\models (x\weqeq ys)$ due to
$$
\tau(x) = s^i = s^{i-1} s = \tau(y) s = \tau(ys)
$$
and $\tau\models (x\weqeq ys)$ for analogous reasons. Hence, $\sigma\models\psi$ and, thereby  $\sigma\models\phi^{s^*}$. This concludes the case $p=1$.

For the case $p\geq 2$, note that $s^i = \wroot{s}^{ip}$. We define the pattern substitution $\tau$ by $\tau(y)\df \wroot{s}^i$, $\tau(z)\df\wroot{s}^{i+1}$, and $\tau(u)=\sigma(u)$ for all other $u\in\Xi$; and claim $\sigma\models\psi$. 

First, note that  as $p\geq 2$ and $i\geq 2$, we have $i+1\leq ip$. This implies  $\tau(z)=\wroot{s}^{i+1} \subword\wroot{s}^{ip} = \tau(x)$ and, hence, $\tau(y)\subword \tau(z)\subword\tau(\sv)$. Now we have 
\[ \begin{array}{rcccl}
	\tau(x)& =&\wroot{s}^{ip}&=&\tau(y^p),\\
	\tau(z)&=& \wroot{s}^{i+1}&=& \tau( y\cdot \wroot{s} ),\\
	\tau(z)&=& \wroot{s}^{i+1}&=& \tau( \wroot{s}\cdot y ).\\
\end{array}\]
Hence, $\sigma\models\psi$, and thereby $\sigma\models\phi^{s^*}$. This concludes the case of $p\geq 2$ and this direction of the proof.

For the \emph{if-direction}, assume $\sigma\models\phi^{s^*}$. Then $\sigma(x)=\emptyword$, $\sigma(x)=s$, or $\sigma\models\psi$. There is nothing to argue in the first two cases, so assume the third holds. Again, we need to distinguish $p=1$ and $p\geq 2$.

We begin with $p=1$, and consider any $v\in\Sigma^*$ such that $\tau\models (x\weqeq ys \land x\weqeq sy)$ for $\tau\df\substsubst{\sigma}{y}{v}$. Then we have 
$
\tau(x)=\tau(ys)=\tau(sy)$ and, hence, $vs=sv$. This holds if and only if $u=\emptyword$ or, due to the fact mentioned above, $\wroot{u}=\wroot{s}$. In either case, we know that there exists some $i\geq 0$ with $v=s^i$. Hence, $\tau(x) = us = s^{i+1}$. As $\sigma(x)=\tau(x)$, we have $\sigma(x)\in s^+$.

For $p\geq 2$, consider $u,v\in\Sigma^*$ such that $\tau\models \bigl(x\weqeq y^p \land z\weqeq y\cdot \wroot{s} \land z\weqeq \wroot{s}\cdot y\bigr)$ for $\tau\df \substsubst{\sigma}{y}{u,z\mapsto v}$. Due to the last two equations, we have $u\wroot{s}=\wroot{s}u$. Again, we invoke the fact, and observe there is some $i\geq 0$ with $u=\wroot{s}^i$. Hence, $\tau(x)=\wroot{s}^{ip}=w^i$ and therefore, $\sigma(x)\in s^*$. This concludes this direction and the whole correctness proof.

Regarding the complexity of the construction, note that the length of $\phi^{\alpha}$ is linear in $|\alpha|$. We conclude that $\phi$ can be constructed in polynomial time.
\end{proof}
\subsection{Proof of Theorem~\ref{thm:FCvsSpanners}}\label{app:thm:FCvsSpanners}
\restateThmFCvsSpanners*
First, note that Freydenberger~\cite{fre:splog} introduces $\splog$ (for \emph{spanner logic}), a fragment of $\fragExPos{\logCreg}$ with $\splog\polyeq\RGXcore$, and $\splogneg$, which extends $\splog$ with negation and has $\splogneg\polyeq\RGXcored$.

We define $\splog$ and $\splogneg$ as follows, based on the definition from~\cite{fre:splog}\footnote{Technically, \cite{fre:splog} does not define $\splogneg$ as a fragment of $\logCreg$, but redefines the semantics in a way that is similar to $\FC$ (this is a particularly odd decision, as our definition acts exactly like the difference of $\RGXcored$, which $\splogneg$ replicates). The results in~\cite{fre:splog} are not affected by this distinction, and this definition works better with the present paper. Furthermore, we do not allow the use of automata in regular constraints, as we do not examine spanners that are based on automata. What we call $\splog$ is called $\splog_{\mathsf{rx}}$ in~\cite{fre:splog}.}.
\begin{definition}\label{def:splogrec}
	Fix a variable $\mv\in \Xi$, the \emph{main variable}. Then $\splog$, the set of all $\splog$-\emph{formulas},  is the subset of $\fragExPos{\logCreg}$ that is obtained from the following recursive rules.
	\begin{itemize}
		\item $(\mv\weqeq\eta_R)\in \splog$ for every $\eta_R \in((\Xi\setdiff\{\mv\})\cup\Sigma)^*$.
		\item If $\varphi_1,\varphi_2\in\splog$, then $(\varphi_1\land \varphi_2)\in\splog$.
		\item If $\varphi_1,\varphi_2\in\splog$ and $\fvar(\varphi_1)=\fvar(\varphi_2)$, then $(\varphi_1\lor \varphi_2)\in\splog$.
		\item If $\varphi\in\splog$ and $x\in \fvar(\varphi)\setdiff\{\mv\}$, then $(\exists x\colon \varphi)\in\splog$.
		\item If $\varphi\in\splog$ and $x\in \fvar(\varphi)$, then $(\varphi\land \constr{\alpha}(x))\in\splog$ for every regular expression $\alpha$.
	\end{itemize}

For $\splogneg$, the set of all $\splogneg$-\emph{formulas},  we add the following rule: 
\begin{itemize}
	\item If $\phi_1,\phi_2\in\splogneg$ and $\fvar(\varphi_1)\subseteq \fvar(\varphi_2)$, then $(\neg\phi_1\land \varphi_2)\in\splogneg$.
\end{itemize}\leavevmode
\end{definition}
There are similarities between this definition of $\splog$ (and $\splogneg$) and that of $\fragExPos{\FCreg}$ (and $\FCreg$): Both have a distinguished variable (the main variable $\mv$ and the universe variable $\sv$) that cannot be bound, and both allow only single variables on the left side of every word equation. 

In fact, when applied to a word $w$, both ensure that variables can only be mapped to factors of $w$. 
That is, for all $\phi\in\splogneg$, we have that $\sigma\models\phi$ implies that $\sigma(x)\subword \sigma(\mv)$ holds for all $x\in\fvar(\phi)$.
But as we are dealing with a fragment of $\logCreg$, the main variable $\mv$ has no special role in the semantics, which is why this property is ensured through the syntax.
This is why disjunction can only be applied to formulas with the same free variables, why regular constraints and negation have to be guarded, and why the left side of every word equation is the main variable.

Together, these restrictions make $\splogneg$-formulas quite unwieldy: Every relation has to be encoded in the main variable (a property that the formulas share with spanner representations). 
This frequently requires additional quantifiers, which makes parameters like the width of a formula or the number of variables far less useful.
The criterion that only $\mv$ appears on the left side of variables is relaxed in~\cite{fre:splog} by using syntactic sugar, but this makes the effect on criteria like width even worse.

In the following, we assume that no $\splogneg$-formula uses $\sv$ and that no $\FCreg$-formula uses $\mv$.
First, note that $\splogneg$-formulas can be directly interpreted as $\FCreg$-formulas.
\begin{lemma}\label{lem:splogAsFCreg}
For every $\phi\in\splogneg$, let $\phi_{\sv}$ denote the $\FCreg$-formula that is obtained by replacing all occurrences of $\mv$ in $\phi$ with $\sv$. Then for every $\sigma$ with $\sigma(\sv)=\sigma(\mv)$, we have $\sigma\models\phi$ if and only if $\sigma\models\phi_{\sv}$.
\end{lemma}
This follows directly from the definition of the syntax of $\splogneg$ and the fact that $\sigma(x)\subword\sigma(\sv)$ holds for every $\sigma$ with $\sigma\models\phi$.
As $\splog$-formulas use neither negation nor universal quantifiers, they can be directly interpreted as $\fragExPos{\FCreg}$-formulas.
The other direction is less straightforward; but we already did most of the work in the proof of \cref{lem:FCtoC}.
\begin{lemma}\label{lem:FCregToSplog}
Given $\phi\in\FCreg$, we can compute in polynomial time $\psi\in\splogneg$ such that for all $\sigma$ with $\sigma(\sv)=\sigma(\mv)$, we have $\sigma\models\phi$ if and only if  $\sigma\models\psi$.
If $\phi\in\fragExPos{\FCreg}$, then $\psi\in\splog$.
\end{lemma}
\begin{proof}
This requires only minor modifications to the proof of \cref{lem:FCtoC} (see \cref{app:lem:FCtoC}). We need to account for two differences. Firstly, $\splogneg$ does not allow the use of universal quantifiers. Secondly, that proof does not mention regular constraints. 

We address these in two steps: First, we replace every subformula $\fa{x}{\phi'}$ with the equivalent formula $\neg\ex{x}{\neg \phi'}$. 
We then convert the formula as in the proof of \cref{lem:FCtoC}, replacing every occurrence of $\sv$ with $\mv$, and every regular constraint $\constr{x}{\alpha}$ with the $\splog$-formula $\ex{p,s}{\mv\weqeq p\conc x \conc s} \land \constr{x}{\alpha}$.

By following the construction in~\cref{app:lem:FCtoC}, one can see that the result of this process is the desired $\splogneg$-formula.
Moreover, if $\phi$ is existential-positive, then the result is a $\splog$-formula. Clearly, the whole process is possible in polynomial time.
\end{proof}
As shown in~\cite{fre:splog}, we have $\splogneg\polyeq\RGXcored$ and $\splog\polyeq\RGXcore$ (Theorem~8.4 and Theorem~4.9, respectively).
We can now use \cref{lem:splogAsFCreg} and  \cref{lem:FCregToSplog} to adapt these conversions.  

In particular, to move from $\FCreg$ to $\RGXcored$ , we first use \cref{lem:FCregToSplog} to obtain a $\splogneg$-formula and then convert this to $\RGXcored$ as described in Section~4.2.1 of~\cite{fre:splog}.
The same route works for $\fragExPos{\FCreg})$, $\splog$, and $\RGXcore$.

To move from $\RGXcored$ to $\FCreg$, we first construct a $\splogneg$-formula that realizes the spanner representation.
The conversion of functional regex formulas to $\splog$ is described in Section~4.2.2 of~\cite{fre:splog} and is only slightly more complicated than the construction for regex patterns in \cref{app:lem:simple} in the present paper. 
The conversion of the relational operators is described in Section~4.2.5 of~\cite{fre:splog}, and it is just straightforward application of disjunction, conjunction, and existential quantifiers to express union, join, and projection. The last operator, equality selection is ``free'' in $\FCreg$ and $\splog$, as ``the same word'' is expressed by using the same variable or writing $x\weqeq y$.
By \cref{lem:splogAsFCreg}, the formula can be interpreted as an $\FCreg$-formula, which concludes the proof.
\subsection{Proof of Lemma~\ref{lem:FClanguagesFO}; and conversions between \FCtxt and \FOstrtxt}\label{app:lem:FClanguagesFO}
\restateLemFClanguagesFO*

Lemma~\ref{lem:FClanguagesFO} is actually a corollary of Lemma~\ref{lem:FOtoFC} further down, which describes a more detailed conversion from $\FC$ to $\FOstr$, and of \cref{lem:FCtoFO}, which describes the opposite direction.

In both directions, the conversions can be performed in polynomial time. Furthermore, $\FOstr$ can be extended with constraints in the same way as $\FC$.
The only difference being that a $k$-ary constraint for $\FC$ would become a $2k$-ary constraint in $\FOstr$, as the latter uses two variables to describe a factor.
In particular, one could use the resulting $\FOstrReg$ instead of $\FCreg$. As a sidenote, one could of course define an $\MSO$-version of $\FOstr$, instead of adding regular constraints to the latter. 
This article does not engage with $\MSO$ for two reasons: Firstly, to address other constraints than regular constraints, one would still need to extend $\MSO$ accordingly. But, more importantly, efficient combined complexity is an important topic for this article, and second-order quantification is much more expensive than first-order quantification.

When translating between $\FC$ and $\FOord$, we need to address the issue that the former refers to words, while the latter refers to positions in a word. Hence, the variables of the two logics do not correspond directly to each other.

This is similar to the situation for comparing $\FC$ and spanners that we encountered in Section~\ref{sec:FCandSpanners}, and we address it analogously through the notion of a formula from one logic \emph{realizing} a formula from the other. 
We start with the direction from $\FO$ to $\FC$.
\begin{definition}\label{def:FOtoFC}
Let $\phi\in\FOstr$. For every assignment $\alpha$ for $\phi$ on some structure $\struct{w}$, we define its \emph{corresponding substitution} $\sigma$ by $\sigma(x)\df w_{\spn{1,\alpha(x)}}$ for all $x\in \fvar(\phi)$. 

A formula $\psi\in\FC$ \emph{realizes} $\phi$ if $\fvar(\psi)=\fvar(\phi)$ and for all $w\in \Sigma^*$, we have $(w,\sigma)\models\psi$ if and only if $\sigma$ is the corresponding substitution of some $\alpha$ with $(w,\alpha)\models\phi$.
\end{definition}
Thus, $\psi$ represents node $i\in\{1,\dots, |w|+1\}$ through the prefix of $w$ that has length $i-1$.
\begin{example}
	$
	 \phi\df\strEqPred(\minConst,x,x,\maxConst) \land \lettPred{\ta}(x) 
	$
is realized by
	$
	\psi\df (\sv\weqeq xx) \land \ex{y}{(x\weqeq \ta y)}.
	$
\end{example}
Like for $\FC$,  the \emph{width} $\formulaWidth(\phi)$ of an $\FOstr$-formula $\phi$ is defined as the maximum number of free variables in any of its subformulas. 

Although the details require some effort, we can convert every $\FOstr$-formula into an $\FC$-formula, and this is possible with an manageable increase in the width:
\begin{lemma}[label=lem:FOtoFC,restate=restateLemFOtoFC]
	Given $\phi\in\FOstr$ with $k\df\formulaWidth(\phi)$, we can compute $\psi\in\FC$  in time $O(k|\phi|)$ that realizes $\phi$. 
	This preserves the properties existential and existential-positive.
	Furthermore, we have $\formulaWidth(\psi)=k+1$.
\end{lemma}
As the construction is a bit lengthy, the proof can be found in Section~\ref{app:lem:FOtoFC}.

The direction from $\FC$ to $\FO$ is less straightforward. We have to increase the number of variables, due to a counting argument:
On a word $w$, the number of possible assignments  can be quadratic in $|w|$ for an $\FC$-variable; but there are only  $|w|+1$ possible choices per $\FO$-variable.
Accordingly, we shall represent each variable $x$ with two variables $\fovop{x}$ and $\fovcl{x}$; and the goal is to express a substitution $\sigma$ in an assignment $\alpha$ by $\sigma(x)=w_{\spn{\alpha(\fovop{x}),\alpha(\fovcl{x})}}$. 
\begin{definition}\label{def:FCtoFO}
Let $\phi\in\FC$ and let $\sigma$ be a substitution for $\phi$. Let $\psi\in\FOstr$ with $\fvar(\psi)\df\{\fovop{x},\fovcl{x}\mid x\in\fvar(\phi)\}$.  
An assignment $\alpha$ for $\psi$ on $\struct{\sigma(\sv)}$ \emph{expresses} $\sigma$ if $\sigma(x)=\sigma(\sv)_{\spn{\alpha(\fovop{x}),\alpha(\fovcl{x})}}$ for all $x\in\fvar(\phi)$.
We say $\psi$ \emph{realizes} $\phi$ if, for all $w\in\Sigma^*$, we have
\begin{enumerate}
	\item if $(w,\alpha)\models \psi$, then $\alpha$ expresses some  $\sigma$ with $(w,\sigma)\models\phi$, and
	\item if $(w,\sigma)\models\phi$, then $(w,\alpha)\models\psi$ holds for all $\alpha$ that express~$\sigma$.
\end{enumerate}\leavevmode
\end{definition}
\begin{example}
 $\phi\df (\sv\weqeq xx)$ is realized by
$
\psi \df \ex{y}{\bigl(\strEqPred(\minConst,y,y,\maxConst) \land \strEqPred(\fovop{x},\fovcl{x},\minConst,y)  \bigr)}. 
$	
\end{example}
Analogously to Lemma~\ref{lem:FOtoFC}, we can convert from $\FC$ to $\FOstr$:
\begin{lemma}[label=lem:FCtoFO,restate=restateLemFCtoFO]
Given $\phi\in\FC$ with $k\df\formulaWidth(\phi)$, we can compute $\psi\in\FOstr$ in time $O(k|\phi|)$ that realizes $\phi$ and has $\formulaWidth(\psi)=2k+3$.
This preserves the properties existential and existential-positive.	
\end{lemma}
Again. the construction is a bit lengthy, and the proof can be found in Section~\ref{app:lem:FCtoFO}.

\subsubsection{Proof of Lemma~\ref{lem:FOtoFC}}\label{app:lem:FOtoFC}
\restateLemFOtoFC*
\begin{proof}
\newcommand{\psig}[1]{(#1\pref \sv)}
We use  $\psig{x}$ as shorthand for the formula $\ex{z}{(\sv\weqeq x z)}$. This formula is frequently used as a guard to ensure that  our construction has the ``prefix invariant'', by which we mean that $\sigma\models\psi$ implies $\sigma(x)\pref\sigma(\sv)$ for all constructed $\psi$ and all $x\in\fvar(\psi)$. Usually, we do not point this out. The reader can safely assume that every occurrence of $\psig{x}$ serves this purpose. Note the use of $\pref$ can increase  the width of the formula by 1; we discuss this in each case.
The main part of the proof is a structural induction along the definition of $\FOstr$. 

\proofsubparagraph{Base cases}
 We begin the construction with the base cases; the length of the constructed formula is discussed at the end of the whole construction.

\begin{itemize}
\item $x\foeq y$ where neither $x$ nor $y$ is $\minConst$ or $\maxConst$ is realized by 
\[(x\weqeq y)\land\psig{x}\land\psig{y}.\]
Simply using $x\weqeq y$ is not enough, as we need to ensure the ``prefix invariant''. This can increase the width of the formula by $1$.
If either of $x$ or $y$ is a constant, we simply replace any occurrence of $\minConst$ with $\emptyword$ and of $\maxConst$ with $\sv$.

\item $x < y$ where neither $x$ nor $y$ is a constant is realized by 
\[\psig{y}\land\biglor_{a\in\Sigma} \ex{z}{(y\weqeq x\conc a \conc z)}.\]
We do not need to include $\psig{x}$, as this is already implicitly ensured by the equations $y\weqeq x\conc a \conc z$ in the disjunction. The new variable $z$ increases the width by one (and we can also use this $z$ for $\pref$).

Now for the constants: If $y=\maxConst$, we consider three cases for $x$. If $x=\maxConst$, the formula is not satisfiable, and we realize it with the contradiction $(\sv\weqeq a)\land (\sv\weqeq aa)$ for some $a\in \Sigma$. If $x= \minConst$, the formula is realized by $\ex{z}{\biglor_{a\in\Sigma} (\sv\weqeq a z)}$. If $x$ is a variable,
we construct the formula as in the general case and replace $y$ with $\sv$.

If $y=\minConst$, the formula is a not satisfiable, and we realize it a contradiction (see above).

Finally, if $y$ is a variable, we only need to consider $x=\maxConst$ and $x=\minConst$. In the first case, we have a contradiction (see above); the second is realized by $\ex{z}{\biglor_{a\in\Sigma} (x\weqeq a z)}$. Neither of the constructions increases the width by more than one.

\item $\lettPred{a}(x)$ is realized by $\ex{z}{(\sv \weqeq x\conc a\conc z)}$ if $x$ is a variable,  $\lettPred{a}(\minConst)$ is realized by $\ex{z}{(\sv \weqeq a\conc z)}$, and  $\lettPred{a}(\maxConst)$ is realized by $\ex{z}{(\sv \weqeq z \conc a)}$. In each case, the width is increased by one.

\item $\nextPred(x,y)$ for variables $x$ and $y$ is realized by 
\[\psig{y}\land \biglor_{a\in\Sigma} (y\weqeq x\conc a).\]
If $x=\maxConst$ or $y=\minConst$, any contradiction realizes $\nextPred(x,y)$. Moreover, $\nextPred(\minConst,\maxConst)$ is realized by $\biglor_{a\in\Sigma}\sv\weqeq a$. 
Finally, for variables $x$ or $y$, we realize $\nextPred(x,\maxConst)$ $\biglor_{a\in\Sigma}\sv\weqeq xa$ and $\nextPred(\minConst,y)$ by $(y\pref \sv)\land \biglor_{a\in\Sigma}y\weqeq a$. Neither of these constructions increases the width by more than one.

\item $\strEqPred(x_1,y_1,x_2,y_2)$ is realized by 
\[\psig{y_1} \land \psig{y_2}\land \ex{z}{(y_1\weqeq x_1 z \land y_2\weqeq x_2 z)}\]
if all four parameters are variables. For constants, we adapt the construction as follows: If $y_i=\minConst$, we replace $y_i\weqeq x_i z$ with $(x_i\weqeq \emptyword)\land (z\weqeq\emptyword)$ and omit $\psig{y_i}$. If $x_i=\minConst$, we replace $x_i$ in the constructed formulas with $\emptyword$ (removing the tautology $\emptyword\weqeq\emptyword$ if  it is created by a combination of this and the previous case occurring together). Every $x_i=\maxConst$ or $y_i=\maxConst$ is replaced with $\sv$. 
Again, all cases increase the width by at most one.
\end{itemize}
\proofsubparagraph{Recursive steps}
For the recursive steps, let $\phi,\phi_1,\phi_2\in\FOstr$ be formulas that are realized by $\psi,\psi_1,\psi_2\in\FC$, respectively. 
\begin{itemize}
\item $\phi_1\land\phi_2$ is realized by $\psi_1\land \psi_2$.
\item $\ex{x}\phi$ is realized by $\ex{x}\psi$. 
\item $\fa{x}\phi$ is realized by \[\fa{x}(\lnot\psig{x}\lor \psi),\] 
which expresses $\fa{x}(\psig{x}\rightarrow \psi)$. This guard is necessary, as the $\FC$-quantifier ranges over all factors of $\sigma(\sv)$, but only prefixes of $\sigma(\sv)$ are relevant for the $\FO$-quantifier.

In cases where we prefer using a second additional variable over introducing a negation, we could instead use the formula
\[\fa{x}{\Bigl(\psi\lor \biglor_{a\in\Sigma}\biglor_{b\in\Sigma\setdiff \{a\}} \ex{z_1}{\bigl(\ex{z_2}{\sv\weqeq z_1\conc a\conc z_2} \land \ex{z_2}{\sv\weqeq z_1\conc b\conc z_2}\bigr)  }\Bigr)}.\]
\item $\lnot\phi$ is realized by $\lnot\psi\land\bigland_{x\in \fvar(\phi)} x\pref \sv$.
\item $\phi_1\lor\phi_2$ is realized by 
\[\bigl(\psi_1\land\bigland_{x\in \fvar(\phi_2)\setdiff\fvar(\phi_1)} x\pref \sv\bigr)
\lor \bigl(\psi_2\land\bigland_{x\in \fvar(\phi_1)\setdiff\fvar(\phi_2)} x\pref \sv\bigr).\]
\end{itemize}

\proofsubparagraph{Complexity}
Regarding the length of the constructed formula, note that the formulas for $<$ and $\nextPred$ depend on $\Sigma$. But as we assume $\Sigma$ to be fixed, this is only a constant factor. 

The only formulas that is affected by the width are the constraints and the disjunction: this leads to a factor of $k$ and brings the length of the final formula to $k|\phi|$. If no disjunctions occur,  this factor is not needed,  and we get a length of $O(|\phi|)$.

As all steps are straightforward, we can construct $\psi$ in time $O(|\psi|)$. 
\end{proof}
\subsubsection{Proof of Lemma~\ref{lem:FCtoFO}}\label{app:lem:FCtoFO}
\restateLemFCtoFO*
\begin{proof}
We show this with a structural induction along the definition of $\FC$.  Recall that it is our goal to represent each $\FC$-variable $x$ through the two $\FO$-variables $\fovop{x}$ and $\fovcl{x}$. We shall construct $\psi$ in such a way that $ \alpha(\fovop{x})\leq\alpha(\fovcl{x})$ holds for all assignments $\alpha$ that satisfy $\psi$.

As we shall see in the case for word equations, the total number of variables can be lowered to $2|\fvar(\phi)|+2$ if all word equations in $\phi$ have $\sv$ on their left side. Our constructions use  $x\leq y$ as shorthand for $x<y \lor x\foeq y$.

\proofsubparagraph{Word equations} Assume that $\phi= (x_L\weqeq \eta_R)$, with $x_L\in\Xi$ and $\eta_R\in(\Xi\cup\Sigma)^*$. We first handle a few special cases before proceeding to the main construction for word equations.

\proofsubparagraph{Word equations, special cases}
 We first handle the rather straightforward case of $\eta_R=\emptyword$. Here, we distinguish two cases, namely $x_L=\sv$ and $x_L\neq \sv$. The first means that we are dealing with the equation $\sv \weqeq \emptyword$. 
This is true if and only if $\struct{w}$ contains only a single node. We express this with
\[\psi\df (\minConst\foeq\maxConst).\]
For $x_L\neq\sv$, we can directly define
\[\psi\df  (\fovop{x_L}\foeq\fovcl{x_L}).\]
Recall that the spans of empty words in some word $w$ are exactly the spans $\spn{j,j}$ with $1\leq j \leq |w|+1$. Now for the more interesting case of $\eta_R\neq\emptyword$.  Here, we need to take care of one more special case; namely, that $\sv$ appears in $\eta_R$.
If $\eta_R$ contains one or more occurrences of $\sv$, we distinguish the following sub-cases:
\begin{enumerate}
	\item $\eta_R$ contains at least one terminal,
	\item $\eta_R$ contains no terminals.
\end{enumerate}
In the first case, we can conclude that there is no $\sigma$ with $\sigma\models\phi$. This is for the following reason: Assume $\sigma(x_L)=\sigma(\eta_R)$. This implies $|\sigma(x_L)|=|\sigma(\eta_R)|$. By definition, we also have $\sigma(x_L)\subword\sigma(\sv)$ and hence $|\sigma(x_L)|\leq |\sigma(\sv)|$. As $\eta_R$ contains $\sv$ and at least one terminal (which is constant under $\sigma$), we have $|\sigma(\eta_R)|\geq|\sigma(\sv)|+1$. Thus, $|\sigma(\eta_R)|>|\sigma(\sv)|\geq |\sigma(x_L)|$. Contradiction. As $\phi$ is not satisfiable, we choose the unsatisfiable formula
\[\psi\df\ex{x}{(\lettPred{a}(x)\land (x\foeq \maxConst)  )}.\]
Recall that we assume that we defined the node $|w|+1$ in $\struct{w}$ to be letter-less, which also ensure that this formula is indeed unsatisfiable. This allows us to construct an unsatisfiable $\fragExPos{\FOstr}$-formula that also works on $\struct{\emptyword}$ and does not assume that $|\Sigma|\geq 2$. 

In the second case, we know that $\eta_R\in \Xi^+$ and that it contains $\sv$ at least once.
If $\eta_R$ contains $\sv$ twice, then $\sigma\models\phi$ can only hold if $\sigma(x)=\emptyword$ holds for all $x\in \var(\eta_R)\cup\{x_L\}$. 
This is due to a straightforward length argument: If $\sigma\models\phi$, then $\sigma(x_L)=\sigma(\eta_R)$ and $\sigma(x_L)\subword \sigma(\sv)$. The first part implies $|\sigma(x_L)|=|\sigma(\eta_R)|$. 
As $\sv$ appears at least twice in $\eta_R$, we have $|\sigma(\eta_R)|\geq 2 |\sigma(\sv)|$.
Putting this together gives 
\[|\sigma(\sv)|\geq |\sigma(x_L)|\geq |\sigma(\eta_R)|\geq2|\sigma(\sv)|,\]
which implies $|\sigma(\sv)|=0$. This proves the claim. In this case, we define
\[\psi\df  \ex{x}{\bigl((\minConst\foeq\maxConst) \land \bigland_{y\in\fvar(\phi) }\fovop{y}\foeq\fovcl{y}\bigr)}.\]
The big conjunction only serves to ensure that $\psi$ has the correct free variables; as there are no other possible assignments in $\struct{\emptyword}$,  we do not need to make the equality explicit.

Hence, we can safely assume that $\eta_R\in\Xi^+$ and that it contains $\sv$  exactly once.
Again we distinguish two cases, namely $|\eta_R|=1$ and $|\eta_R|\geq 2$.

 If $|\eta_R|=1$, we have $\phi=(x_L\weqeq \sv)$. If $x_L=\sv$, we are dealing with the trivial formula $\sv\weqeq\sv$, and can just define 
\[\psi\df \ex{x}{(x\weqeq x)},\]
or some other trivially satisfiable formula. If $x_L\neq \sv$, we define
\[\psi\df (\fovop{x_L}\foeq \minConst )\land (\fovcl{x_L}\foeq\maxConst)\]
to express this equality. It is convenient not to use $\strEqPred$ here, as $x_L$ must encompass the whole structure.

Now for $|\eta_R|\geq 2$, where $\eta_R$ contains exactly one occurrence of $\sv$. If $x_L=\sv$, we can see from a straightforward length argument that $\sigma\models\phi$ if and only if $\sigma(y)=\emptyword$ for all $y\in \var(\eta_R)\setdiff\{\sv\}$. We express this with the formula
\[\psi\df \bigland_{y\in \var(\eta_R)\setdiff\{\sv\}} \fovop{y}\foeq\fovcl{y}.\]
If $x_L\neq \sv$, we also need to ensure that $\sigma(x_L)=\sigma(\sv)$ holds, as we have $\sigma(x_L)\subword\sigma(\sv)$ by definition and $\sigma(x_L)\supword\sigma(\sv)$ from the fact that $\sv$ occurs in $\eta_R$. We define
\[\psi\df (\fovop{x_L}\foeq \minConst) \land (\fovcl{x_L}\foeq\maxConst)\land \bigland_{y\in \var(\eta_R)\setdiff\{\sv\}} \fovop{y}\foeq\fovcl{y}.\]
This also takes care of the case where $x_L$ occurs in $\eta_R$. Then, we must have $\sigma(x_L)=\emptyword$ in addition to $\sigma(x_L)=\sigma(\sv)$. 

\proofsubparagraph*{Word equations, main construction}
After covering these special cases, we can proceed to the main part of the construction. 
Let $\eta_R=\eta_1 \cdots \eta_n$, $n\geq 1$, with $\eta_i\in (\Xi\cup \Sigma)^+$ and $\eta_i\neq\sv$ for all $1\leq i\leq n$. Note that $x_L=\sv$ might hold. 

We shall first discuss how to construct an $\FOstr$-formula with $n+1$ variables in addition to the $2|\fvar(\phi)|$ free variables from $\{\fovop{x},\fovcl{x}\mid x\in\var(\eta_R)\}$ that are required by definition. After that, we shall describe how to reduce this to three additional variables (by reordering quantifiers and re-using variables, as  commonly done for $\FO$ with a bounded number of variables, analogous to \cref{example:patCompress}).

These $n+2$ additional variables are the variables $y_1,\ldots,y_{n+1}$. The idea behind the construction is that each pair $(y_i,y_{i+1})$ shall represent the part of $\sigma(\eta_R)$ that is created by~$\eta_i$. 
If we do not want to keep the number of variables low, we define 
\begin{equation*}
\hat{\psi}\df\ex{y_1,\ldots,y_{n+1}}{}\begin{cases}
(y_1\foeq \minConst)\land \bigland_{i=1}^n \psi_i(y_i,y_{i+1})  \land (y_{n+1}\foeq\maxConst) & \text{if $x_L=\sv$,}\\
\strEqPred(\fovop{x_L},\fovcl{x_L},y_1,y_{n+1})\land \bigland_{i=1}^n \psi_i(y_i,y_{i+1}) & \text{if $x_L\neq\sv$,}
\end{cases} 
\end{equation*}
where the formulas $\psi_i$ are defined as follows for all $1\leq i \leq n$:
\begin{equation*}
\psi_i(y_i,y_{i+1}) \df \begin{cases}
\lettPred{a}(y_i) \land \nextPred(y_i,y_{i+1}) & \text{if $\eta_i=a\in \Sigma$},\\
\strEqPred(\fovop{x},\fovcl{x},y_i,y_{i+1})& \text{if $\eta_i=x\in Xi$.}
\end{cases}
\end{equation*} 
Although $\hat{\psi}$ is directly obtained from the pattern $\eta_R$, some explanations are warranted. Firstly, note that $\sv$ only plays a role if we have $x_L=\sv$. In this case, the use of $\minConst$ and $\maxConst$ ensures that $\eta_R$ encompasses all of $\sv$.

Moreover, observe that the construction ensures that  $\fvar(\hat{\psi})$ is the set of all $\fovop{x}$ and $\fovcl{x}$ such that $x\in\fvar(\hat{\phi})$. If $x_L=\sv$, then $\fvar(\hat{\phi})=\var(\eta_R)$, and the variables $\fovop{x}$ and $\fovcl{x}$ are ``introduced'' in the $\psi_i$ where $\eta_i=x$ holds. But if $x_L\neq \sv$ and $x_L\notin\var(\eta_R)$, then  $\strEqPred(\fovop{x_L},\fovcl{x_L},y_1,y_{n+1})$ not only ensures that $x_L$ and $\eta_R$ are mapped to the same word, but also that  $\fvar(\hat{\psi})$ contains $\fovop{x_L}$ and $\fovcl{x_L}$.

Finally, we observe that the construction does not need to specify that $y_i \leq y_{i+1}$ or $\fovop{x}\leq \fovcl{x}$ holds. By definition, $\nextPred$ and $\strEqPred$ guarantee this property and can act as guards.

Keeping this in mind, one can now prove by induction that for every $w\in\Sigma^*$, we have $(w,\alpha)\models \hat{\psi}$ if and only if $\alpha$ expresses some pattern substitution $\sigma$ with $(w,\sigma)\models\phi$. In other words, $\hat{\psi}$ realizes $\phi$. 
All that remains is to reduce the number of variables through a standard re-ordering and renaming process. 

 We first discuss the case of $x_L=\sv$, where we need only two variables.
Observe that for $2\leq i \leq n$, the variable $y_i$ is only used in the sub-formulas $\psi_{i-1}$ and $\psi_{i+1}$. Similarly, $y_1$ is only used in $\psi_1$ and in $(y_1\foeq \minConst)$, and $y_{n+1}$ is only used in $\psi_n$ and $(y_{n+1}\foeq \maxConst)$.
This allows us to use shift the quantifiers into the conjunction, which leads to  the following formula:
\begin{gather*}
\psi'\df \ex{y_1,y_2}{}\bigl((y_1\foeq\minConst)\land \psi_1(y_1,y_2)\\
\hphantom{\psi'\df \ex{}{}} \land \ex{y_3}{\bigl(\psi_2(y_2,y_3)} \\
\hphantom{\psi'\df \ex{y_1}{}}\land \ex{y_4}{\bigl(\psi_3(y_3,y_4)} \\
\hphantom{\psi'\df \ex{y_1,y_2,y_3}{}}\vdots\\
\hphantom{\psi'\df \ex{y_1,y_2}{}}\land \ex{y_{n+1}}{\bigl(\psi_n(y_n,y_{n+1})\land (y_{n+1}\foeq\maxConst)\bigr)}\cdots
\bigr)\bigr)\bigr),
\end{gather*}
for which  $\psi'\equiv\hat{\psi}$ holds. As observed above, each variable $y_i$ is only used together with $y_{i+1}$. Accordingly, we now obtain $\psi$ from $\psi'$ by replacing every variable $y_i$ where $i$ is odd with $z_1$, and every $y_i$ where $i$ is even with $z_2$. Then $\psi$ has only $|\fvar(\phi)|+2$ variables. More over, $\psi\equiv\psi'\equiv \hat{\psi}$ holds; and as we already established that $\hat{\psi}$ realizes $\phi$, we conclude that $\psi$ realizes $\phi$.

For the case of $x_L\neq \sv$, observe that $\hat{\psi}$ contains $\strEqPred(\fovop{x_L},\fovcl{x_L},y_1,y_{n+1})$. Hence, we cannot move the quantifier for $y_{n+1}$ to the ``bottom'' of the formula. Instead, we define
\begin{gather*}
\psi'\df \ex{y_1,y_2,y_{n+1}}{}\Bigl(\strEqPred(\fovop{x_L},\fovcl{x_L},y_1,y_{n+1})\land \psi_1(y_1,y_2)\\
\hphantom{\psi'\df\ex{y_{n+1}}{}} \land \ex{y_3}{\bigl(\psi_2(y_2,y_3)}\\ 
\hphantom{\psi'\df\ex{y_{n+1},y_2}{}} \land \ex{y_4}{\bigl(\psi_3(y_3,y_4)} \\ 
\hphantom{\psi'\df\ex{y_{n+1},y_2,y_1y_1}{}}\vdots\\
\hphantom{\psi'\df\ex{y_{n+1},y_2,y_1}{}}\land \ex{y_{n}}{\bigl(\psi_n(y_n,y_{n+1})\bigr)}\cdots
\bigr)\bigr)\Bigr).
\end{gather*}
We now obtain $\psi$ by renaming the $y_i$ with $1\leq i\leq n$ as in the previous case. Hence, the only difference is that $y_{n+1}$ remains unchanged, which leads to a total of $2|\fvar(\phi)|+3$ variables.

\proofsubparagraph{Conjunctions} If $\phi=\phi_1\land\phi_2$, we define
$\psi\df \psi_1\land\psi_2$, where $\psi_1$ and $\psi_2$ realize $\phi_1$ and $\phi_2$, respectively. The correctness of this construction follows directly from the induction assumption and Definition~\ref{def:FCtoFO}.

\proofsubparagraph{Disjunctions}  If $\phi=\phi_1\lor\phi_2$, we first construct the $\FOstr$-formulas $\psi_1$ and $\psi_2$ that realize $\phi_1$ and $\phi_2$, respectively. 
We cannot just define $\psi$ as $\psi_1\lor\psi_2$.  
Unless $\fvar(\phi_1)=\fvar(\phi_2)$ holds, this definition would accept assignments that do not realize any pattern substitution. 
For example, if we have $(w,\alpha)\models\psi_1$ with $\alpha(\fovop{x})>\alpha(\fovcl{x})$ for some variable $x\in\fvar(\phi_2)\setdiff\fvar(\phi_1)$, then $(w,\alpha)\models\psi_1\lor\psi_2$ holds.  

We address this problem by guarding variables that are only free in exactly one formula, and define 
\begin{equation*}
\psi\df \Bigl(\psi_1\land \bigland_{x\in\fvar(\phi_2)\setdiff\fvar(\phi_1)} \bigl(\fovop{x}\leq\fovcl{x}) \bigr)\Bigr)
\lor \Bigl(\psi_2\land \bigland_{x\in\fvar(\phi_1)\setdiff\fvar(\phi_2)} \bigl(\fovop{x}\leq\fovcl{x}) \bigr)\Bigr).
\end{equation*}
For all $w\in\Sigma^*$, we now have $(w,\alpha)\models\psi$ if and only if, firstly, $(w,\alpha)\models\psi_i$ for an $i\in\{1,2\}$  and, secondly,  $\alpha(\fovop{x})\leq \alpha(\fovcl{x})$ for all $x\in\fvar(\phi)$. This  holds if and only if $\alpha$ expresses some $\sigma$ with $(w,\sigma)\models\phi_i$.

\proofsubparagraph{Negations} If $\phi=\neg\hat{\phi}$, we first construct $\hat{\psi}$ that realizes $\hat{\phi}$, and then define 
\[\psi\df \lnot\hat{\psi}\land \Bigl(\bigland_{x\in\fvar(\hat{\phi})} \fovop{x}\leq\fovcl{x} \Bigr).\]
We face an issue that is analogous to the one for disjunction; defining $\lnot\hat{\psi}$ would lead to a formula that accepts assignments that do not express a pattern substitution. Again, the solution is guarding the free variables of $\hat{\phi}$.

\proofsubparagraph{Existential quantifiers} If $\phi=\ex{x}{\hat{\phi}}$, construct a formula $\hat{\psi}$ that realizes $\hat{\phi}$, and define  
$\psi\df \ex{\fovop{x},\fovcl{x}}{\hat{\psi}}$. 
As $\fovop{x}\leq \fovcl{x}$ is guaranteed as an induction invariant, we do not need to guard the two variables.

\proofsubparagraph{Universal quantifiers} If $\phi=\fa{x}{\hat{\phi}}$, construct $\hat{\psi}$ that realizes $\hat{\phi}$, and define 
\[\psi\df \fa{\fovop{x},\fovcl{x}}{\bigl((\fovop{x}>\fovcl{x})\lor\hat{\psi}\bigr)},\]
which amounts to defining $\fa{\fovop{x},\fovcl{x}}{\bigl((\fovop{x}\leq\fovcl{x})\rightarrow\hat{\psi}\bigr)}$. 
Again, we need to deal with the induction invariant: 
If we simply defined $\fa{\fovop{x},\fovcl{x}}{\hat{\psi}}$, then the formula would be invalid on all non-empty $w$.  

\proofsubparagraph{Complexity considerations} The special cases for word equations can be checked in time $O(|\phi|)$ and create formulas of constant length. The main construction for word equations creates a formula of length $O(|\eta_R|)$ and takes proportional time.

The only recursive cases that create formulas of a length more than linear are negation and disjunctions. Here, the guards increase the formula length to $O(k|\phi|)$, which dominate the final formula length and the total running time.
Hence, if $\phi$ contains neither  negations nor disjunctions, we have $|\psi|\in O(|\phi|)$; and the same holds for the run time.
\end{proof}
\subsection{Proof of Lemma~\ref{lem:anbn}}\label{app:lem:anbn}
\restateLemAnBn*

\newcommand{\NEstrEqPred}{\mathsf{NE}\strEqPred}
\newcommand{\FOne}{\FO[\NEstrEqPred]}
\newcommand{\newstruct}[2]{\tilde{\mathcal{A}}_{#1}^{\mathtt{#2}}}
\newcommand{\newstructgen}[2]{\tilde{\mathcal{A}}_{#1}^{#2}}
We start with some preliminaries; the actual proof is in Section~\ref{sec:lem:anbn:actualProof}
\subsubsection{Decomposing the structure}
In the proof, we use $\sqcup$ to denote the union of two disjoint sets.
The key part of the argument is the following formulation of the Feferman-Vaught theorem:
\begin{fvtheorem}[Theorem~1.6 in~\cite{DBLP:journals/apal/Makowsky04}]
	For every $q \in \mathbb{N} $ and for every first order formula $\phi$ of quantifier rank $q$ over a finite vocabulary, one can compute effectively a reduction sequence 
\[	{\psi}_1^A,\ldots, {\psi}_{k}^A, {\psi}^B_1, \ldots, {\psi}^B_{k}\]
	of first order formulas 
	 the same vocabulary and a Boolean function $B_{\phi}:\{0,1 \}^{2k} \rightarrow \{0,1\}$ such that 
\[	\mathcal{A} \sqcup  \mathcal{B} \models \phi \]
	if and only if
	$B_{\phi}(	b_1^A,\ldots, b_{k}^A, b^B_1, \ldots, b^B_{k}) = 1$ where 
	$b_j^A = 1 $ if and only if  $\mathcal{A} \models \psi_j^A$ and
	$b_j^B = 1 $ if and only if  $\mathcal{B} \models \psi_j^B$.
\end{fvtheorem}
This proof uses $\FOstr$-formulas instead of $\FC$-formulas (due to \cref{lem:FClanguagesFO}).
Intuitively, we show that any formula $\phi$ and structure $\struct{w}$ for some word $w\in \ta^* \tb^*$ can be translated into a formula $\psi$ that operates on the  union of two disjoint structures $\newstruct{w}{a}$ and $\newstruct{w}{b}$ such that $\struct{A}\models\phi$ if and only if $\newstruct{w}{a} \sqcup  \newstruct{w}{b}$ satisfies $\psi$.
We then apply the Feferman-Vaught theorem on $\psi$ and obtain some kind of separation of it.
Finally, we use the pigeonhole principle
to compose a word that it is outside of the language.

Formally, let $\phi^\prime \in \FOstr$.
Recall that $\phi^\prime$ is evaluated on the structures $\struct{w}$ for $w\in\Sigma^*$, with the universe $\{1,\dots,|w|+1\}$, where the node $|w|+1$ is not marked with any letter.
Also recall that the vocabulary of $\FOstr$ contains the two unary letter predicates $\lettPred{\ta}$ and $\lettPred{\tb}$, the binary relations $<$ and $\nextPred$, the 4-ary relation $\strEqPred$, and the constant symbols $\minConst$ and~$\maxConst$.

To apply the Feferman-Vaught theorem, we need to split $\struct{w}$ into two disjoint structures, and to rewrite $\phi$ into a suitable formula $\psi$.
In this case, \quot{suitable} means that $\phi$ and $\psi$ are equivalent on words of the form $w = \ta^m \tb^n$ with $m,n \geq 1$.
On these words, $\struct{w}$ contains nodes $1,\dots,m$ that are marked $\ta$, nodes $m+1,\dots,m+n$ that are marked $\tb$, and the unmarked node $m+n+1$.
Our goal is to split all non-unary relations in $\struct{w}$ into a structure $\newstruct{w}{a}$ for the  $\ta$-part and structure $\newstruct{w}{b}$  for the $\tb$-part. 
The only technical issue that we need to deal with is that the $\strEqPred$-relation contains tuples $(i_1,j_1,i_2,j_2)$ with $w_{\spn{i_1,j_1}}=w_{\spn{i_2,j_2}}$.
In these tuples, $j_1$ is the first position that is \emph{not} in $w_{\spn{i_1,j_1}}$, and likewise for $j_2$.
While this relation is more convenient when expressing spanners and converting from and to $\FC$-formulas, it creates issues when splitting the structures. 
Hence, we first define a 4-ary relation $\NEstrEqPred$ that contains those $(i_1,j_1,i_2,j_2)$ with $i_1\leq j_1$ and $i_2\leq j_2$ that have $w_{i_1}\cdots w_{j_1} = w_{i_2}\cdots w_{j_2}$. 
Hence, unlike $\strEqPred$, this relation only describes the equality of non-empty words.

We now can directly split the universe into $\{1,\dots,m\}$ and $\{m+1, \dots, m+n+1\}$, and define
\begin{itemize}
	\item $\newstruct{w}{a}$ over the universe $\{1,\ldots, m \}$,
	\item $\newstruct{w}{b}$ over the universe $\{m+1,\ldots, m+n+1 \}$,
\end{itemize}
where each structure $\newstructgen{w}{c}$ with $c\in\{\ta,\tb\}$ has the relation $\lettPred{c}$ and the relations $<_c$, $\nextPred_c$, and $\NEstrEqPred_c$, which are restrictions of the corresponding relations in $\struct{w}$ or of $\NEstrEqPred$ to the universe of $\newstructgen{w}{c}$.
Likewise, we use  $\minConst_c$ and $\maxConst_c$, where $\minConst_\ta$ and $\maxConst_\ta$ refer to $1$ and $m$, and $\minConst_\tb$ and $\maxConst_\tb$ refer to $m+1$ and $m+n+1$, respectively.
We use $\FOne$ to denote the set of all formulas over this modified vocabulary and  observe the following:
\begin{lemma}\label{lem:anbnrew}
	For every  $\phi \in \FOstr$,
	we can construct $\psi\in\FOne$ 
	such that for every $w = \ta^m \tb^n$ with $m,n \ge 1$, we have
	$\struct{w} \models \phi$ if and only if $\newstruct{w}{a} \sqcup  \newstruct{w}{b} \models \psi$.
\end{lemma}
\begin{proof}
We obtain $\psi$ from $\phi$ by rewriting all parts that use the constant symbols $\minConst$ and~$\maxConst$ or any of the relation symbols $<$, $\nextPred$, or $\strEqPred$.
As we are only interested in words of the form $\ta^m\tb^n$ with $m,n\geq 1$, we can replace $\minConst$ with $\minConst_\ta$ and $\maxConst$ with $\maxConst_\tb$. 
We now replace the relations as follows:
\begin{itemize}
	\item Every occurrence of $x<y$ is replaced with the formula
	\[(x\mathbin{<_{\ta}}y )\lor( x\mathbin{<_{\tb}}y) \lor \Big( \big(  x\mathbin{<_{\ta}}\maxConst_\ta \lor x\foeq \maxConst_\ta \big) \land \big(  \minConst_\tb\mathbin{<_{\tb}}y \lor  \minConst_\tb\foeq y \big)\Big),\]
	which covers the cases that both variables are in the $\ta$-part, both are in the $\tb$-part, or the remaining case that $x$ is in the $\ta$-part and $y$ in the $\tb$-part.
	\item Every occurrence of $\nextPred(x,y)$ is replaced with
		\[
		\nextPred_\ta(x,y) \lor \nextPred_\tb(x,y) \lor \big(x\foeq\maxConst_\ta\land y\foeq\minConst_\tb\big), 
		\]
	where first two cases have both variables in the same part (as above), and the last describes that $x$ is the last $\ta$ and $y$ the first $\tb$.
	\item To simplify the explanation of the last case, we describe it in two steps. Every $\strEqPred(x_1,y_1,x_2,y_2)$ is first replaced with
\begin{multline*}
	\big(x_1\foeq y_1 \land x_2\foeq y_2\big)
	 \lor \exists z_1,z_2\colon \big( \nextPred(z_1,y_1)\land \nextPred(z_2,y_2) \land \NEstrEqPred_\ta(x_1,z_1,x_2,z_2)    \big)\\
	 \lor \exists z_1,z_2\colon \big( \nextPred_\tb(z_1,y_1)\land \nextPred_\tb(z_2,y_2) \land \NEstrEqPred_\tb(x_1,z_1,x_2,z_2)    \big)
\end{multline*}
The first conjunct describes the case where we have two occurrences of the empty word (which can be in any part of $w$). In all other cases, equal words must both be in the $\ta$-part or the $\tb$-part, which means that they are covered by the respective $\NEstrEqPred_c$. As these relations bound words with their last position (unlike $\strEqPred$), we use the $z_i$ to obtain these positions.
After this, we replace each of the two $\nextPred(z_i,y_i)$
with 
\[
\nextPred_\ta(z_i,y_i) \lor \big(z_i\foeq\maxConst_\ta\land y_i\foeq\minConst_\tb\big),
\]
to account for cases where $x_i$ is in the $\ta$-part and $y_i$ in the $\tb$-part.
\end{itemize}
Cases where these subformulas involve constants are handled accordingly.
On words from $\ta^+\tb^+$, each of the new subformulas acts exactly like the one it replaces. 
Hence, for every $w = \ta^m \tb^n$ with $m,n \ge 1$, we have
$\struct{w} \models \phi$ if and only if $\newstruct{w}{a} \sqcup  \newstruct{w}{b} \models \psi$.
\end{proof}
\subsubsection{Actual proof of Lemma~\ref{lem:anbn}}\label{sec:lem:anbn:actualProof}
\begin{proof}
Assume that there is some $\phi \in \FOstr$ with  $\Lang(\phi)=\{\ta^n \tb^n  \mid n\ge 1\}$
(by  \cref{lem:FClanguagesFO}, this is the same as assuming that this is an $\FC$-language). 
We apply Lemma~\ref{lem:anbnrew} 
to $\phi$ and obtain a formula $\psi\in\FOne$  
such that for all $w = \ta^m \tb^n$,  we have
$\newstruct{w}{a} \sqcup  \newstruct{w}{b} \models \psi$ if and only if $m=n$.
We now invoke the Feferman-Vaught theorem for $\psi$ and obtain a sequence of first order formulas
\[{\psi}_1^A,\ldots, {\psi}_{k}^A, {\psi}^B_1, \ldots, {\psi}^B_{k}\]
and a Boolean function $B_{\phi}:\{0,1\}^{2k} \rightarrow \{ 0,1\}$ such that for any word $w\in \{\ta^m \tb^n \mid m,n\geq 1 \}$, we have
$\newstruct{w}{a} \sqcup   \newstruct{w}{b} \models \psi$ if and only if 
$B_{\phi }({b}_1^A,\ldots, {b}_{k}^A,{b}_1^B,\ldots, {b}_{k}^B)=1$; where
${b}_j^A = 1$ if and only if  $\newstruct{w}{a} \models {\psi}_j^A$ and 
${b}_j^B = 1$ if and only if  $\newstruct{w}{b} \models {\psi}_j^B$.

Since $\{\ta^n \tb^n  \mid n\ge 1\}$ contains infinitely many words, 
we can use the pigeonhole principle to conclude that there exists  $m \ne n$
such that for $w_m\df \ta^m\tb^m$ and $w_n\df\ta^n\tb^n$, 
we have for every $j$
\begin{itemize}
	\item $\tilde{\mathcal{A}}_{w_m,\ta} \models {\psi}_j^A$ if and only if  $\tilde{\mathcal{A}}_{w_n,\ta} \models {\psi}_j^A$, and
\item $\tilde{\mathcal{A}}_{w_m,\tb} \models {\psi}_j^B$ if and only if  $\tilde{\mathcal{A}}_{w_n,\tb} \models {\psi}_j^B$.
\end{itemize}
In other words, $w_m$ and $w_n$ produce the same vectors of $2k$ bits, which means that the Boolean function $B$ has the same result.
Therefore, we can conclude that $\tilde{\mathcal{A}}_{w_m,\ta} \sqcup  \tilde{\mathcal{A}}_{w_n,\tb} \models \psi$.
Together with
Lemma~\ref{lem:anbnrew}, 
this gives us $\ta^m \tb^n \models \phi$, which is  a contradiction.
\end{proof}
\subsection{Proof of Theorem~\ref{thm:equalLength}}\label{app:thm:equalLength}
\restateThmEqualLength*
\begin{proof}
Recall that we assume $|\Sigma|\geq 2$. 
We first prove that we can extend Lemma~\ref{lem:anbn} to show that there is no 
$\FCreg$-formula that defines  the language $L_{el}\df \{\ta^n\tb^n\mid n\geq 1\}$.
Assume or the sake of a contradiction that there is a sentence $\phi\in\FCreg$ such that $\Lang(\phi)=L_{el}$. 
Our goal is now to prove that there exists a sentence $\psi\in\FC$ such that $\Lang(\phi)=\Lang(\psi)$.
 
To construct $\psi$, we first obtain $\psi'\in\FCreg$ by replacing every constraint $\constr{x}{\alpha}$ in $\phi$ with a constraint for the language $\Lang(\alpha)\cap \ta^*\tb^*$. This is possible, as each of these languages $\Lang(\alpha)\cap \ta^*\tb^*$ is regular, due to the fact that the class of regular languages is closed under intersection.

Then  for all $w\in\ta^*\tb^*$, we have $w\in \Lang(\psi')$ if and only if $w\in\Lang(\phi)$; as on these words, all variables in $\phi$ can only be mapped to elements of $\ta^*\tb^*$ as well. 
Next, we use Lemma~6.1 from~\cite{fre:splog}, which states (for core spanners) that every bounded regular language is an $\fragExPos{\FC}$-language; where a language $L$ is bounded if it is subset of a language $w_1^*\cdots w_k^*$ with $k\geq 1$ and $w_1,\ldots,w_k\in\Sigma^*$. 
Clearly, $\ta^*\tb^*$ is bounded, which means that all constraints in $\psi'$ use bounded regular languages. 

Thus,  we can obtain $\psi\in\FC$ from $\psi'$ by replacing every constraint in $\psi'$ with an equivalent $\fragExPos{\FC}$-formula.
Then we have $\psi\equiv \psi'$, which gives us $w\models\psi$ if and only if $w\models\phi$  for all $w\in\ta^*\tb^*$.
We conclude that
\[
\phi\equiv \constr{\sv}{\ta^*\tb^*}\land\psi.
\]
Now $\psi$ is an $\FC$-formula; and as  $\ta^*\tb^*$ is a simple regular expression, we can rewrite it into an equivalent  $\FC$-formula, using  Lemma~\ref{lem:simple}. 
Hence, $\phi$ is equivalent to an $\FC$-formula; which means that $\Lang(\phi)=L_{el}$ is an $\FC$-language.
This contradicts Lemma~\ref{lem:anbn}; hence, $\phi$ cannot exist.

Now assume that there the equal length relation is expressible in $\FCreg$, that is, assume there is some $\phi_{el}(x,y)$ such that $\sigma\models\phi_{el}$ if and only if $|\sigma(x)|=|\sigma(y)|$.
Then we have $\Lang(\phi_{el})=L_{el}$ for
\[
\phi_{el}\df \sv\weqeq xy \land \phi_{el}(x,y) \land \constr{x}{\ta^*}\land \constr{y}{\tb^*}.
\]
This contradicts the previous paragraph, which means that the equal length relation is not expressible in $\FCreg$.
\end{proof}
By Theorem~\ref{thm:FCvsSpanners}, this inexpressibility also translates to $\RGXcored$.
While Fagin et al.~\cite{fag:spa} showed that this relation cannot be expressed in $\RGXcore$ (see~\cite{fre:splog} for a simpler proof), this is the first inexpressibility result for generalized core spanners.

\end{document}